\RequirePackage{fix-cm}
\documentclass[11pt,a4paper]{article}
\usepackage[margin=2cm]{geometry}
\usepackage{xparse}
\newsavebox{\fminipagebox}
\NewDocumentEnvironment{fminipage}{m O{\fboxsep}}
 {\par\kern#2\noindent\begin{lrbox}{\fminipagebox}
  \begin{minipage}{#1}\ignorespaces}
 {\end{minipage}\end{lrbox}%
  \makebox[#1]{%
    \kern\dimexpr-\fboxsep-\fboxrule\relax
    \fbox{\usebox{\fminipagebox}}%
    \kern\dimexpr-\fboxsep-\fboxrule\relax
  }\par\kern#2
 }
\usepackage{paralist}
\usepackage{tikz}
\usepackage{amsmath}
\usepackage{amssymb}
\usepackage{amsthm}
\usepackage{listings}
\usepackage[shortlabels]{enumitem}
\usepackage{stmaryrd}
\usepackage{hyperref}
\usepackage{booktabs} 
\usepackage{mathpartir}
\usepackage{tikz}
\usetikzlibrary{positioning}
\usepackage{subcaption}
\usepackage{ifpdf}
\newcommand{\dom}{\mathrm{dom}}
\ifpdf
  \DeclareGraphicsExtensions{.pdf,.jpg,.png}
\else
  \DeclareGraphicsExtensions{.eps,.png}
\fi
\graphicspath{{.}{fig/}}

\usepackage{xspace}
\usepackage{soul}
\usepackage[numbers]{natbib}

\newcommand{\sparsenftwo}{NF$^2$\xspace}
\newcommand{\partitioning}{Partitioning\xspace}
\newcommand{\banana}[1]{\llparenthesis #1 \rrparenthesis}

\newcommand{\Spec}{\mathit{Spec}}
\newcommand{\Conc}{\mathit{Conc}}
\newcommand{\Spain}{\mathit{Spain}}
\newcommand{\tri}{\triangleright}
\newcommand{\BB}{\mathbb{B}}
\newcommand{\RR}{\mathbb{R}}

\newcommand{\sysname}{\textsc{Eris}\xspace}
\newcommand{\symbolic}{\circledS\xspace}

\newtheorem{theorem}{Theorem}[section]
\newtheorem{lemma}{Lemma}[section]
\theoremstyle{definition}
\newtheorem{example}{Example}[section]
\newtheorem{definition}{Definition}[section]

\newcommand\xqed{%
  \leavevmode\unskip\penalty9999 \hbox{}\nobreak\hfill
  \quad\hbox{$\Diamond$}}

\begin{document}

\title{\sysname: Efficiently measuRing dIscord in multidimensional Sources}

\author{Alberto Abell{\'o} \and
        James Cheney
        }



\maketitle

\begin{abstract}
Data integration is a classical problem in databases, typically decomposed into schema matching, entity matching and data fusion. To solve the latter, it is mostly assumed that ground truth can be determined.
However, in general, the data gathering processes in the different sources are imperfect and cannot provide an accurate merging of values.
Thus, in the absence of ways to determine ground truth, it is important to at least quantify how far from being internally consistent a dataset is.
Hence, we propose definitions of \emph{concordant} data and define a \emph{discordance} metric as a way of measuring disagreement to improve decision making based on trustworthiness.

We define the \emph{discord measurement} problem of numerical attributes in which given a set of uncertain raw observations or aggregate results (such as case/hospitalization/death data relevant to COVID-19) and information on the alignment of different conceptualizations of the same reality (e.g., granularities or units), we wish to assess whether the different sources are concordant, or if not, use the discordance metric to quantify how discordant they are. We also define a set of algebraic operators to describe the alignments of different data sources with correctness guarantees, together with two alternative relational database implementations that reduce the problem to linear or quadratic programming. These are evaluated against both COVID-19 and synthetic data, and our experimental results show that discordance measurement can be performed efficiently in realistic situations.
\end{abstract}

\section{Introduction}\label{sec:intro}

The focus of this work is decision making environments performing complex OLAP-like multidimensional queries \cite{abello.encyclopedia} that extensively use numerical aggregation and involve multiple data sources requiring integration.
Data integration traditionally has three steps \cite{DBLP:journals/jiis/CanalleSL21}: (1) Schema matching and alignment, which overcomes semantic and structural heterogeneity between attributes and entities from different sources; (2) Entity matching (a.k.a. Record linkage), which detects records that correspond to the same real-world entity, and (3) Data fusion (a.k.a. Record merging), which aims to identify the correct one among conflicting values. Thus,
Data fusion refers to the combination of data from different, heterogeneous sources in order to provide a more precise understanding of reality than offered by those sources separately \cite{DBLP:journals/inffus/GutierrezRPLS22}. The quality of its result is clearly affected by the disparity of the involved sources.
For instance, in data warehousing environments, consistent and well known data from different sources go through a well structured cleaning and integration process.
In the wild, sources are typically incomplete and not well aligned, and such data cleaning and integration processes are far from trivial, resulting in imperfect comparisons.
Like in the parable of the blind men describing an elephant after touching different parts of its body (i.e., touching the trunk, it is like a thick snake; the leg, like a tree stump; the ear, like a sheath of leather; the tail tip, like a furry mouse; etc.), in many areas like epidemiology, social sensing or information extraction, different data sources reflect the same reality in slightly different and partial ways, and there is not any ground truth available, requiring truth discovery processes \cite{DBLP:journals/sigkdd/LiGMLSZFH15}.
For example, during the COVID-19 pandemic, it was problematic to have reliable information on number of cases and deaths, since many different actors were independently gathering data that had later to be integrated to make them globally meaningful. Evaluating the reliability of each source was crucial for decision making, and we develop \sysname a tool that facilitates doing this.

In such a complex context, where estimating the source reliability and inferring true information is necessary, we require a tool to measure discrepancies for available data.
Typically, consistency is used to measure to which degree a dataset is free of contradictions \cite{DBLP:journals/inffus/GutierrezRPLS22}, which is done most of the time by simply counting differences in the sources \cite{CaseStudy}, or maybe something more elaborate like using the Shapley value to weight primary key violations as in \cite{DBLP:journals/lmcs/LivshitsK22}, or based on the number of necessary repairs like \cite{DBLP:conf/lpnmr/Bertossi19}. A more complete overview and classification of these kinds  of metrics can be found in \cite{DBLP:conf/ecai/ParisiG20}.
Nevertheless, we contend that better
measures than counting exist for many scenarios in the case of coincidence of numerical attributes, whose distance can be precisely quantified.
Hence, in this case, a binary metric aiming at full consistency does not look realistic and we require a more precise one showing how far values are from each other.

\paragraph{Problem.}
Thus, having in mind that we often analyse data by placing different indicators/features/measures (e.g., number of patients or deaths) in a multidimensional space (e.g., geography and time), the problem we approach in this work is the measurement of numerical disagreement between different  sources. To solve this problem we have to face several difficulties: (a) sources need to firstly go through a difficult and error prone alignment to make them comparable, (b) there can be many alternative metrics, and (c) given the amount of data in today's scenarios, these must be computed efficiently.
Hence, we aim at defining a (a) \textbf{declarative}, (b) \textbf{flexible} and (c) \textbf{scalable} method to quantify discrepancies in the different numerical attributes.

Relational Database Management System (RDBMS) are scalable and provide a declarative query language, together with a flexible mechanism to deal with uncertain and incomplete information by using NULL values.
However, they do not provide any guarantees on alignments and it is well known that NULLs are overloaded with different meanings such as unknown, nonexisting or no-information \cite{abiteboul.foundations}.

\paragraph{Approach.}
First, we restrict the use of NULL only for nonexisting or no-information, and propose to enrich the data model with \emph{symbolic variables} that allow to represent the partial knowledge we might have about uncertain numerical values, and integrate this in an RDBMS whose query results are processed in a solver to generate the desired metric.
Our approach is (a) declarative in that it provides a high-level language (based on standard relational query languages) for expressing the intended alignments among sources.  This high-level language can be translated down to linear or quadratic programming problems that can be solved efficiently.  The translation is proved correct and users do not need to carry it out themselves or be concerned with the low-level details of the encoding.  It is (b) flexible because users can firstly use different distance metrics, and also make changes to alignments (e.g., to accommodate changes to source data formats) using the high-level language rather than manually changing a low-level system of equations.  It is (c) scalable in that despite the NP-completeness of general quadratic programming problems, our approach can find optimal solutions measuring the discord (i.e., distance away from being consistent) in a dataset quickly using off-the-shelf solvers.
This is so, because we only allow linear expressions in the characterization of uncertainty of values and use convex functions in the discord measurement.

\paragraph{Contributions.}
In this paper, we consider problems we call \emph{concordance checking} and \emph{discordance measurement}.  Concordance is the problem of determining whether disparate data sources we wish to integrate are consistent with each other according to some specification (of how they should be related).  Discordance measurement is the problem of determining how close or distant the observed data are from being concordant.
We define a flexible setting, that can be instantiated with different distance metrics (see \cite{DBLP:journals/corr/abs-2208-01644} for alternatives), for the evaluation of the trustworthiness of different sources of multidimensional data based on their concordancy/discordancy using standard linear or quadratic programming solvers.
Moreover, since besides errors and conflicts in data, different conceptualizations are also a problem \cite{DBLP:journals/csur/MountantonakisT19}, we define an algebra that allows to easily describe alignments between sources, and guarantees the correctness of their symbolic evaluation.
While using symbolic variables for NULLs is not a new idea, introduced for example in classical models for incomplete information such as c-tables and v-tables~\cite{imielinski84jacm} and used more recently in data cleaning systems such as LLUNATIC~\cite{geerts.llunatic}, our approach generalizes unknowns to be arbitrary (linear) expressions.

To our knowledge, there is not any system that can automatically generate the measurement of discordance, and even less in the presence of semantic heterogeneities between the sources.
More concretely, in this paper, we contribute:

\begin{enumerate}
    \item A definition and formalization of the problem of \textbf{discord measurement} of databases under some merging processes, independently of the concrete distance metric being used.
    \item A set of algebraic operations to describe high-level \textbf{alignment specifications} that allow to describe merging processes of multidimensional tables with symbolic numerical expressions.
    \item An automatic translation from such specifications to low-level linear or quadratic programs with accompanying proofs of  correctness.
    \item A novel \emph{coalescing} operator that automatically generates \textbf{concordancy constraints} over symbolic tables, that can be efficiently checked with off-the-shelf software, as exemplified in our prototype.
    \item A prototype, \sysname,\footnote{Eris is the Greek goddess of discord.} and an experimental comparison of two alternative implementations using an RDBMS (PostgreSQL) and a  quadratic programming solver (OSQP~\cite{osqp}). For the sake of prototyping, we assume data resides in a single DBMS, but the same approach applies to virtual data integration given the corresponding wrappers.
    \item An analysis of discordances in the epidemiological surveillance systems of six European countries during the COVID-19 pandemic based on EuroStats and Johns Hopkins University (JHU) data.
\end{enumerate}
\paragraph{Organization.}
Section~\ref{sec:example} presents a motivational example that helps to identify the problem formally defined in Section~\ref{sec:problem}, whose solution based on an algebraic query language for symbolic evaluation is presented in Section~\ref{sec:solution}. Section~\ref{sec:implementation} details two alternative relational implementations, which are evaluated in Section~\ref{sec:experiments}. Our experimental results show that both approaches provide acceptable performance, and illustrate the value of our approach for assessing the discordance of COVID-19 epidemiological surveillance data at different times and countries between March 2020 and February 2021.
The paper concludes with related work and conclusions in Sections~\ref{sec:relatedwork} and~\ref{sec:conclusions}.  A preliminary short paper (presenting the problem statement informally via examples and excluding the main technical content in sections \ref{sec:problem}--\ref{sec:experiments}) has been published in \cite{AlbertoJames.DOLAP}.

\section{Motivating example}\label{sec:example}
We used COVID-19 data in our experiments and examples, which are widely available and of varying quality, making them a good candidate for discordance evaluation.
We consider that a network of actors (i.e., governmental institutions) take \emph{primary} measurements of COVID-19 cases and \emph{derive} some aggregates from those.
In an idealized setting, we would expect to know all the relationships and have consistent measurements for each primary attribute, and each derived result would be computed exactly with no error.
However, some relationships are unknown and both primary and derived attributes are noisy, biased, unknown or imperfect.
We illustrate now how to model it using database schemas and views, and describe the different problems we need to solve in this scenario.\footnote{We assume some familiarity with relational model, queries, views, SQL, etc.~\cite{abiteboul.foundations} as well as with the multidimensional data model~\cite{abello.encyclopedia}.}

\begin{figure}
\begin{center}
\includegraphics[width=10cm, clip]{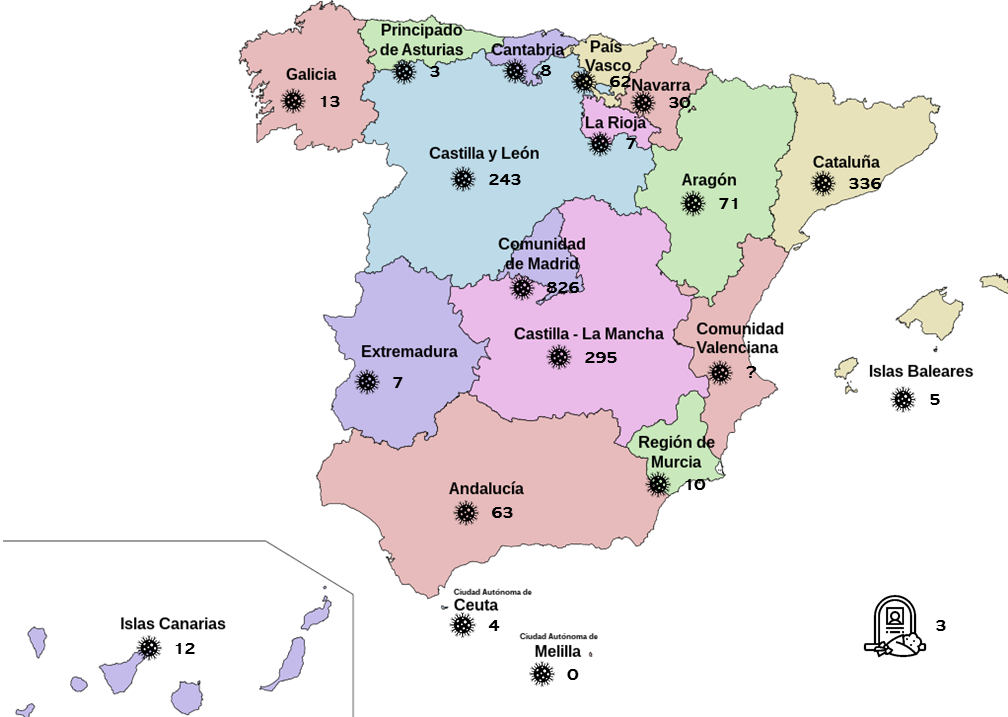}
\end{center}
\caption{\label{fig:PanemMap} $\Spain$'s map $[$source: Wikipedia$]$ with epidemiological information reported on week 2020W24 (total of 2,142 cases and three deaths in the whole country)}
\end{figure}
\begin{example}\label{label:exampleIntro}
$\Spain$, as depicted in Fig.~\ref{fig:PanemMap}, comprises nineteen regions ($R_\text{I},$ $\ldots ,R_{\text{XIX}}$).
In each region, there are several hospitals and a person living in $R_i$ is monitored by at most one hospital.
Hospitals report their number of cases of COVID-19 to their regional governments, and each regional government reports to the Ministry of Health (MoH).

Given their management autonomy, the different regions in $\Spain$ use different and imperfect monitoring mechanisms and report separately the COVID-19 cases they detect every week.
Suppose that despite being gathered daily at health facilities, $\Spain$ is only reporting weekly to the European Centre for Disease Prevention and Control (CDC) partial information at the region level and the overall information of the country.
We can model this using relational tables with the weekly region and country information, and try to use SQL to \textbf{measure discord} between them.

{\small
\begin{lstlisting}[language=sql,escapechar=\%]
ReportedRegion(%\underline{region, week}%, cases)
ReportedCountry(%\underline{week}%, cases)
\end{lstlisting}
}
\vspace{-0.5cm}
\xqed\end{example}

The first thing that must be done before measuring discrepancies is to overcome semantic and schematic heterogeneity. Thus, in terms of SQL, we can align the schemas through named queries (a.k.a. views).

\begin{example}\label{example.view}
Before making any measurement, we need to align the two sources by describing the \textbf{merging process}.
In this case, the following view aggregates the regional data for each week, which ought to coincide with the values per country:

{\small\normalfont
\begin{lstlisting}[language=sql,escapechar=\%]
CREATE VIEW AggReported AS
SELECT week, SUM(cases) AS cases
FROM ReportedRegion GROUP BY week;
\end{lstlisting}
}
\vspace{-0.5cm}
\xqed\end{example}

Once we know that quantities in the attributes are using the same units, scales, etc., and assuming that we already have properly identified the different entities, we can simply count coincidences in the attribute values.

\begin{example}\label{example.assertions}
Ideally, if all COVID-19 cases were detected and properly reported, the week should unambiguously determine the number of cases (i.e., information derived from reported cases, both at region and country levels, must coincide). In terms of SQL, as in~\cite{CaseStudy}, this could be checked with the assertion of a \textbf{concordancy constraint} in the form of a simple query like the following.

{\small\normalfont
\begin{lstlisting}[language=sql,escapechar=\%]
SELECT count(*)
FROM ReportedCountry
NATURAL JOIN AggReported;
\end{lstlisting}
}
\vspace{-0.5cm}
\xqed\end{example}

Nevertheless, as explained above, achieving exact consistency seems unlikely in any real setting.
Pure coincidence (or even string similarity metrics such as Levenshtein distance~\cite{DBLP:journals/fcsc/YuLDF16}) does not give an idea of the magnitude of discrepancies in numeric data.
For example, in the case of European countries, and according to the real data used in the experiments in Section~\ref{section.CaseStudy}, we can see that the reported cases at country and region level only coincide for one country (DE) in week 24.
If we use Levenshtein distance with a threshold of 20\% of the overall number of digits, we get three more coincidences for UK and one more for IT (still none for ES, NL and SE).
Thus, using existing techniques (i.e., assertions) it is possible to check full consistency (i.e., value coincidence) among data sources when there is no uncertainty, but it is not straightforward to quantify to which extent the various data sources are consistent with the expected relationships, in the presence of unknown values or suspected errors in reporting.

\begin{example}\label{example.data}
We can see that the following database is not consistent with the view specification above, in part because the cases of one of the nineteen Spanish regions (i.e., $R_{\text{XIX}}$, which in reality corresponds to \emph{Comunidad Valenciana}) are not declared, but also because the sum of cases per region (1,995) do not add up to the overall amount in the country (2,142).
Thus, it is not enough to say that the database is inconsistent, but we can see that there is a discrepancy of $2142-1995=147$ cases.

{\small\normalfont
\begin{lstlisting}[language=sql,escapechar=\%]
ReportedRegion("I","2020W24",13)
%$\dots$%
ReportedRegion("XVIII","2020W24",12)
ReportedCountry("2020W24",2142)
AggReported("2020W24",1995)
\end{lstlisting}
}

First of all, it is important to realize that the simplistic approach of using NULL just worsens the problem, because replacing any value with it would only result in a loss of information. Instead, we can assign an error factor $\epsilon_i$ to every value $v_i$ in the database, and measure the average of squared difference from each number of weekly cases $v_i$ to the midpoint $m$ (a.k.a. average) so that $v_i\cdot(1+\epsilon_i)= m$ with the following query. According to \cite{DBLP:journals/corr/abs-2208-01644}, one of the most common goodness of fit measures is least squares error.

{\small\normalfont
\begin{lstlisting}[language=sql,escapechar=\%]
SELECT week,
  (((r.cases+a.cases)/(2.0*r.cases)-1)^2
  +((r.cases+a.cases)/(2.0*a.cases)-1)^2)/2
FROM ReportedCountry r JOIN AggReported a
  ON r.week=a.week;
\end{lstlisting}
}
It is important to notice, firstly the dependence of the query on the distance being used, and also how its complexity grows with the number of variables and their corresponding alignments.
\xqed\end{example}

\begin{figure}
    \begin{center}
    \includegraphics[width=10cm clip]{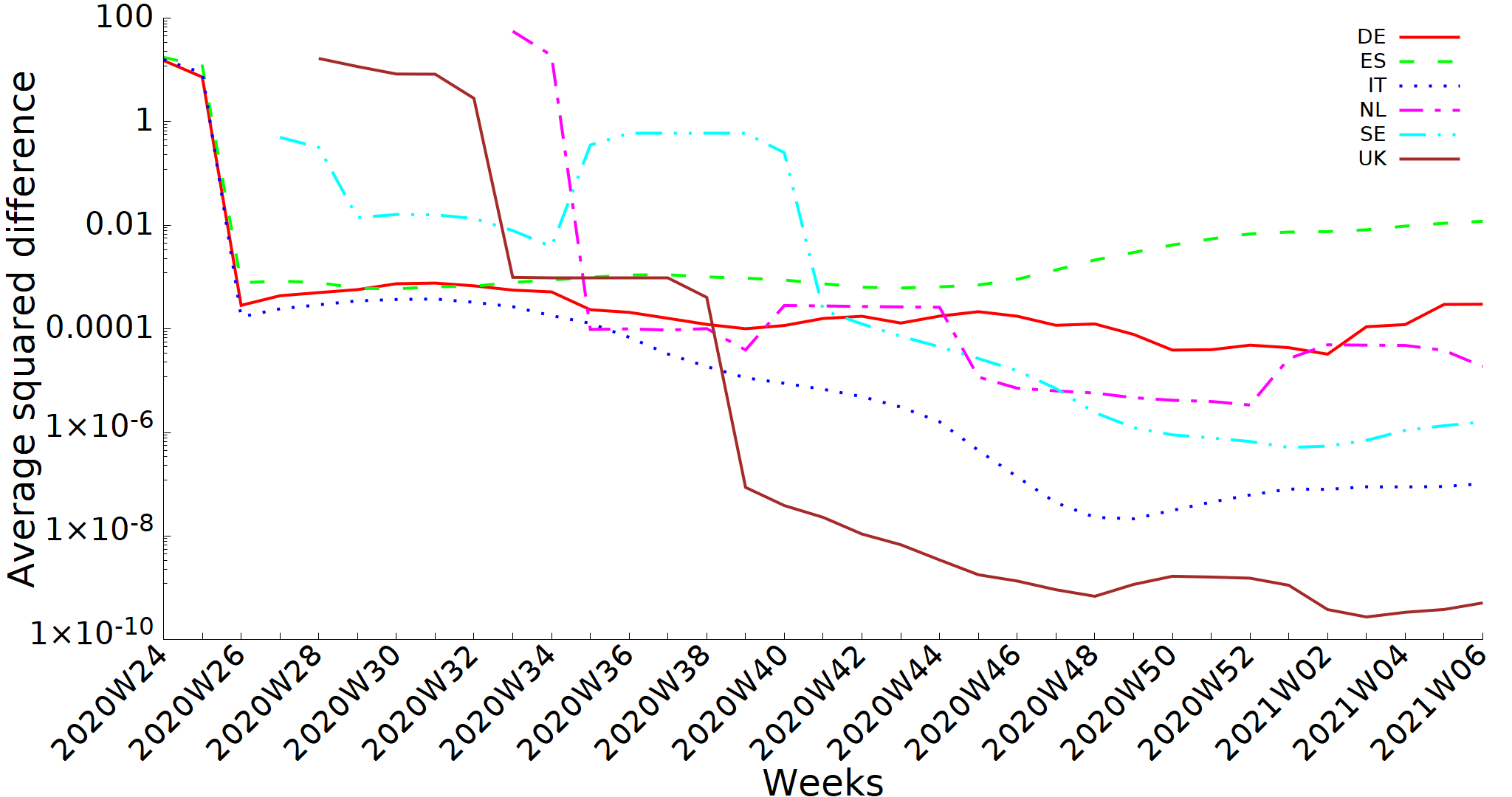}
    \end{center}
    \caption{Average squared error per week (against average)}\label{fig:CasesErrorSQL}
\end{figure}

\begin{example}\label{example.CasesErrorSQL}
We now consider the same $\Spain$ scenario, but using the real data of six European countries discussed in detail in Section~\ref{section.CaseStudy}. Taking as reference value simply the average of the two reported, we apply a similar query, obtain the value for all 39 weeks in our case study, and get the line chart in Fig.~\ref{fig:CasesErrorSQL}, which shows the average discordance of each dataset along time in a running average of five weeks.
We need to use logarithmic scale because the distance (measured as the average squared difference) between values can vary from one week to another in several orders of magnitude.
\xqed\end{example}
Even though, as seen in the previous example, we can go beyond simply counting discrepancies (a.k.a. voting) with only SQL, we contend that we require some specific mechanism to properly and flexibly quantify discordancy between different sources. From the query generation point of view, it is easy to realize that having more than two sources, or simply considering potential error variables in all tuples (e.g., those nineteen actually reported at region level) substantially complicates the SQL code. Indeed, if we think of manually generating the formula for any potential alignment (overcoming any semantic heterogeneity) between multiple sources, it is clear that it is not only error prone, but simply unfeasible.\footnote{Section~\ref{section.CaseStudy} presents, in the same COVID-19 case study, a more realistic and complex alignment of 35 algebraic operations on multiple sources (Fig.~\ref{fig:COVIDquery}) showing all the potential of our discordance quantification technique. For this case, \sysname automatically generates, with all the correctness guarantees, two SQL queries of 145 and 225 operators in the PostgreSQL access plan, and a Python file of 8\,538 characters.}
Moreover, the mechanism should be flexible enough to facilitate changing the definition of distance (e.g., from relative to absolute values), or the function (e.g., from sum of squares to simple sum of values), or even assigning weights to the distances depending on their sources.
To our knowledge, there is not any other tool that allows to do this.

\section{Problem formulation}\label{sec:problem}

Given an \emph{idealized} scenario (specified by its schema and views) and a collection of actual observations (both primary and derived),  we can still consider two  problems:
\begin{enumerate}[(A)]
\item\label{problem:estimate}\emph{Value estimation}: Estimate the values of numerical attributes of interest (e.g., the number of COVID-19 cases across $\Spain$) that make the system consistent, a special case of Truth Discovery \cite{DBLP:journals/sigkdd/LiGMLSZFH15}.
\item\label{problem:evaluate}\emph{Discord evaluation}: Evaluate how far is the actual, discordant dataset from an idealized concordant one.
\end{enumerate}

Problem~\ref{problem:estimate} is the well-studied statistical estimation problem through numerical fusion operators \cite{DBLP:journals/tsmc/Bloch96},
so this  can be very difficult, 
especially where the interrelationships among different data sources are complex (see \cite{DBLP:journals/csur/BleiholderN08} for a survey of existing systems).  Instead, we consider problem~\ref{problem:evaluate}, which is roughly analogous to computing the distance function used in truth discovery approaches \cite{DBLP:journals/sigkdd/LiGMLSZFH15}, as an indication of source reliability.  However, unlike conventional truth discovery, which considers homogeneous datasets from different sources that might have different reliability but follow a common format, we consider situations where there are heterogeneous data sources providing complementary views of the real phenomenon, but where the available data sources have nontrivial relationships.  Traditional truth discovery seeks to identify consensus values that minimize the distance from these to the observations, whereas in our setting, we have a single, heterogeneous dataset and we want to measure how far it is away from being consistent with our expectations.  That is, we wish to measure the distance from the observed data to the nearest self-consistent dataset (of which there may be infinitely many), not just a finite number of distances between homogeneous datasets.

Given a (probably incomplete but overlapping) set of instances, we assume only a merging process specification in the form of \emph{expectations} about their alignment, expressed using database queries and views, and try to answer the following questions.
Considering on the one hand the queries and views specifying the expected behavior, and on the other the data corresponding to observations of some of the inputs, intermediate results, or (expected) outputs, is the observed numeric data complete and concordant considering the alignment specification?  If there is missing data, can the existing datasets be extended to some complete instance that is concordant?  Finally, how far from being fully consistent is the numerical data?

Consequently, we aim at extending DBMS functionalities for generic concordance evaluation as a way to quantify how far away the data are from being consistent.
Although our goal in this paper is not to find a realistic estimate of the true values of unknown or uncertain data, but instead  to quantify how close the data are to our expectations under the given alignment, we need to make some assumption on this.
As in Example~\ref{example.data} and~\ref{example.CasesErrorSQL}, taking the average of multiple points is always possible, but over-simplistic.
Thus, we contend that using the value minimizing the errors of all sources, although more complex, is more principled (e.g., in our case, it gives a more comparable measure and avoids the need of using logarithmic scale, as in Fig.~\ref{fig:CasesErrorSQL}).
It is important to clarify that while the approach produces estimates for the uncertain values as a side-effect, they may not have any statistical validity unless additional work is done to statistically characterize the sources of uncertainty, which we see as a separate problem.

\paragraph{Notation.}
We assume some familiarity with foundations of relational databases, as covered for example by textbooks like Abiteboul et al.~\cite{abiteboul.foundations}.
We use the following notational conventions for tuples and relations: a tuple $t$ is a finite map from some set of attribute names $\{A_1,\ldots,A_n\}$ to data values $d \in D$.
We use letters such as $K,U,V,W$, etc. to denote sets of attribute names, and sometimes write $U,V$ to indicate the union of disjoint attribute sets (i.e., $U \cup V$ when $U$ and $V$ are disjoint) or $U,A$ to indicate addition of new attribute $A$ to attribute set $U$ (i.e., $U \cup \{A\}$ provided $A \notin U$). We also write $U \backslash V$ for the difference of attribute sets.  Data values $D$ include real numbers $\RR$, and (as discussed below) value attributes are restricted to be real-valued.  The domain of a tuple $\dom(t)$ is the set of attribute names it maps to some value.  We write $t.A$ for the value associated with $A$ by $t$, and $t[U]$ for the tuple $t$ restricted to domain $U$.  We write $t.A := d$ to indicate the tuple obtained by mapping $A$ to $d$ and mapping all other attributes $B \in \dom(t)$ to $t.B$.  Note that $\dom(t.A := d) = \dom(t) \cup \{A\}$ and this operation is defined even if $A$ is not already mapped by $t$.  Furthermore, if $V = \{B_1,\ldots,B_n\}$ is an attribute set and $u$ is a tuple with domain $U\supseteq V$, then we write $t.V := u$ as an abbreviation for $(\cdots(t.B_1 := u.B_1).B_2:= u.B_2\cdots).B_n := u.B_n$, that is for the result of (re)defining $t$ to match $u$ on the attributes from $V$.  Finally, when the range of $t,u$ happens to be $\RR$, that is, $t$ and $u$ are real-valued vectors, we write $t + u$ for the vector sum, $\alpha \cdot t$ for scalar multiplication by $\alpha \in \RR$,  and when $Z$ is a finite multiset of such real-valued vectors, we write $\sum Z$ for their vector sum.

Relational databases generally have schemas that describe the field names and types of each relation in the database, as well as integrity constraints such as key and foreign key constraints.  Our approach assumes data adhering to a simple  multidimensional data model; specifically this means we consider the fields of each relation to be split into two sets, \emph{keys} which identify a particular row uniquely and may be either numeric or descriptive (e.g., geographical location, date) and \emph{values} that must be numeric and give quantitative indicators associated with keys.  Relations of this form are essentially partial finite maps from the keys to the values. We define schemas as follows to ensure that they have this form.

\begin{definition}[Finite Map Signature]
A finite map signature is a relational schema with a primary key annotation, which is written $K \tri V$, where $K$ and $V$ are disjoint attribute names.
A relation matches such signature if its attributes are $K \cup V$, the key-fields $K$ are elements of the data domain $D$, the value-fields $V$ are real numbers $\RR$, and it satisfies the Functional Dependency (FD) $K \to V$ (i.e., any two tuples in the relation that match on the fields in $K$ also match on the fields in $V$), that is we write $R : K \tri V$ for what is conventionally written as $R(\underline{A_1,\ldots,A_m},B_1,\ldots,B_m)$ when $K = A_1,\ldots,A_m$ and $V = B_1,\ldots,B_m$.
\end{definition}

\begin{definition}[Finite Map Schema]
A \emph{finite map schema} $\Sigma$ is an assignment of relation names $R_1,\ldots,R_n$ to \emph{finite map signatures}.  A database instance $I$ matches $\Sigma$ if each relation $R$ of $I$ matches the corresponding signature $\Sigma(R)$.  We write   $Inst(\Sigma)$ for the set of all instances of a schema $\Sigma$.
\end{definition}

\begin{example}\label{example.eschemas}
Thus, our example in page \pageref{label:exampleIntro} contains tables
\[
\begin{array}{ll}ReportedRegion:\{region,week\}\tri\{cases\} & \text{ and }\\
ReportedCountry:\{week\}\tri\{cases\}\,.
\end{array}\]
We also introduce here (for later use) a new table
\[DeclaredDeaths:\{week\}\tri\{deaths\}\]
reporting the total number of deaths across the whole country, with the following data:

{\small\normalfont
\begin{lstlisting}[language=sql,escapechar=\%]
DeclaredDeaths("2020W24",3)
\end{lstlisting}
}
\vspace{-0.5cm}
\xqed\end{example}

Given that our goal is to define and measure the degree of discordance of different data sources with complementary multidimensional information (e.g., $Repor\-tedRegion$, $ReportedCountry$, $DeclaredDeaths$) under a data fusion process, it is important to notice that discrepancies occur when they assign different values to the same key.
Consequenly, we will work in a setting where not only the source tables, but also query results (and view schemas) need to satisfy such finite map constraints, indicating that different coincident sources quantify features of the same object or event.
We introduce a specialized query language and type system to maintain these constraints while dealing with uncertainty, which arises from completely unknown/missing values, or reported measurements that have some unknown error.

Specifically, we define a high-level alignment definition language and a flexibly configurable metric (see \cite{DBLP:journals/corr/abs-2208-01644} for alternatives) that can be efficiently computed with off-the-shelf software.
Indeed, the key contribution of this paper is that both checking concordance and measuring discord can be done by augmenting the data model with \emph{symbolic expressions}, and this in turn can be done consistently and efficiently in an RDBMS with the right set of algebraic operations guaranteeing correctness.  We formalize this intuition next.

\section{Proposed solution}\label{sec:solution}
In the following, we introduce the three mechanisms that constitute our technique, and how to use them together to tackle the problem at hand.
Firstly, in Section~\ref{sec:algebra}, we define a variant of relational algebra for queries over tables that are finite maps. This guarantees that the result of any query is still a finite map.
Then, in Section~\ref{section.symbolicEvaluation}, the concept of \emph{symbolic tables} representing uncertainty is defined and the effect of each operator over them formally established and exemplified, together with their correctness proof.
Afterwards, in Section~\ref{section.fusion}, a new abstract operator (a.k.a. \emph{fusion}) is introduced to establish the behaviour in finding coincident instances in the presence of an \emph{alignment specification} of different sources, and show how it can be implemented by reduction to linear or quadratic programming.
Finally, in Section~\ref{sec:alltogether}, using the previous toolset, we formally define and exemplify the concordance and discordance problems.

\subsection{Restricted algebra}\label{sec:algebra}

We consider a variation of relational algebra over finite maps, whose type system ensures that the finite map property is preserved in any query result.
\[\begin{array}{rcl}
c &\mathrel{::=} & A = B \mid A < B \mid \neg c \mid c \wedge c'\\
e &::=& \alpha \in \RR \mid A \mid e + e' \mid e - e' \mid e \cdot e' \mid e / e'\\
q &\mathrel{::=} &R \mid \sigma_c(q) \mid \hat{\pi}_{W}(q) \mid q \Join q' \mid q \uplus_B q' \mid q \backslash q' \\
 & & \mid \rho_{B \mapsto B'}(q) \mid \varepsilon_{B:=e}(q) \mid \gamma_{K;V}(q)
\end{array}\]

Conditions $c$ and expressions $e$ are typical sublanguages containing Boolean combinations of (in)equalities among attributes, and real-valued arithmetic operations over attributes, respectively.
Queries $q$ are loosely based on the standard relational algebra with extensions for grouping/aggregation and expression evaluation.  They include several standard forms such as relation names $R$, selections $\sigma_c(-)$, projections $\hat{\pi}_W(-)$, set difference $\backslash$, renaming $\rho_{B \mapsto B'}$, and joins ($\Join$), as well as discriminated union $\uplus_B$, expression evaluation $\varepsilon_{B:=e}(-)$, and aggregation $\gamma_{K;V}(q)$.
\begin{figure}[tb]
\begin{mathpar}\small
\inferrule*{R : K \tri V \in \Sigma}
{\Sigma \vdash R : K \tri V}
\and
\inferrule*{
\Sigma \vdash q : K \tri V \\
K \vdash c : \BB
}{
\Sigma \vdash \sigma_c(q) : K \tri V}
\and
\inferrule*{
\Sigma \vdash q : K \tri V \\
W \subseteq V
}{
\Sigma \vdash \hat{\pi}_{W}(q) : K \tri V\backslash W}
\and
\inferrule*{\Sigma \vdash q_1 : K_1 \tri V_1\\
\Sigma \vdash q_2 : K_2 \tri V_2\\
V_1 \cap V_2 = \emptyset}
{\Sigma \vdash q_1 \Join q_2 : K_1 \cup K_2 \tri V_1 , V_2}
\and
\inferrule*{\Sigma \vdash q : K \tri V\\
\Sigma \vdash q' : K \tri V}
{\Sigma \vdash q \uplus_B q' : K,B \tri V}
\and
\inferrule*{\Sigma \vdash q : K \tri V}
{\Sigma \vdash \rho_{A \mapsto B}(q) : K[A \mapsto B] \tri V[A \mapsto B]}
\and
\inferrule*{\Sigma \vdash q : K \tri V \\
\Sigma \vdash q' : K \tri \emptyset}
{\Sigma \vdash q \backslash q' : K \tri V}
\and
\inferrule*{\Sigma \vdash q : K \tri V \\
K,V \vdash e : \RR}
{\Sigma \vdash \varepsilon_{A = e}(q) : K \tri V,A}
\and
\inferrule*{\Sigma \vdash q : K \tri V\\
K'\subseteq K \\ V' \subseteq V}
{\Sigma \vdash \gamma_{K';V'}(q) : K'\tri V'}
\end{mathpar}
\caption{Well-formed queries}\label{fig:wf}
\end{figure}
Fig.~\ref{fig:wf} defines the well-formedness relation for queries.  The rules in Fig.~\ref{fig:wf} define the relations $\Sigma \vdash q : K \tri V$ inductively as the least relation satisfying the rules, where each rule is interpreted as an implication ``if the hypotheses (shown above the line) hold then the conclusion (shown below the line) holds.'' For more background on type systems, inference rules and inductive reasoning about them, see a standard textbook such as Pierce~\cite{pierce02tapl}.  The judgment $\Sigma \vdash q : K \tri V$ states that in schema $\Sigma$, query $q$ is well-formed and has type $K \tri V$.  Intuitively, this means that if $q$ is run on an instance of $\Sigma$, then it produces a result relation that is a finite map $ R: K \tri V$.  We make use of additional judgments for well-formedness of selection conditions $c$ ($U \vdash c : \mathbb{B}$) and expressions $e$ ($U \vdash e : \mathbb{R}$) which are standard and omitted.
Later on, in the next section, a type system is defined that specification queries are required to adhere to, which is illustrated in Figure~\ref{fig:spec-wf}, page \pageref{fig:spec-wf}.

\begin{description}
\item[Selection $\sigma_c(R)$] behaves as in relational algebra; the selection criterion $c$ is evaluated on each row in the input table and those rows satisfying $c$ are retained in the output, while the rest are discarded.  The type system restricts selection criteria to consider only predicates over key values  that evaluate to boolean $\BB = \{\mathit{true},\mathit{false}\}$.  
\item[Projection-away $\hat{\pi}_U(R)$] projects the fields in $U$ away (that is, removes them from the rows of its argument).  The type system  only allows projection-away of value-fields, that is, requires $U \subseteq V$; this ensures the results of these operations are still finite maps.  
\item[Join $R\Join S$] takes two finite maps, whose key fields may overlap, and returns tuples formed by fusing all pairs of tuples whose common fields have the same values.  Unlike in general relational algebra, the arguments to a join may only have overlapping key-fields and not shared value-fields.  
\item[Discriminated union $R\uplus_B S$] combines two finite maps, adding a new field $B$ whose value will differentiate the tuples coming from the first or respectively second argument.  We do not allow arbitrary unions because the union of two finite maps is not a finite map if the domains overlap.  
\item[Difference $R \backslash S$] removes the keys in $S$ from $R$, where $S$ has no value fields (i.e., is just a set of keys).
\item[Renaming $\rho_{B \mapsto B'}(R)$] that changes the name of field $B$ to $B'$. 
\item[Derivation $\varepsilon_{B:=e}(R)$] performs arithmetic calculations by adding a new value-field $B$ which is initialized by evaluating expression $e$ using the field values in each row.
\item[Grouping/aggregation $\gamma_{K;V}(R)$] performs grouping on key-fields $K$ and aggregation by summing the value-fields $V$. The constraint that grouping can only be performed on key-fields and aggregation on value-fields ensures that the results are still finite maps.
\end{description}

\begin{example}\label{example.queries}
We now illustrate the query language above on the running example scenario introduced in page \pageref{label:exampleIntro}.
Thus, we can get the sum of all cases reported by region in a given week using our query language as $\gamma_{week;cases}(ReportedRegion)$.  As discussed earlier, this results in a sum of 1,995 cases (not the 2,142 one would expect).  We also expect that the number of deaths is 0.015 times the number of reported cases, which would be written as $\varepsilon_{deaths=0.015*cases}(ReportedCountry)$, but again does not coincide with the three declared deaths.  
\xqed\end{example}

Some of these restrictions help ensure that the query operations preserve the finite map property described by the typing rules.  Others, though not necessary for this purpose, are nevertheless helpful later when we generalize the queries to evaluate over symbolic values.  In particular forbidding selections or joins that involve comparing value fields will help avoid the need for some of the complexities  encountered in c-tables~\cite{imielinski84jacm} or work on provenance for aggregate queries~\cite{amsterdamer.aggregates}.  These restrictions have not posed problems in part because, especially in the multidimensional model, it is usually not necessary or desirable to perform exact comparisons on (continuous-valued) value fields when describing how different sources are aligned. Instead, we often do want to express that different sources should be close together, but we can do this by introducing symbolic variables that represent unknown errors distorting the true value, and imposing equational constraints using fusion and coalescing operators, as explained later.

To allow for better understanding of how well or badly the data conforms to our expectations, expressed using queries, we next consider symbolic evaluation of queries over tables in which some values can be variables (or more generally, expressions).

\subsection{Symbolic evaluation}\label{section.symbolicEvaluation}

The basic idea is to represent \emph{unknown} real values with variables $x \in X$.  Variables can occur multiple times in a table, or in different tables, representing the same unknown value, and more generally unknown values can be represented by (linear) expressions.
However, key values $d \in D$ used in key-fields are required to be known.  This reflects the assumption that the source database is \emph{partially closed} \cite{fan.relative}, that is, we assume the existence of master data for the keys (i.e., all potential keys are coincident and known).

\begin{definition}[Symbolic Expression]
Let $X$ be some fixed set of variables.  A \emph{symbolic expression  $e$ over $X$} is a real-valued expression in $\RR[X]$, the set of polynomials in $\RR$ with variables from $X$.  We normally consider only linear expressions (e.g., $a_0+a_1x_1+\dots+a_nx_n$).
\end{definition}

\begin{definition}[Symbolic Table]
A \emph{symbolic table}, or \emph{s-table (over $X$)} $R : K \tri V$ is a table (with the name prepended with $\symbolic$) in which attributes from $K$ are mapped to discrete non-null values in $D$ (as before) and value attributes in $V$ are mapped to symbolic expressions in $\RR[X]$.
\end{definition}

We define the \emph{domain} of an s-table $dom(R)$ to be the set of values of its key attributes.
We say that an s-table is \emph{linear} if each value attribute is a linear expression in variables $X$.

An s-table is \emph{ground} if all of the value entries are actually constants, i.e. if it is an ordinary database table (though still satisfying the functional dependency $K \to V$).  The restrictions we have placed on symbolic tables are sufficient to ensure that symbolic evaluation is correct with respect to evaluation over ground tables, as we explain below.

We now clarify how real-world uncertain data (e.g., containing NULLs for unknown values instead of variables, or containing values that might be wrong) can be mapped to s-tables. 

Suppose we are given an ordinary database instance $I\in Inst(\Sigma)$, which may have missing values (a.k.a., NULLs) and uncertain values (i.e., reported values which we do not believe to be exactly correct).  To use the s-table framework, we need to translate such instances to s-instances $I'$ that represent the set of possible complete database instances that match observed data $I$.
In doing this and to allow for the possibility that some reported values present in the data might be inaccurate, we replace such values with symbolic expressions containing variables from $X$. We restrict ourselves to linear expressions for the sake of efficiency, but still our approach allows to do this in many ways, with different justifications based on the application domain, and consequently different quantifications of discordance that should be interpreted accordingly. To exemplify it, in this paper, we replace uncertain values $v$ with $v\cdot(1+x)$ (or simply $x$ if $v=0$) where $x$ is an error variable.  The reason for doing this and not simply $v+x$ is that the weight we associate in our experiments with an error variable $x$ is $x^2$, so the cost of errors is scaled to some extent by the magnitude of the value (e.g., it should be easier to turn $1{,}000{,}000$ into $2{,}000{,}000$ than $1$ into $10$).
On the other hand, a natural way to handle missing values in s-tables is to replace each one in the original relational table with a distinct variable with no associated cost.
In particular scenarios, we might instead prefer to replace each NULL with an expression $c\cdot(1+x)$ where $c$ is an educated guess as to the likely value, but here we consider only the simple approach of a NULL mapped to a variable.
In general, we can assign to attributes any linear expression, with any number of variables, and reuse these variables in any number of attributes or tables.

\begin{example}\label{example.stables}
It is easy to see that, in our example on page \pageref{label:exampleIntro}, there are many possibilities of assigning cases of COVID-19 to the different regions of $\Spain$ that add up to 2,142 in the studied week, and consequently improve the consistency of our database, which may be easily represented by replacing constants by symbolic expressions ``$v_i(1+x_i)$'', where $v_i$ is the corresponding value and $x_i$ is an error parameter representing that cases may be missed or overreported in every region. The cases for region $R_{\text{XIX}}$, that were not reported at all, could then be simply represented by a variable $x_{\text{XIX}}$.
Nevertheless, this may not completely explain the mismatch between cases reported at the country and region levels, and there might also be some doubly-counted or hidden cases in $\Spain$ (for example, in the \emph{Ciudad Autonoma de Melilla}, which this week declared not to have any cases), which we represent by variable $(1+y)$.
On the other hand, we can also consider census data and try to attribute all the excess deaths to COVID-19, which clearly involves some imprecision, too.
So we should apply some error term $(1+z)$ to the declared number of deaths coming from the census, as well.
Therefore, s-tables $\symbolic ReportedRegion:\{region,week\}\tri\{cases\}$, $\symbolic ReportedCountry:\{week\}\tri\{cases\}$ and $\symbolic Declared\-Deaths:\{week\}\tri\{deaths\}$ would contain:

{\small\normalfont
\begin{lstlisting}[language=sql,escapechar=\%]
%$\symbolic$%ReportedRegion("I","2020W24",%$13*(1+x_\text{I})$%)
%$\dots$%
%$\symbolic$%ReportedRegion("XVIII","2020W24",%$12*(1+x_{\text{XVIII}})$%)
%$\symbolic$%ReportedRegion("XIX","2020W24",%$x_{\text{XIX}}$%)
%$\symbolic$%ReportedCountry("2020W24",%$2142*(1+y)$%)
%$\symbolic$%DeclaredDeaths("2020W24",%$3*(1+z)$%)
\end{lstlisting}
}
\vspace{-0.5cm}
\xqed\end{example}

We now make the semantics of our query operations over s-tables precise in Fig.~\ref{fig:evaluation}. An essential property of this semantics is that (for both ground and symbolic tables) the typing rules ensure that a well-formed query evaluated on a valid instance of the input schema yields a valid result table, preserving the desired properties, that is, ensuring that the resulting tables are valid s-tables, and moreover ensuring that the semantics of query operations applied to s-tables is consistent with their behavior on ground tables.  Moreover, symbolic evaluation preserves linearity, which is critical for ensuring that the constrained optimization problems arising from symbolic evaluation fit standard frameworks and can be efficiently solved.

\begin{figure}[tb]
\begin{eqnarray*}
\sigma_c(R) &=& \{t \mid t \in R, c(t) = \mathit{true}\}\\
\hat{\pi}_W(R) &=& \{t[K,V\backslash W] \mid t\in R\}\\
R \Join S &=& \{t[K_R\cup K_S,V_R,V_S] \\
& & \qquad \mid t[K_R,V_R] \in R, t[K_S,V_S] \in S\}\\
R \uplus_D S &=& \{t.D := 0 \mid t \in R\} \cup \{t.D := 1 \mid t \in S\}\\
R \backslash S &=& \{t \mid t \in R, t[K] \notin S\}\\
\rho_{B\mapsto B'}(R) &=& \{t[K\backslash B,V\backslash B].B':=t.B \mid t \in R\}\\
\varepsilon_{B := e}(R) &=& \{t.B := e(t) \mid t \in R\}\\
\gamma_{K';V'}(R) &=& \{t[K'].V' := sum(R,t,K',V')
 \mid t \in R\}
 \\
sum(R,t,K,V) &=& \sum ( t'[V] \mid t' \in R, t[K] = t'[K])
\end{eqnarray*}
\caption{Symbolic Evaluation}
\label{fig:evaluation}
\end{figure}

The following paragraphs describe and motivate the behavior of each operator, and informally explain and justify the correctness of the well-formedness rules ensuring that the result of (symbolic) evaluation is a valid (symbolic) table.
\begin{compactitem}
\item Selection ($\sigma_c(R) : K \tri V$ where $R : K \tri V$). We permit $\sigma_c(R)$  when $c$ is a Boolean formula referring only to fields $A,B,\ldots \in K$.  If comparisons involving symbolic values were allowed, then the existence of some rows in the output could depend on unknown variable values, so would not be representable just using s-tables.
\item Projection-away ($\hat{\pi}_{W}(R) : K \tri V\backslash W$ where $R : K \tri V$ and $W \subseteq V$).  The projection operator projects-away only value-fields. Discarding key-fields could break the finite map property by leaving behind tuples with the same keys and different values.
\item Join ($R \Join S : K_1 \cup K_2 \tri V_1 \cup V_2$ where $R:K_1 \tri V_1$, $S : K_2 \tri V_2$ and $V_1 \cap V_2 = \emptyset$).  Joins can only overlap on key-fields, for the same reason that selection predicates can only select on keys: if we allowed joins on value-fields, then the result of a join would not be representable as an s-table.
\item Discriminated union ($R \uplus_D S : K,B \tri V$ where $R : K \tri V$ and $S: K \tri V$).  The union of two finite maps may not satisfy the functional dependency from keys to values. We instead provide a discriminated union that tags the tuples in $R$ and $S$ with a new key-field $B$  to distinguish the origin.  
\item Renaming ($\rho_{B \mapsto B'}(R) : K[B \mapsto B'] \tri V[B \mapsto B']$ where $R : K \tri V$). Note that since $K$ and $V$ are disjoint, the renaming applies to either a key-field or a value-field, but not both.  In any case, this clearly preserves the finite map property.
\item Difference ($R \backslash S$ where $R : K \tri V$ and $S : K \tri \emptyset$).  The difference of two maps discards from $R$ all tuples whose key-fields are present in $S$.  The result is a subset of $R$ hence still a valid finite map. We assume $S$ has no value components; if not, this can be arranged by projecting them away in advance.
\item Derivation  ($\varepsilon_{B := e}(R) : K \tri V,B$ where $R : K \tri V$ and $e$ is a linear expression over value-fields $V$). 
No new keys are introduced so the finite map property still holds.
\item Aggregation ($\gamma_{K';V'}(R)$ where $R : K \tri V$ and $K'\subseteq K$ and $V' \subseteq V$).  We allow grouping on key-fields and aggregation of value-fields (possibly discarding some of each).  We consider SUM as the only primitive form of aggregation; COUNT and AVERAGE can be easily defined from it.
\end{compactitem}

\begin{example}\label{example.algebra}
Given $\symbolic ReportedRegion:\{region,week\}\tri\{cases\}$, and $\symbolic ReportedCountry:\{week\}\tri\{cases\}$, we can define views:

{\footnotesize
\begin{align*}
\symbolic AggReported &:=\gamma_{\{week\};\{cases\}}(\symbolic ReportedRegion)\\
\symbolic InferredDeaths &:= \hat{\pi}_{cases}( \\
& \varepsilon_{deaths:=0.015*cases}(\symbolic ReportedCountry))
\end{align*}
}

The first one corresponds to the SQL in Example~\ref{example.view}, while the second assumes an average Case-Fatality Ratio (CFR) of $1.5\%$ to estimate the number of deaths based on the number of reported COVID-19 cases in the country. Notice the projection is necessary, because the expression evaluation operation adds a new field, so we must get rid of the $cases$  for the resulting table's signature to match that of $DeclaredDeaths$.
Regarding CFR, we use a single value for simplicity, but it could be declared in an auxiliary table containing different ones per week, as long as these do not contain any variable, which would break the linearity of the expression.
\xqed\end{example}

The above discussion gives a high-level argument that if the input tables are finite maps (satisfying their respective functional dependencies as specified by the schema) then the result table will also be a finite map that satisfies its specified functional dependency.  More importantly, linearity is preserved: if the s-table inputs to an operation are linear, and all expressions in the operation itself are linear, then the resulting s-table is also linear.

\paragraph{Correctness}
We interpret s-tables as mappings from valuations to ground tables, obtained by evaluating all symbolic expressions in them with respect to a global valuation $h : \RR^X$.

\begin{definition}[Valuation]
A \emph{valuation} is a function $h : \RR^X$ assigning constant values to variables.  Given a symbolic expression $e$, we write $h(e)$ for the result of evaluating $e$ with variables  $x$ replaced with $h(x)$.  We then write $h(t)$ for the tuple obtained by replacing each symbolic expression $e$ in $t$ with $h(e)$ and write $h(R)$ to indicate the result of evaluating the expressions in $R$ to their values according to $h$, that is, $h(R) = \{h(t) \mid t \in R\}$.  Likewise for an instance $I$, we write $h(I)$ for the ground instance obtained by replacing each $R$ with $h(R)$.  An s-table $R$ \emph{represents} the set $\banana{R} = \{h(R) \mid h : \RR^X\}$ of ground tables obtained by applying all possible $h$ to $R$.  We write $\banana{I}$ for the set of all ground instances obtainable from an s-instance $I$ by some $h\in \RR^X$, defined as $\banana{I} = \{h(I) \mid h \in \RR^X\}$, and we write $q\banana{I}$ for $\{q(I') \mid I' \in \banana{I}\}$, the set of all possible results of evaluating $q$ on a ground table represented by $I$.
\end{definition}

Given a query expression in our algebra, we can  evaluate it on a ground instance since every ground instance is an s-instance and s-table operations do not introduce variables that were not present in the input.
Further, given a set of ground instances, we can (conceptually) evaluate a query on each element of the set.

\begin{theorem}\label{thm:commutation}
Let $q$ be a well-formed query mapping instances of $\Sigma$ to relations matching $K \tri V$, and let $I$ be an s-instance of $\Sigma$. Then $q\banana{I} = \banana{q(I)}$.
\end{theorem}
The proof is in Appendix~\ref{appendix.correctness}, and is similar to correctness proofs for c-tables~\cite{imielinski84jacm}; the main complication is that in s-instances, variables occurring in different tables are intended to refer to the same unknown values, whereas in c-tables such variables are scoped only at the table level.

\subsection{Fusion, alignment, and coalescing}\label{section.fusion}
We now consider how s-tables and symbolic evaluation can be used to reduce concordance checking and discordance measurement to linear programming and quadratic programming problems, respectively, when we find more than one tuple with the same key and multiple (symbolic) values for the same attribute.
We first consider an abstract \emph{fusion} operator:
\begin{definition}[Fusion]\label{definition.fusion}
Given two ground relations  $R,S : K \tri V$, their \emph{fusion} $R \sqcup S$ is defined as $R \cup S$, provided it satisfies the functional dependency $K \to V$, otherwise the fusion is undefined.
\end{definition}

Thus, our goal is to fuse different sources. However, since they are independent, their concrete values can come in a variety of formats, being expressed in different units, scales or even computation stages (e.g., benefit vs income and expenses). As explained in \cite{DBLP:journals/inffus/MotroA06}, until the conflicts at the representation level have been resolved, those at the data level cannot be resolved (or even measured), either. Therefore,
we represent the expected relationships between source and derived data using a generalization of view specifications called \emph{alignment specifications}.
The alignments must always be defined by the user, considering the domain knowledge, but our set of algebraic operators guarantees the correctness from a computational point of view (i.e., identity of tuples is preserved and expressions are guaranteed to be linear). We generalize fusion to many sources, and actually allow
alignment specifications to define derived tables as the fusion of multiple views.

\begin{definition}[Alignment Specification]
Let $\Sigma$ and $\Omega$ be finite map schemas with disjoint table names $R_1,\ldots,R_n$ and $T_1,\ldots,T_m$, respectively.  Let $\Delta$ be a sequence of view definitions, one for each $T_i$, of the form $T_i := q_1 \sqcup \cdots \sqcup q_k$, where each $q_i$ is a query over finite maps, that refers only to table names in $\Sigma$ and $T_1,\ldots,T_{i-1}$.
The triple $\Spec = [\Sigma,\Omega,\Delta]$ is called an \emph{alignment specification}.
\end{definition}

\begin{figure}[tb]
\begin{mathpar}\small
\inferrule*{\strut}{\Sigma \vdash \cdot : \cdot}
\and
\inferrule*{\
\Sigma \vdash q_1 : K \tri V \\
\cdots\\
\Sigma \vdash q_n : K \tri V \\
\Sigma,S : K \tri V \vdash \Delta : \Omega}
{\Sigma \vdash S := q_1 \sqcup \cdots \sqcup q_n,\Delta : (S{:}K\tri V,\Omega)}
\end{mathpar}
\caption{Well-formed alignment specifications $[\Sigma,\Omega,\Delta]$}\label{fig:spec-wf}
\end{figure}
Alignment specifications are considered well-formed according to the rules in Figure~\ref{fig:spec-wf}.  The first rule handles the base case where the $\Delta$ and $\Omega$ parts of the specification are empty.  The second rule says that a specification $[\Sigma,(S:K \tri V,\Omega),(S:= q_1 \sqcup \cdots \sqcup q_n,\Delta)]$ is well-formed provided each $q_i$ is a query producing an output matching $K \tri V$, and provided the rest of $\Delta$ is well-formed with respect to $\Omega$ if the type of $S$ is added to $\Sigma$.  The purpose of adding $S$ to $\Sigma$ here is to ensure later view definitions may refer both to the tables initially in $\Sigma$ and to earlier view definitions, but view definitions cannot be cyclic (for example $S$ cannot refer to itself or to a view defined later).

\begin{example}\label{example.fusion}
Given the s-tables in Example~\ref{example.stables} and views in Example~\ref{example.algebra}, we can specify the alignment of COVID-19 cases (corresponding to Example~\ref{example.assertions}) and deaths by the following views:

{\footnotesize
\begin{eqnarray*}
\symbolic SumOfCases &:=& \symbolic ReportedCountry \sqcup \symbolic AggReported\\
\symbolic NumberOfDeaths &:=& \symbolic DeclaredDeaths \sqcup \symbolic InferredDeaths
\end{eqnarray*}
}
These view definitions express our intention that in concordant data, the country-level reported data would correspond to the sum of the region-level reports and the deaths according to the census data would be 1.5\% of reported cases, as shown in the equations in Example~\ref{example.queries}.  However, importantly, the fusion operator is not restricted to two arguments, it can express simultaneous coincidence among multiple inputs.
\xqed\end{example}

We implement the abstract fusion operation on s-tables by first making the discriminated union of the input relations ($R \uplus_B S$), and then using a unary operation, called \emph{coalescing}, whose behavior on sets $\mathcal{R}$ of ground tables $R : K,B \tri V$ is
$\kappa_B(\mathcal{R}) = \{\hat{\pi}_B(R) \in \mathcal{R} \mid \text{$\hat{\pi}_B(R)$ satisfies $K \to V$}\}$.
Intuitively, coalescing of a set of tables $R \in \mathcal{R}$ applies a projection $\hat{\pi}_B$ to each $R$, and returns those projected tables that still satisfy the FD $K \to V$.  These are the rows where the values associated with the corresponding keys are consistent across all inputs in which the keys are present.
To represent the result of coalescing using s-tables, we augment them with constraints.
A constraint $\phi$ is simply a conjunction of linear equations;
a constrained s-table is a pair $(R,\phi)$ that represents the set of possible ground tables $\banana{R,\phi} = \{h(R) \mid h : \RR^X, h \models \phi\}$; finally a constrained s-instance $(I,\phi)$ likewise represents a set of ground instances of $I$ obtained by valuations satisfying $\phi$.
We can implement coalescing as an operation on constrained s-tables as follows:
\[\small\begin{array}{rcl}
\kappa_D(R,\phi)&=&(T,\phi\wedge \psi)\\
R^+ &=& \{t\in R \mid \exists t'\in R, t[K\backslash D]=t'[K\backslash D]\wedge t[D]\neq t'[D]\}\\
T&=&\{t[K\backslash D,B_1:=L_{t[K\backslash D],B_1},\dots]\mid t\in R^+\}\\
&\cup&\{t[K\backslash D,V]\mid t\in R \backslash R^+\}\\
\displaystyle\psi &\equiv& \bigwedge_{t\in R^+}\bigwedge_{B_i\in V} L_{t[K\backslash D],B_i}=t[B_i]
\end{array}\]

That is, let $R^+$ be the set of tuples in $R$ for which there exists another tuple that has the same values on $K\backslash D$ but differs on $D$, and let $R \backslash R^+$ be the remaining tuples (for which there are no such sibling tuples).  Thus, $R^+$ is the set of tuples of $R$ potentially violating the FD $K\backslash D \rightarrow V$, and $R \backslash R^+$ is the largest subset of $R$ that satisfies this FD.  Then $\kappa_D(R)$ consists of table $T$ obtained by filling in new variables $L_{t[K\backslash D],B_i}$.  
$L$-values are used only where there may be disagreement, and we use the value from $R$ otherwise.  The constraint $\psi$ consists of equations between the observed values of each attribute and the corresponding $L$-value.  No constraints are introduced for tuples in $R \backslash R^+$, where there is no possibility of disagreement.

\begin{example}\label{example.coalescing}
Thus, $\symbolic SumOfCases$ is implemented in terms of our algebra as $$\kappa_D(\symbolic ReportedCountry \uplus_D \symbolic AggReported),$$ and $\symbolic NumberOfDeaths$ as $$\kappa_D(\symbolic DeclaredDeaths \uplus_D \symbolic InferredDeaths).$$ As a result, we introduce two new variables (i.e., respectively $L_{2020W24,cases}$ and $L_{2020W24,deaths}$), and would obtain the following constrained s-tables:

{\small\normalfont
\begin{lstlisting}[language=sql,escapechar=\%]
%$\symbolic$%SumOfCases("2020W24",%$L_{2020W24,cases}$%)
\end{lstlisting}
\[\small\begin{array}{rcl}
\psi_{SumOfCases} &\equiv& L_{2020W24,cases}=2142(1+y) \\
     &\wedge& L_{2020W24,cases}=13(1+x_\text{I})+\dots \\
     & &+12(1+x_{\text{XVIII}})+x_{\text{XIX}}
\end{array}\]
\begin{lstlisting}[language=sql,escapechar=\%]
%$\symbolic$%NumberOfDeaths("2020W24",%$L_{2020W24,deaths}$%)
\end{lstlisting}
\[\small\begin{array}{rcl}
\psi_{NumberOfDeaths} &\equiv& L_{2020W24,deaths}=0.015*1995(1+y) \\
       &\wedge& L_{2020W24,deaths}=3*(1+z)
\end{array}\]
}

Notice that if another source had provided data at region level, two possible alignments would appear: (a) aggregating also the new source (as done before with $\symbolic ReportedRegion$) and making the correspondence also through the same existing fusion $\symbolic SumOf\-Cases$ (i.e., $L_{2020W24,cases}$ would be simply reused with another conjunct clause in $\psi$), or (b) creating a new independent fusion of this new source and $\symbolic ReportedRe\-gion$ which would generate thirteen new $L$-values (one per region) and the corresponding thirteen logic clauses to be added to $\psi$.
Obviously, (b) is more restrictive than (a), but the choice would depend on the user defining the most appropriate alignment to the use case.
\xqed\end{example}

\subsection{Discord measurement}\label{sec:alltogether}

Having introduced (constrained) s-tables, and evaluation for query operations and coalescing over them, we finally show how these technical devices allow us to define concord and measure discord.

\begin{definition}[Concordant Instance]
Given a specification $\Spec = [\Sigma,\Omega,\Delta]$, an instance $I$ of schema $\Sigma$ is \emph{concordant} if there exists an instance $J$ of $\Omega$ such that for each view definition $T_i := q_1 \sqcup \cdots \sqcup q_n$ in $\Delta$, we have $J(T_i) = q_1(I,J) \sqcup \cdots \sqcup q_n(I,J)$ where $q(I,J)$ is the result of evaluating $q$ on the combined instance $I,J$ and all of the fusion operations involved are defined.  The concordant instances of $\Sigma$ with respect to $\Spec$ are written $\Conc(\Spec)$.
\end{definition}

\begin{definition}[Concordance]
Given $[\Sigma,\Omega,\Delta]$, let $I$ be an s-instance.
We say $I$ is \emph{concordant} if there exists a concordant instance $C \in \banana{I}$.
\end{definition}

Given an alignment specification $\Spec = [\Sigma,\Omega,\Delta]$ and an s-instance $I$, we can check concordance by symbolically evaluating $\Delta$ on $I$ to obtain an s-instance $J$ as follows.  For each view definition $T_i := q_1\sqcup \cdots \sqcup q_n$ in $\Delta$ in order, evaluate $q_1,\ldots,q_n$ to their s-table results and fuse them using the coalescing operator (repeatedly if $n>2$).  This produces a new s-table $T_i'$ and a constraint $\phi_i$.  Add $T_i := T'_i$ to $J$ and continue until all of the table definitions in $\Delta$ have been symbolically evaluated (i.e., $J = [T_1 := T_1',\ldots,T_m:=T_m']$). Thus, the constrained s-instance $(I',\phi)$ where $\phi = \bigwedge_{i=1}^m \phi_i$ characterizes the set of possible concordant instances based on $I'$, and in particular $I$ is concordant if $\phi$ is satisfiable.

\begin{example}\label{example.solutions}
From the constrained s-tables in Example~\ref{example.coalescing}, obtained by the corresponding coalescing operation, we get the next intertwined system of equations:
\begin{eqnarray*}
       2142(1+y) &=& 13(1+x_\text{I})+\dots+12(1+x_{\text{XVIII}}) \\
       & & +x_{\text{XIX}} \\
       0.015*1995(1+y) &=& 3*(1+z)
\end{eqnarray*}

Obviously,  this system has many solutions.
One solution $S_1$ consists of taking all $x_i$ to be zero, $y=-0.07$ and $z=8.975$.  This corresponds to assuming there is no error in the eighteen regions' reports and there are no cases in Region XIX.  Another solution $S_2$ sets $x_\text{I} = ... = x_{\text{XVIII}} = 0$ and $x_{\text{XIX}} = 147$, then $y = 0$ and $z = 9.71$ which corresponds to assuming $R_{\text{XIX}}$ had all of the missing COVID-19 cases in $\Spain$ that week.  Of course, whether $S_1$ or $S_2$ (or some other) is more plausible depends on domain-specific knowledge.
\xqed\end{example}

Given a \emph{cost function} $\delta$ assigning a cost to each solution, we can compare different solutions in terms of how much correction is needed (or discord exists).
Thus, the discord is, intuitively, the shortest  $\delta$-distance between the actual observed, uncertain data (represented as a set of possible instances) and a hypothetical concordant database instance that is consistent with the alignment specification.
The more distant from any such concordant instance, the more discordant our data are.

\begin{definition}[Discordance]
Given $\Spec = [\Sigma,\Omega,\Delta]$, let $\delta : Inst(\Sigma) \times Inst(\Sigma) \to \RR$ be a measure of distance between pairs of elements of $Inst(\Sigma)$.  The \emph{discordance} of a (constrained) s-instance $I$ is $$\inf_{J \in \banana{I}, C \in \Conc(\Spec)}\delta(J,C)$$
\end{definition}

Then, the degree of discordance of $I$ given the alignment $\Delta$ and according to $\delta$ (i.e., $disc_\delta(I,\phi)$) equals the solution to the quadratic programming problem formed by minimizing $\delta$ subject to $\phi$.
Depending on the choices of metrics, this leads to well-understood constrained optimization problems such as linear programming or least squares optimization.
Linear programming has polynomial time complexity, and the popular simplex algorithm is worst-case exponential but has good behavior in practice~\cite{spielman01stoc}.  Quadratic programming is NP-Complete in general, but with good computational behaviour in many practical applications~\cite{osqp}.  Moreover in the \emph{convex} case when the cost function is based on a \emph{positive semidefinite} matrix, quadratic programming is very well-behaved, having a single global solution that can be found in polynomial time~\cite{vavasis01optimization}.  In our experiments, we used the convex metric
$\delta$ defined as the sum of squares of the error variables in $I$.

Like the alignment, the cost function evaluating
the discordance is also strongly domain knowledge dependent, but our system is flexible enough to consider different alternatives (e.g., different weights; if the weights are non-negative then the problem remains convex). Nevertheless, in the rest of this paper, we will only use simplistic (and convex) cost functions like the sum of squares suggested in \cite{DBLP:journals/sigkdd/LiGMLSZFH15}.   This guarantees that the problems we generate are solvable in polynomial time; we also establish feasibility in practice for datasets of moderate size. Although the underlying quadratic solver we use can accommodate more general quadratic functions of the symbolic variables, we leave exploration of more sophisticated cost functions to future work.

\begin{example}
Considering simply the sum of the squares of the variables:
\[c_1(\vec{x},y,z) = (\sum_{i \in \{\text{I},\ldots,\text{XIX}\}} x_i^2) + y^2 + z^2\]
For the two solutions in Example~\ref{example.solutions}, $S_1$ has cost $\approx 80.56$, while $S_2$ has cost 21,703.28, so with the above cost function the first one is much closer to being concordant, because a large change to $x_{\text{XIX}}$ is not needed.  Alternatively, we might give the unknown number of cases in $R_{\text{XIX}}$ no weight, reflecting that we have no knowledge about what it might be, corresponding to the cost function
\[c_2(\vec{x},y,z) = (\sum_{i \in \{\text{I},\ldots,\text{XVIII}\}} x_i^2) + y^2 + z^2\]
that assigns the same cost to $S_1$ but assigns cost $94.28$ to $S_2$, indicating that if we are free to assign all unaccounted cases to $x_{\text{XIX}}$ then the second solution is closer to concordance.

Furthermore, we could also weight variables considering the reliability of the different regions as well as the central government, and the historical information of the census, but it is important to notice that such weights would depend on knowledge of the domain.
\xqed\end{example}

It is important to mention also that as a side-product, the instanciations of the $L$-values introduced by coalescing could be used to obtain a concordant instance, but this instance is not guaranteed to provide the true values of the uncertain indicators, it is only an estimate minimizing the distance metric.  Thus domain-specific knowledge or statistical techniques also need to be applied to characterize the quality of these estimates.

\begin{figure*}
\hspace{0.15\linewidth}
\begin{minipage}[c]{0.35\linewidth}
{\tiny\center
\textbf{ReportedRegion}\\
\begin{tabular}{|l|r|r|c|}\hline
\multicolumn{1}{|c|}{\textbf{dist.}} & \multicolumn{1}{c|}{\textbf{week}} & \multicolumn{1}{c|}{\textbf{cases}} & \multicolumn{1}{c|}{\textbf{deaths}} \\ \hline \hline
I & 2020W24 & $13x_{\text{I}}+13$ & $\cdots$ \\ \hline
II & 2020W24 & $3x_{\text{II}}+3$ & $\cdots$ \\ \hline
III & 2020W24 & $8x_{\text{III}}+8$ & $\cdots$ \\ \hline
IV & 2020W24 & $62x_{\text{IV}}+62$ & $\cdots$ \\ \hline
V & 2020W24 & $30x_{\text{V}}+30$ & $\cdots$ \\ \hline
VI & 2020W24 & $7x_{\text{VI}}+7$ & $\cdots$ \\ \hline
VII & 2020W24 & $71x_{\text{VII}}+71$ & $\cdots$ \\ \hline
VIII & 2020W24 & $826x_{\text{VIII}}+826$ & $\cdots$ \\ \hline
IX & 2020W24 & $243x_{\text{IX}}+243$ & $\cdots$ \\ \hline
X & 2020W24 & $295x_{\text{X}}+295$ & $\cdots$ \\ \hline
XI & 2020W24 & $7x_{\text{XI}}+7$ & $\cdots$ \\ \hline
XII & 2020W24 & $336x_{\text{XII}}+336$ & $\cdots$ \\ \hline
XIII & 2020W24 & $5x_{\text{XIII}}+5$ & $\cdots$ \\ \hline
XIV & 2020W24 & $63x_{\text{XIV}}+63$ & $\cdots$ \\ \hline
XV & 2020W24 & $10x_{\text{XV}}+10$ & $\cdots$ \\ \hline
XVI & 2020W24 & $4x_{\text{XVI}}+4$ & $\cdots$ \\ \hline
XVII & 2020W24 & $x_{\text{XVII}}$ & $\cdots$ \\ \hline
XVIII & 2020W24 & $12x_{\text{XVIII}}+12$ & $\cdots$ \\ \hline
\end{tabular}
}
\end{minipage}
\begin{minipage}[c]{0.06\linewidth}
{\huge $\Rightarrow$}
\end{minipage}
\begin{minipage}[c]{0.35\linewidth}
\textbf{Non-First Normal Form (\sparsenftwo)}\\
\fbox{
\begin{minipage}[c]{\linewidth}
{\tiny\center
\textbf{ReportedRegion}\\
\begin{tabular}{|l|r|r|r|}\hline
\multicolumn{1}{|c|}{\textbf{dist.}} & \multicolumn{1}{c|}{\textbf{week}} & \multicolumn{1}{c|}{\textbf{cases}} & \multicolumn{1}{c|}{\textbf{deaths}}  \\ \hline \hline
I & 2020W24 & $\{\{13,x_{\text{I}}\},13\}$ & $\cdots$ \\ \hline
$\cdots$ & $\cdots$ & $\cdots$ & $\cdots$  \\ \hline
XVII & 2020W24 & $\{\{1,x_{\text{XVII}}\},0\}$ & $\cdots$ \\ \hline
XVIII & 2020W24 & $\{\{12,x_{\text{XVIII}}\},12\}$ & $\cdots$ \\ \hline
XIX & 2020W24 & $\{\{1,x_{\text{XIX}}\},0\}$ & $\cdots$ \\ \hline
\end{tabular}
\\
}
\end{minipage}
}
\vspace{0.1cm} \\
\textbf{Normalized (Partitioning)}\\
\fbox{
\begin{minipage}[c]{\linewidth}
{\tiny\center
\textbf{ReportedRegion}\\
\begin{tabular}{|l|r|r|r|}\hline
\multicolumn{1}{|c|}{\textbf{dist.}} & \multicolumn{1}{c|}{\textbf{week}} & \multicolumn{1}{c|}{\textbf{cases}} & \multicolumn{1}{c|}{\textbf{deaths}}  \\ \hline \hline
I & 2020W24 & $13$ & $\cdots$  \\ \hline
$\cdots$ & $\cdots$ & $\cdots$ & $\cdots$  \\ \hline
XVII & 2020W24 & $0$ & $\cdots$ \\ \hline
XVIII & 2020W24 & $12$ & $\cdots$ \\ \hline
XIX & 2020W24 & $0$ & $\cdots$  \\ \hline
\end{tabular}
\\
\vspace{0.20cm}
\textbf{ReportedRegion\_cases}\\
\begin{tabular}{|l|r|l|r|}\hline
\multicolumn{1}{|c|}{\textbf{dist.}} & \multicolumn{1}{c|}{\textbf{week}} & \multicolumn{1}{c|}{\textbf{var.}} & \multicolumn{1}{c|}{\textbf{coeff.}}  \\ \hline \hline
I & 2020W24 & $x_{\text{I}}$ & $13$  \\ \hline
$\cdots$ & $\cdots$ & $\cdots$ & $\cdots$  \\ \hline
XVII & 2020W24 & $x_{\text{XVII}}$ & $1$ \\ \hline
XVIII & 2020W24 & $x_{\text{XVIII}}$ & $12$ \\ \hline
XIX & 2020W24 & $x_{\text{XIX}}$ & $1$  \\ \hline
\end{tabular}
\\
\vspace{0.20cm}
\textbf{ReportedRegion\_deaths}\\
\begin{tabular}{|l|r|l|r|}\hline
\multicolumn{1}{|c|}{\textbf{dist.}} & \multicolumn{1}{c|}{\textbf{week}} & \multicolumn{1}{c|}{\textbf{var.}} & \multicolumn{1}{c|}{\textbf{coeff.}}  \\ \hline \hline
$\cdots$ & $\cdots$ & $\cdots$ & $\cdots$  \\ \hline
\end{tabular}
\\
}
\end{minipage}
}

\end{minipage}
\caption{\label{fig:map} Implemented encodings of s-tables}
\end{figure*}

\section{Implementation}\label{sec:implementation}

We now describe the techniques employed in \sysname, an implementation of our approach.  The systems of equations resulting from constraints generated on coalescing tables or instances are linear, so they can be solved using linear algebra solvers.  However, it may not be immediately obvious how to evaluate queries over s-tables to obtain the resulting systems of equations efficiently.
One strategy would simply be to load all of the data from the database into main memory and evaluate the s-table query operations in-memory.  While straightforward, this may result in poor performance or duplication of effort, as many query operations which are efficiently executed within the database (such as joins) may need to be hand-coded in order to obtain acceptable performance.
Instead, we propose a relational representation of s-table queries such that the s-table operations can be implemented as ordinary (extended) relational algebra queries over the representation. Thus,
we consider two in-database representations of s-tables, illustrated in Fig.~\ref{fig:map}: A denormalized sparse vector representation using nested user-defined data types (\sparsenftwo), and a normalized representation using multiple flat relations (\partitioning).
Hence, whichever representation we choose, we need to transform the original ground tables into s-tables with variables. This could be done by creating simply a copy, but we found this would materialize a great deal of intermediate data that is not ultimately needed. Instead, the most efficient approach we found is defining the s-tables as \emph{views} over the sources.
These views are straightforward, except for the generation of the variables (using SQL:1999 features such as \texttt{ROW\_NUMBER} to create unique ids for each of them), and the consideration of special cases (i.e., NULL and zero), that is done by means of standard SQL \texttt{CASE} clauses to distinguish them.
Notice that this approach based on views and \texttt{CASE} clauses avoids the manual definition of expressions for every value, and allows the definitions of general rules based on the table name, attribute name, or even concrete attribute values to do it. For example, a rule could indicate the usage of a given expression for some concrete region or change the expression from a given point in time on.

In the \sparsenftwo approach, we add a user-defined type for the symbolic (linear) expression fields.  There are several ways of doing this, like for example using arrays to represent vectors of coefficients for a (small) fixed number of variables, or using a sparse representation that can deal with arbitrary numbers of variables efficiently when most coefficients are zero.  Having experimented with several options, we chose a representation in which symbolic expressions $\sum_i a_i \cdot x_i + b$ are represented as sparse vectors, specifically as a pair of the value of $b$ and an array of pairs $[(a_i, x_i),\ldots]$ of coefficients and variable names. Thus,   $\symbolic ReportedRegion$ \sparsenftwo implementation would correspond to the following SQL view.

{\small\normalfont
\begin{lstlisting}[language=sql,escapechar=\%]
CREATE VIEW ReportedRegion_NF2 AS
SELECT region, week,
 ROW(ARRAY[ROW(
  CASE
   WHEN cases IS NULL OR cases = 0 THEN 1
   ELSE cases
   END, (
  CASE
   WHEN cases IS NULL THEN 'MarkNULL_'
   ELSE ''
   END||'ReportedRegion_cases_')||row_number()
    OVER (ORDER BY region, week))::term],
  CASE
   WHEN cases IS NULL THEN 0.0
   ELSE cases
   END)::sparsevec AS cases
 FROM ReportedRegion;
\end{lstlisting}
}

The keys remain unchanged, and for the value we define a row of type \texttt{sparsevec}, which is actually an array of terms, represented by pairs ``(coefficient,variable)''.
The coefficient is ``1'' if the reported value was either NULL or zero, or the reported value otherwise.
Regarding the variable, we create a different one for every tuple using the ``row\_number'', and for the sake of traceability we concatenate to it both the table and attribute names. Moreover, we also add a special mark in case the reported value was NULL, so this can be considered in the cost function.
We implemented addition, scalar multiplication, and aggregation (sum) of linear expressions using PostgreSQL's user-defined function facility.
With this encoding, the SQL queries corresponding to our algebra are straightforward by applying the standard translation and inserting user-defined functions.  We therefore do not present this translation in detail.

Though many RDBMSs support similar mechanisms to create user-defined functions and aggregates, they are not standard and so \sparsenftwo is not very portable.  Thus, the alternative approach we present, \partitioning, relies only on standard SQL:1999. In this approach, we represent an $s$-table $R : K \tri V$ with $n$ symbolic value-fields $B_1,\ldots,B_n$ using $n+1$ relational tables, as follows:

\begin{itemize}
    \item $R_0 : K \tri V$ is a ground table with all constant terms.
    \item For each symbolic field $B_i \in V$, $R_{B_i} : K,X \tri C$ is a ground table mapping keys $K$ and an additional key $X$ (corresponding to the variables) to a real value-field $C$, so that $(k,x,c) \in R_{B_i}$ when $c$ is the coefficient of $x$ in the $B_i$-attribute of $R(k)$.
\end{itemize}

We consider the relations corresponding to the symbolic value-fields to be collected in a record $\vec{R} = (B_1=R_{B_1},\ldots)$, and we write $(R_0,\vec{R})$ for the full representation.
This representation admits relatively simple translations of each of the s-table query operations, as shown in Figure~\ref{fig:translation}.

\begin{figure*}[tb]
\[
\begin{array}{rcl}
\sigma_P(R_0,\vec{R}) &=& (\sigma_P(R_0), \sigma_P(\vec{R}))\\
\hat{\pi}_W(R_0,\vec{R}) &=& (\hat{\pi}_W(R_0), \vec{R}[V\backslash W])\\
\rho_{[B \mapsto B']}(R_0,\vec{R}) &=& (\rho_{[B \mapsto B']}(R_0),\vec{R}[B \mapsto B']))
\\
(R_0,\vec{R}) \uplus_D (S_0,\vec{S}) &=& (R_0 \uplus_D S_0, (B:=R.B \uplus_D S.B)_{B\in V})\\
\varepsilon_{B:=c}(R_0,\vec{R}) &=& (\varepsilon_{B:=c}(R_0),(\vec{R},B:=\emptyset))\\
\varepsilon_{B:=B_i+B_j}(R_0,\vec{R}) &=& (\varepsilon_{B:=B_i+B_j}(R_0), (\vec{R},B:=\tilde{\gamma}_{K,X;C}(R.B_i \uplus_D R.B_j)))\\
\varepsilon_{B:=\alpha\cdot B_i}(R_0,\vec{R}) &=& (\varepsilon_{B:=\alpha\cdot B_i}(R_0),(\vec{R},B:=\tilde{\varepsilon}_{C:=\alpha\cdot C}(R_{B_i})))\\
(R_0,\vec{R}) \backslash (S_0,()) &=& (R_0 \Join ( \pi_K(R_0)\backslash S_0),(B=R.B \Join \pi_K(R_0)\backslash S_0)_{B\in V})\\
(R_0,\vec{R}) \Join (S_0,\vec{S}) &=& (R_0\Join S_0, (B:=R.B_R \Join (\pi_{K_S}(S_0)),
 B':=S.B_S\Join(\pi_{K_R}(R_0)))_{B_R\in V_R, B_S\in V_S})\\
\gamma_{K';V'}(R_0,\vec{R}) &=& (\gamma_{K';V'}(R_0),(B:=\tilde{\gamma}_{K',X;C}(R.B))_{B \in V'})
\end{array}
\]
\caption{Translation of queries to Partitioning implementation}\label{fig:translation}
\end{figure*}

The operations $\tilde{\gamma}$ and $\tilde{\varepsilon}$ are zero-filtering versions of aggregation and derivation respectively, which remove any rows whose $C$-value is zero.  Filtering out zero coefficients is not essential but avoids retaining unnecessary records since absent coefficients are assumed to be zero.
In the rule for selection, recall predicate $P$ only mentions key attributes; we write $\sigma_P(\vec{R})$ for the result of applying the selection to each table in $\vec{R}$.  In the rule for projection-away, we assume $R : K \tri V$ and $W \subseteq V$, and $\vec{R}[V\backslash W]$ is the record resulting from discarding the fields corresponding to attributes in $W$.  Likewise in the rule for renaming, $\vec{R}[B \mapsto B']$ stands for renaming field $B$ of $\vec{R}$ to $B'$ if present, otherwise applying $\rho_{[B \mapsto B']}(-)$ to each table in $\vec{R}$.
In the rule for addition, we introduce a dummy discriminant in the union, and just use zero-filtering aggregation $\tilde{\gamma}_{K,X;C}$  to sum coefficients grouped by key and variable name (i.e., getting rid of the dummy discriminant).  Likewise, in the case for scalar multiplication, $\tilde{\varepsilon}_{C:=\alpha\cdot C}(R_{B_0})$ operation does an in-place update and finally filters out any zero coefficients. Note that for the sake of understanding, we provide separate rules for assigning a field a constant value, adding two fields, and scalar multiplication of a field, while the query language given earlier allows derivations to assign a new field the result of any linear expression.  Arbitrary linear expressions can be handled by introducing intermediate fields and projecting them away at the end, for example $\varepsilon_{B := C+42}(q) = \hat{\pi}_D(\varepsilon_{B := C + E}(\varepsilon_{E:=42}(q))$.
The rule for difference is slightly tricky because since $S_0$ does not have value attributes in the key, so just subtracting it from each of the $R.B$ would not work.  Instead, we compute the set of keys present after the difference and restrict each $R.B$ to that set of keys (using a join).  The rule for join is likewise a little more involved: given $R : K_R \tri V_R$ and $S : K_S \tri V_S$, since $V_R$ and $V_S$ are disjoint, it suffices for each field $B_R$ of $V_R$ to join the corresponding table $R.B_R$ with the keys of $S_0$, i.e. $\pi_{K_S}(S_0)$, and symmetrically for $S$'s value-fields $B_S$.  Finally, for aggregation we assume $R: K \tri V$ with $K' \subseteq K$ and $V'\subseteq V$, and again  use $\tilde{\gamma}$.

Finally, we comment on constraint generation performed by the coalescing operator.
It simply detects repeated values after projecting out the discriminant and generates the corresponding constraints as an extra query (i.e., a query over an s-table is actually a pair of queries: one retrieving the data $I'$ without repetitions in the key and fresh $L$-values in the values, and another independent query creating the constraints $\phi$ over those $L$-values).  With this approach, coalescing (hence fusion) becomes  first-class and can be freely composed with the other operations, so we can convert a specification into a single composed query (referring to the views that define the s-tables) whose translation generates the equations directly.  This is the approach we have evaluated, which significantly improves the naive materialization approach, because it avoids the need to load and scan numerous intermediate s-tables.

\section{Evaluation}\label{sec:experiments}

\begin{figure}
\begin{center}
\includegraphics[width=10cm]{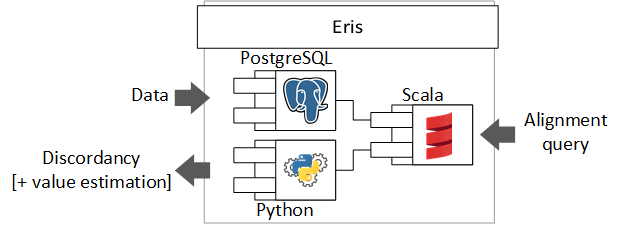}
\end{center}
\caption{\label{fig:ErisArchitecture} Eris architecture}
\end{figure}

\sysname (whose components are depicted in Figure~\ref{fig:ErisArchitecture}) was implemented\footnote{\url{https://github.com/dtim-upc/Eris}} in approximately 4000 lines of Scala code, with approximately 100 lines of SQL defining auxiliary operations, user-defined types, and functions involving sparse vectors.

Before launching anything, all the user data needs to be uploaded into regular PostgreSQL tables.
Then, on choosing the preferred representation of s-tables (either \sparsenftwo or \partitioning), the corresponding views are created to virtually generate the variables.
Once this is done, the input of the system is any alignment specification expressed in our algebra.
Our Scala code transforms such a specification into regular SQL that returns the requested data from user ground tables.
As soon as some of the s-tables are coalesced and some potential violations of the corresponding FDs appear, $L$-values are automatically created, and this triggers another SQL query for the generation of the constraints and the specific Python code to find the values of the involved variables from the retrieved information.
The corresponding linear and quadratic programming subproblems are solved using version 0.6.1 of OSQP~\cite{osqp}, called as a Python library with the default configuration and no parameter tuning. 

Since, to our knowledge, there is not any other system that can automatically generate a (configurable) measurement of discordance in the presence of semantic heterogeneities between the sources, we cannot make any meaningful comparison to show that this is faster or can do things that the others can not.
Instead, we show its scalability in terms of query performance, and the expressive power and usefulness by means of a use case. We do not try to justify the goodness of the sum of squares as an indicator of discordance, because this is absolutely configurable in \sysname, hence, justifying its use is out of the scope of this work

Experiments were run on a workstation equipped with an Intel Xeon E5-1650 with 6 cores, 32 GB RAM, running Ubuntu 16.10, and using a standard installation of PostgreSQL 9.5. They evaluate \sysname from the perspective of both performance and usefulness.

\subsection{Performance microbenchmarks}

\begin{figure*}
\includegraphics[scale=0.06]{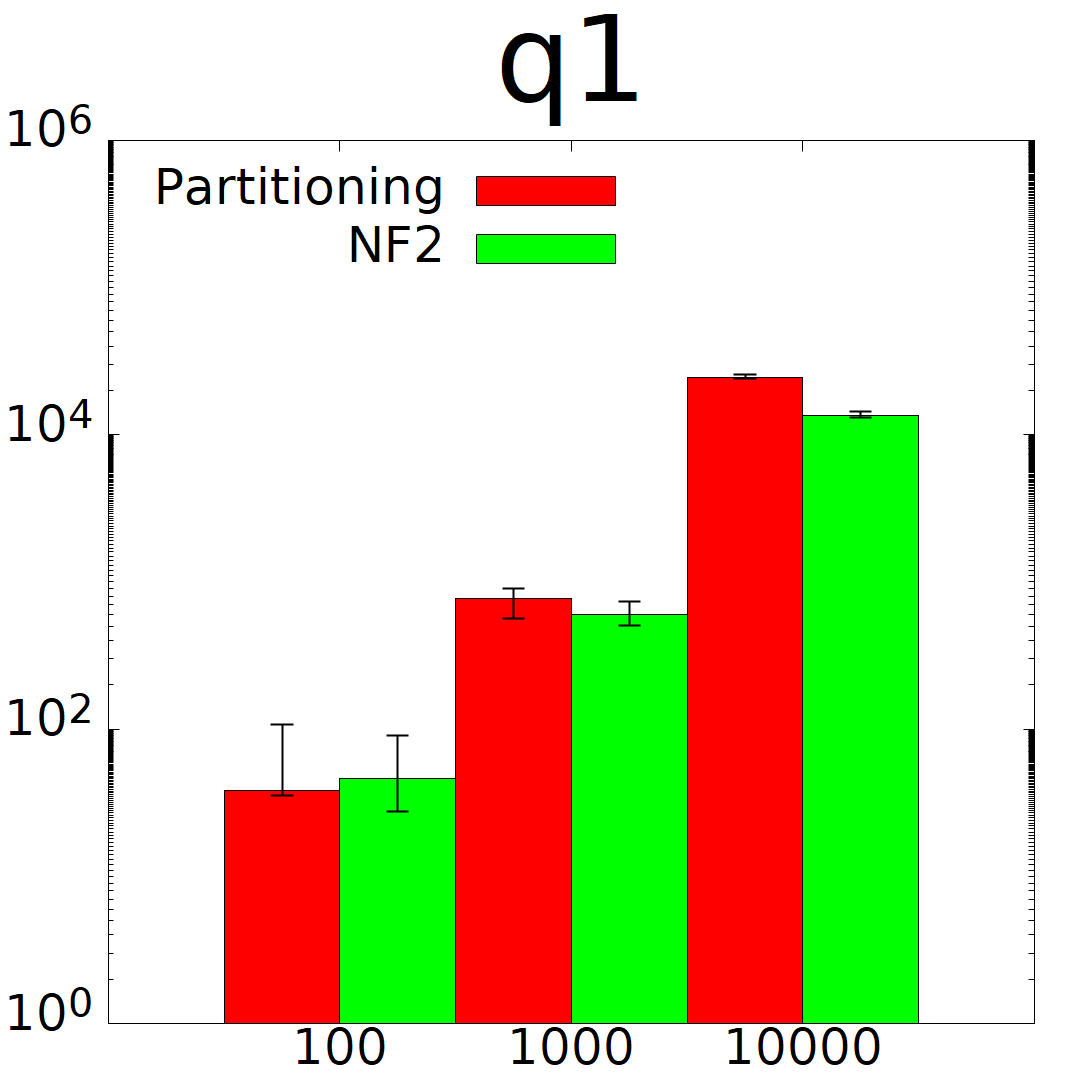}
\includegraphics[scale=0.06]{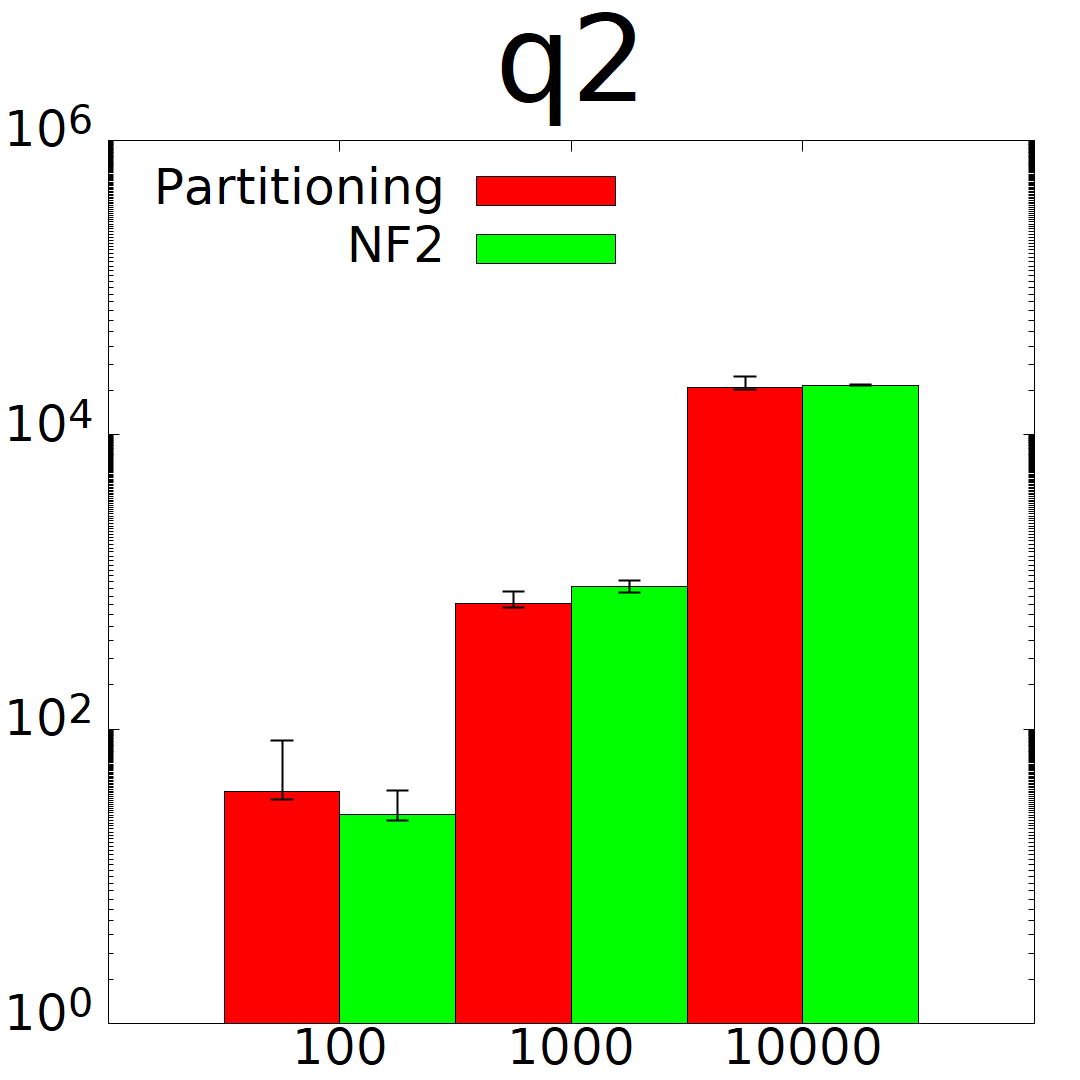}
\includegraphics[scale=0.06]{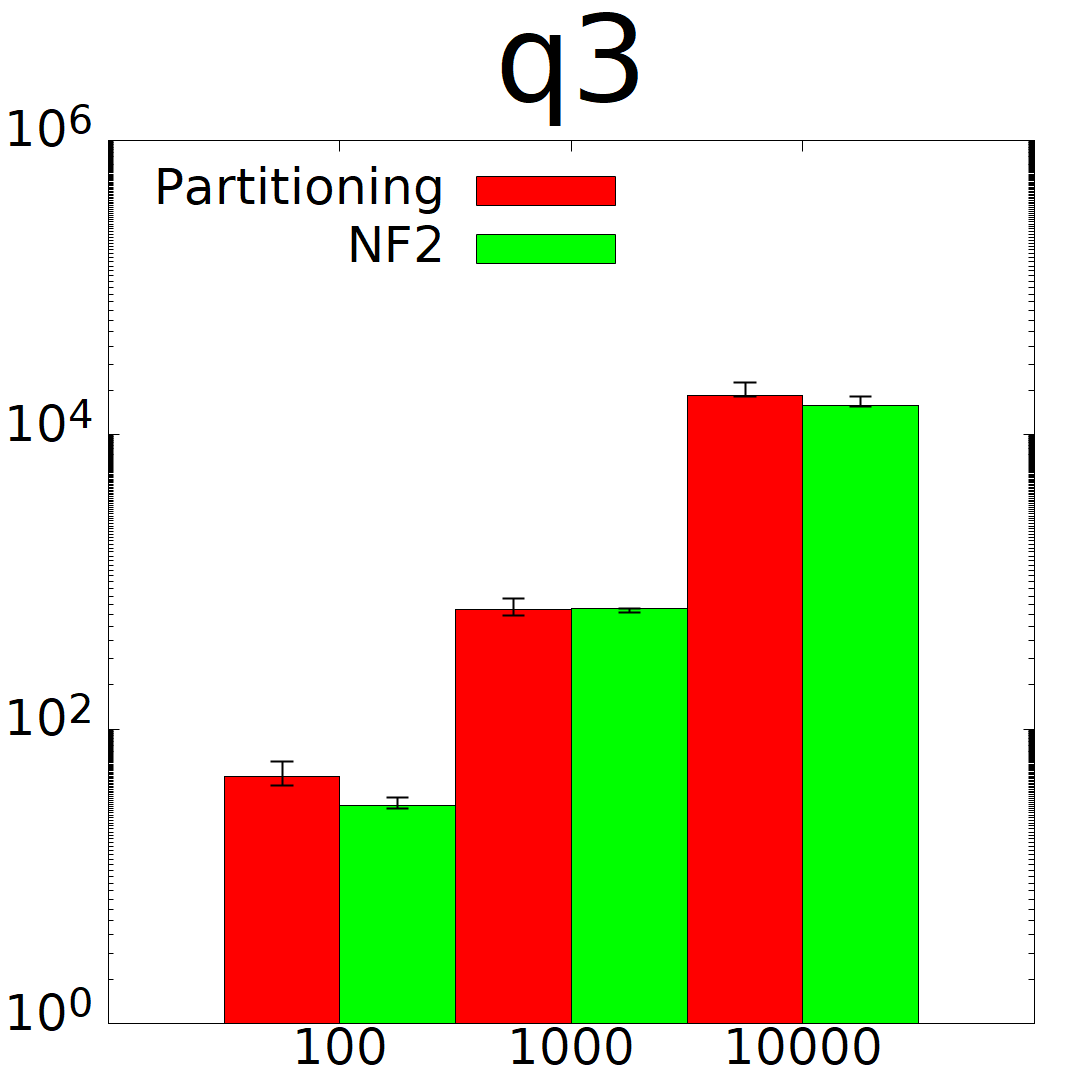}
\includegraphics[scale=0.06]{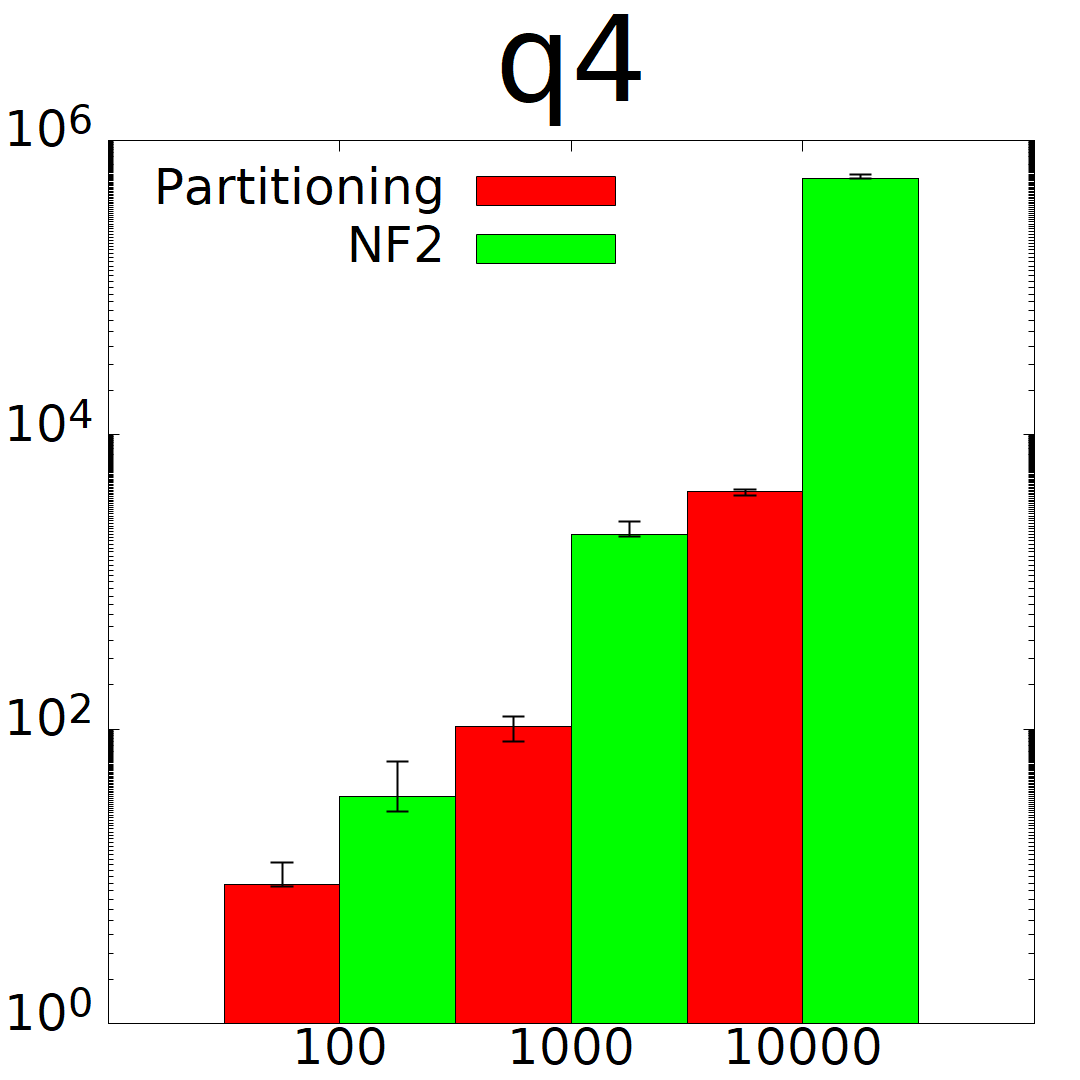}
\includegraphics[scale=0.06]{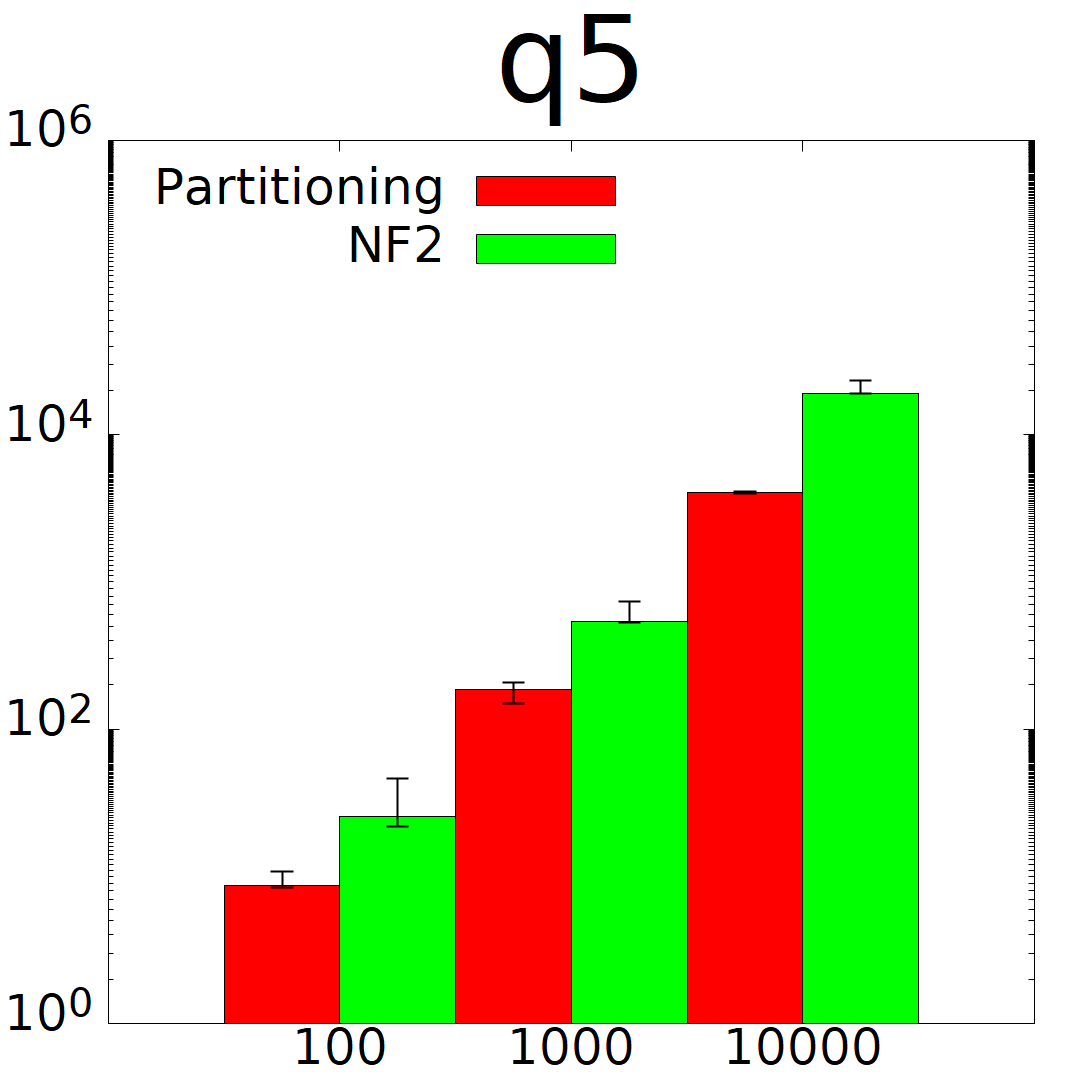}
\includegraphics[scale=0.06]{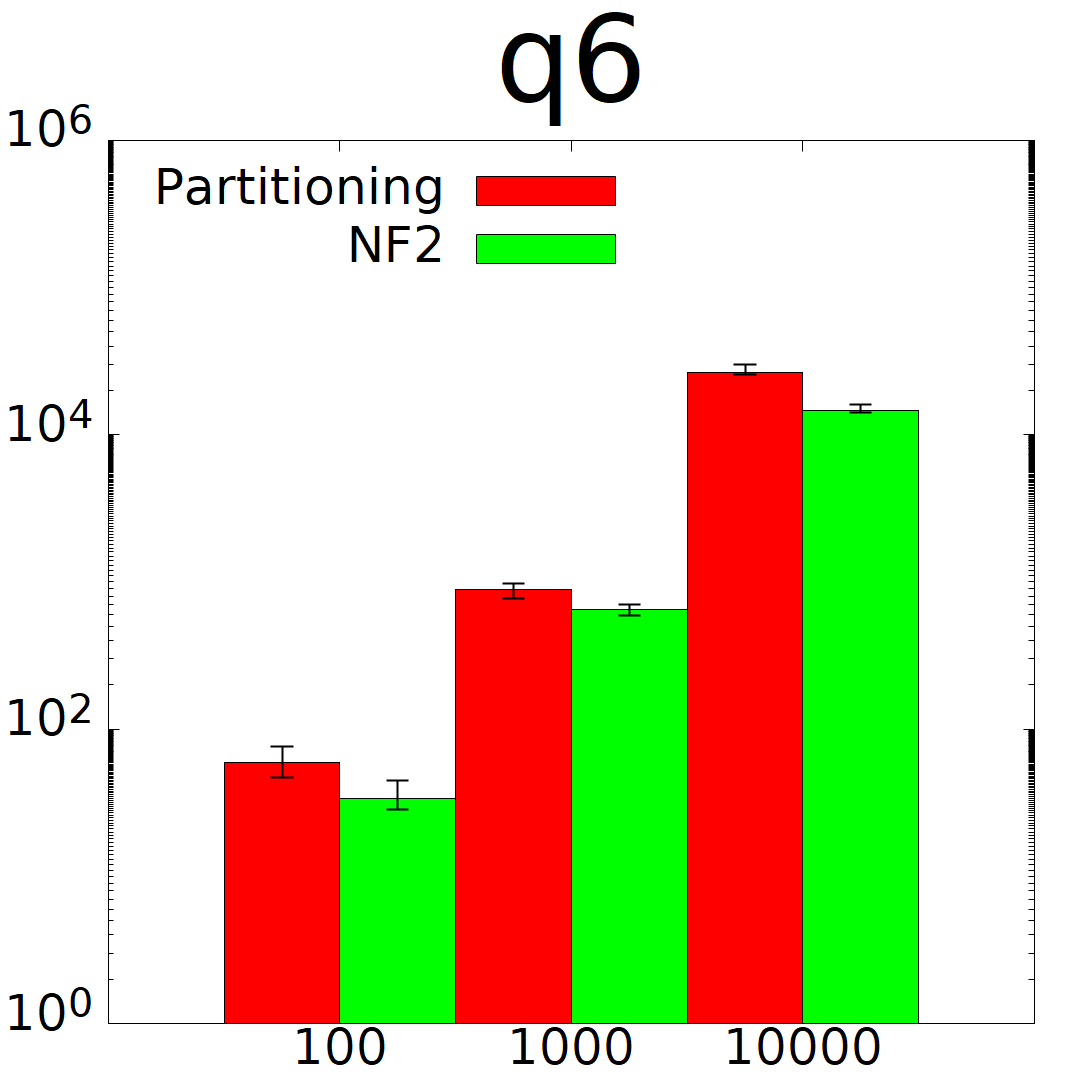}
\includegraphics[scale=0.06]{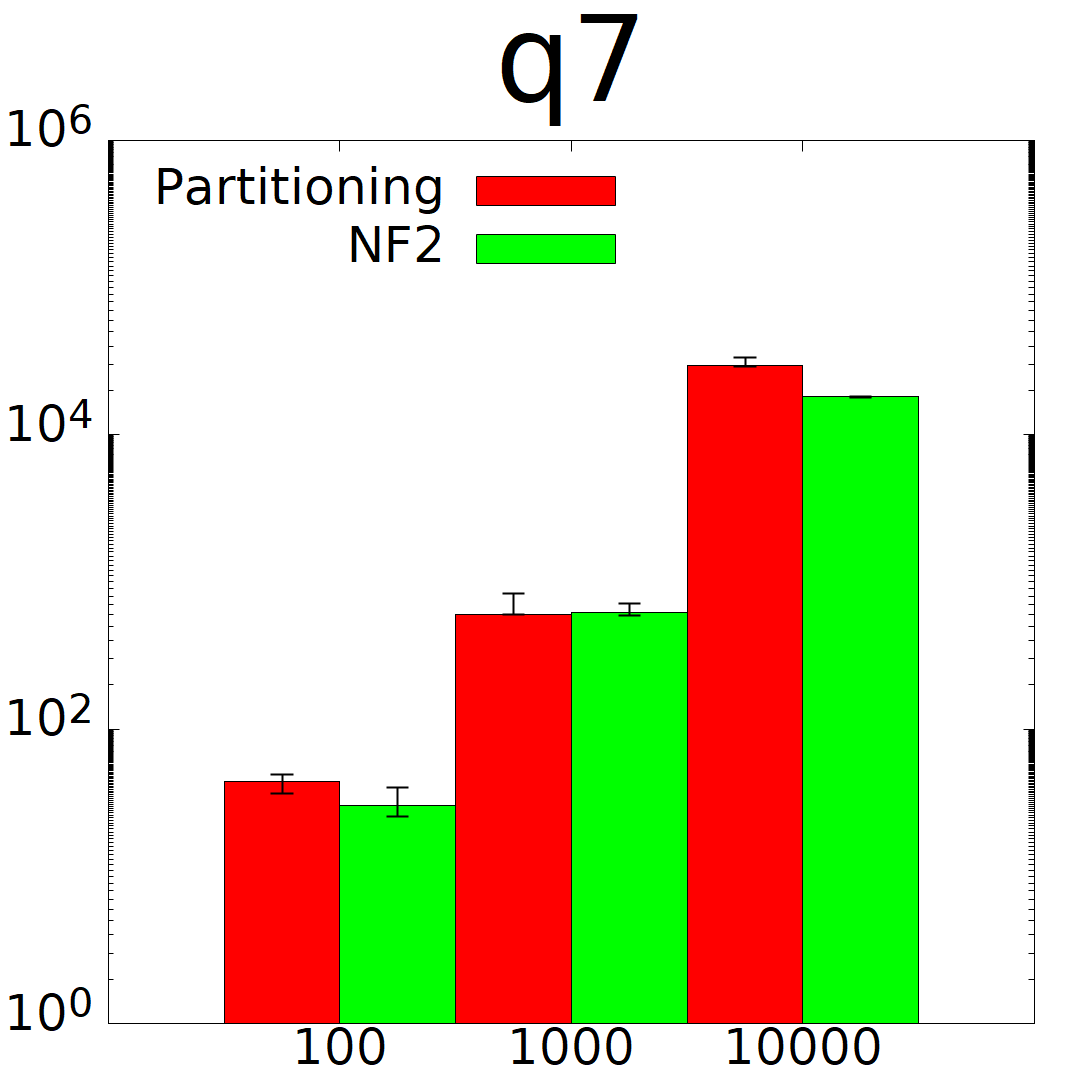}
\caption{Query evaluation performance (milliseconds)}\label{fig:eval-results}
\begin{minipage}[t]{12cm}
        \includegraphics[scale=0.06]{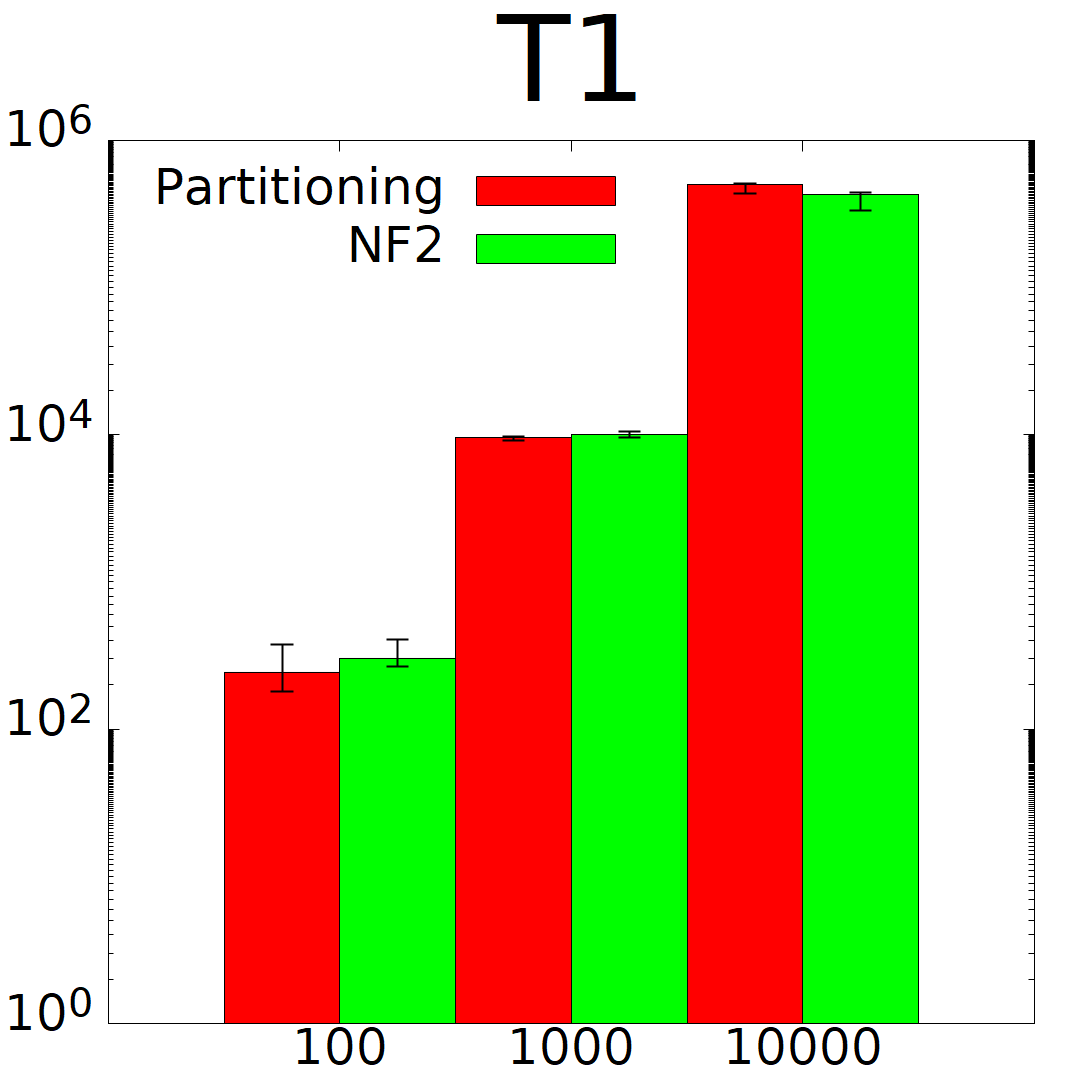}
        \includegraphics[scale=0.06]{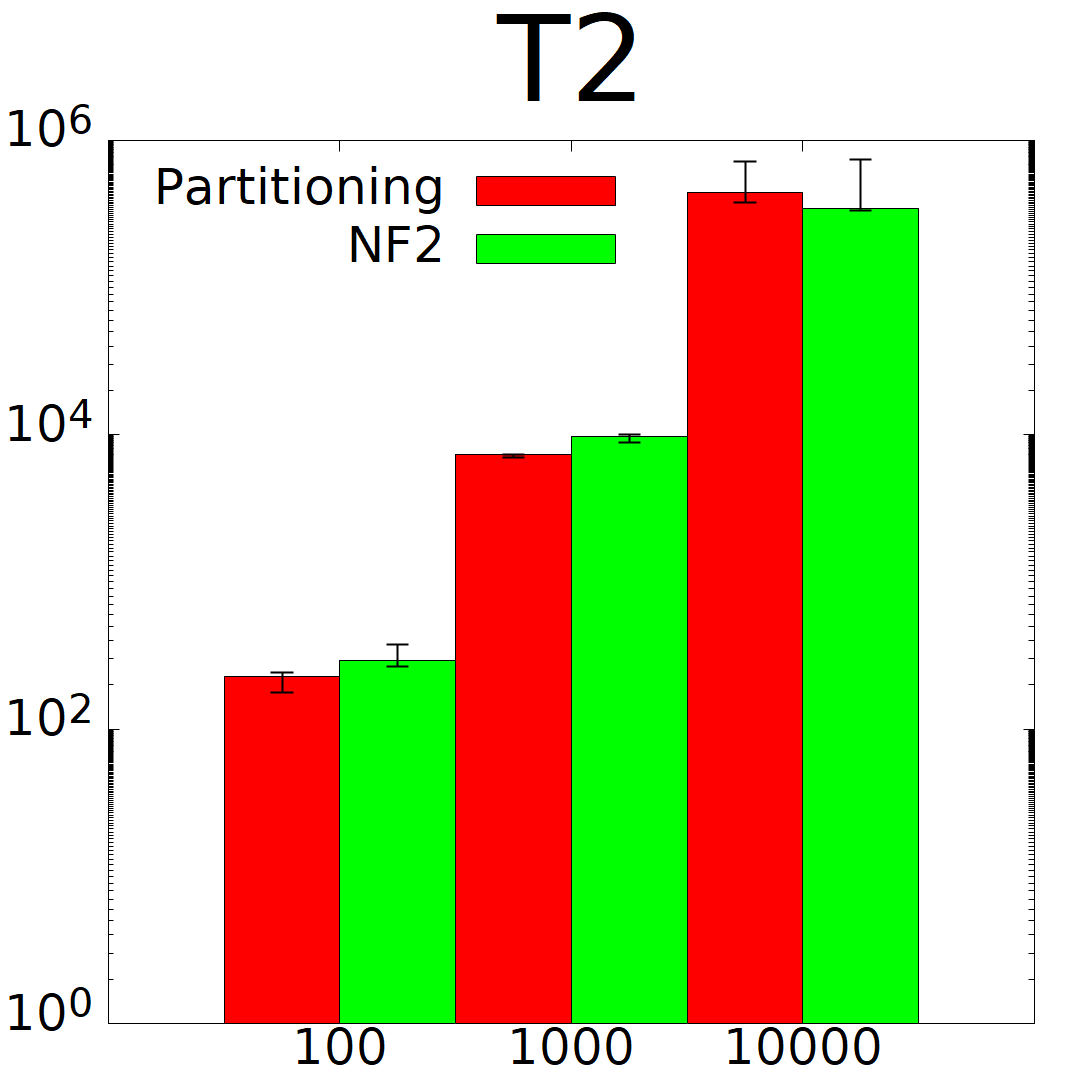}
        \includegraphics[scale=0.06]{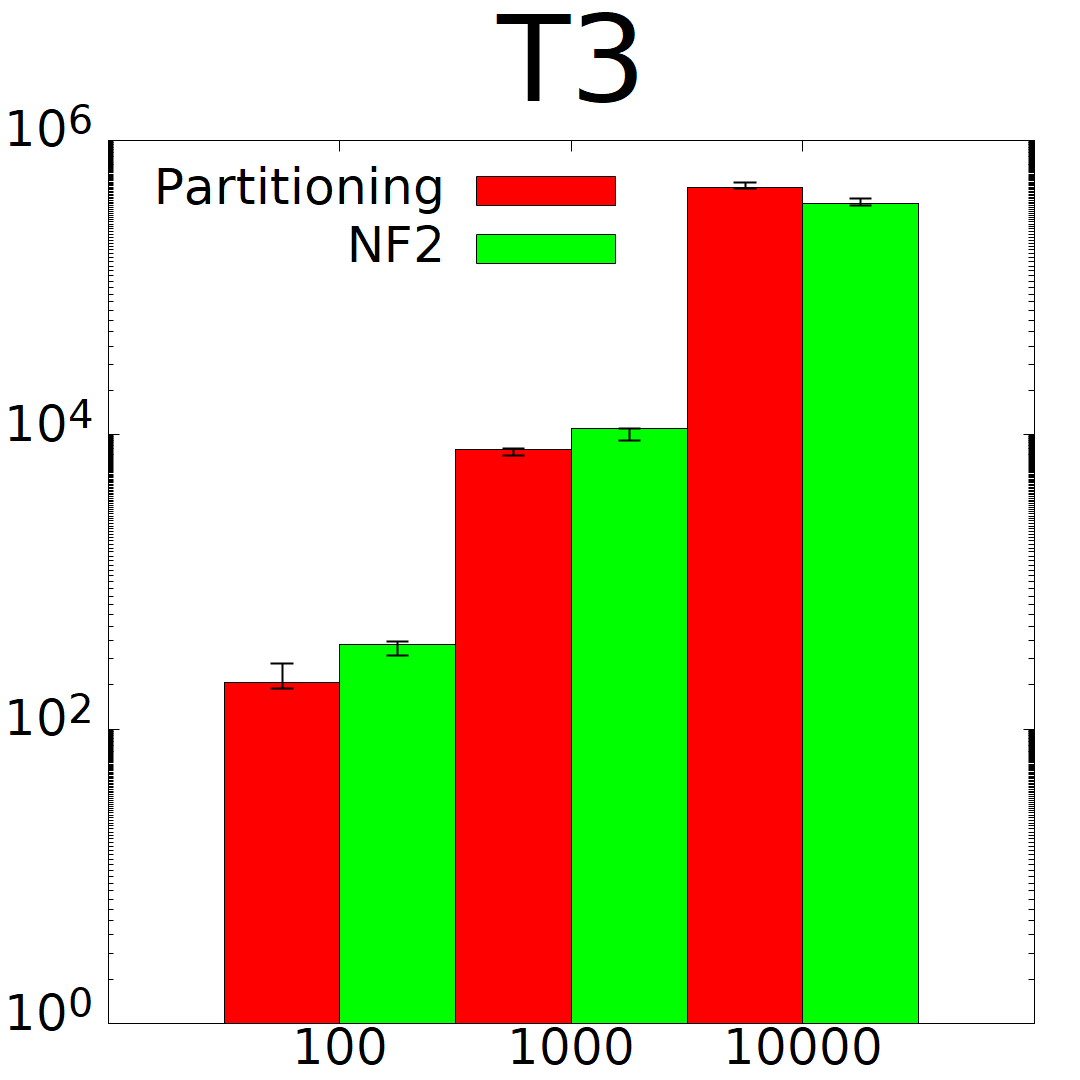}
        \includegraphics[scale=0.06]{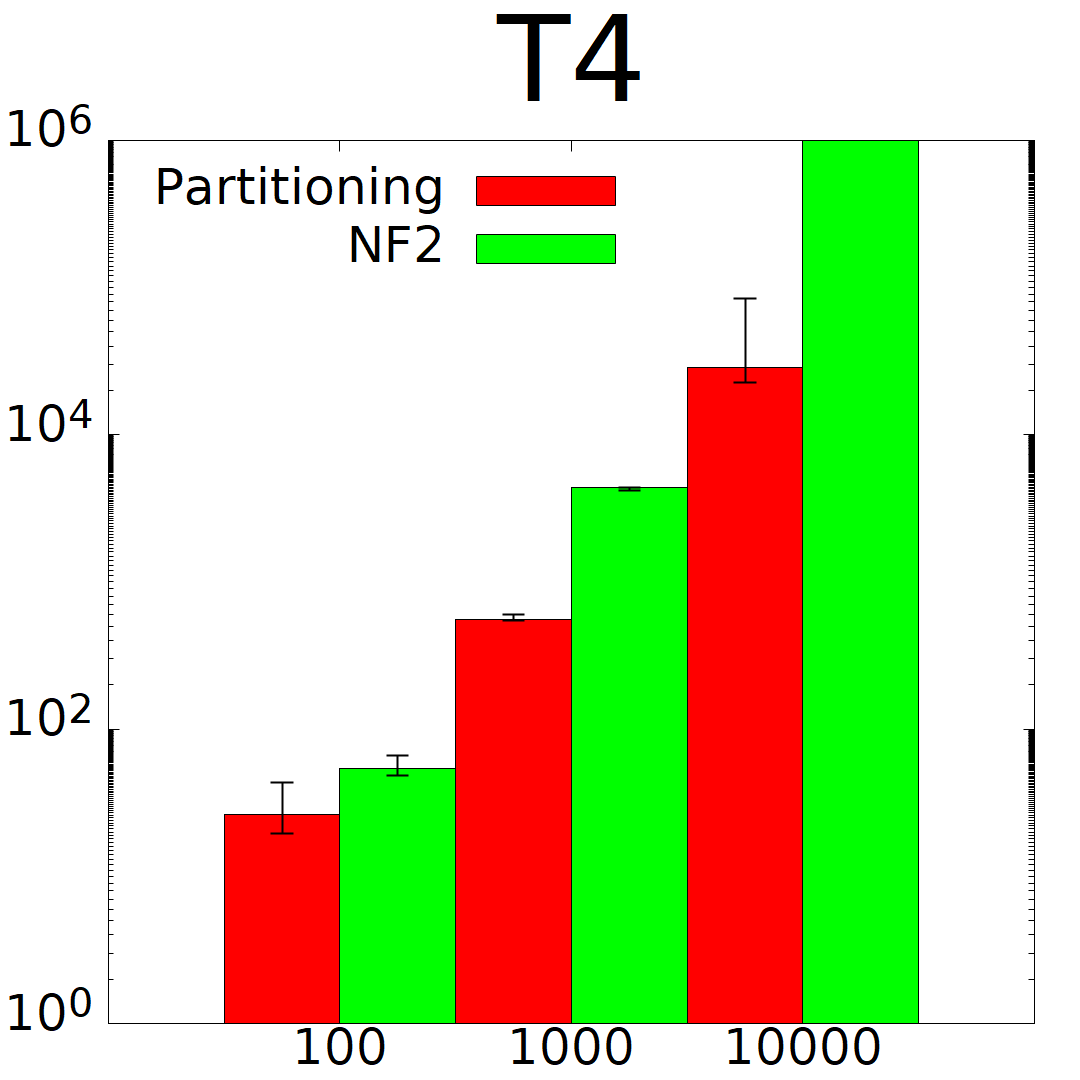}
        \includegraphics[scale=0.06]{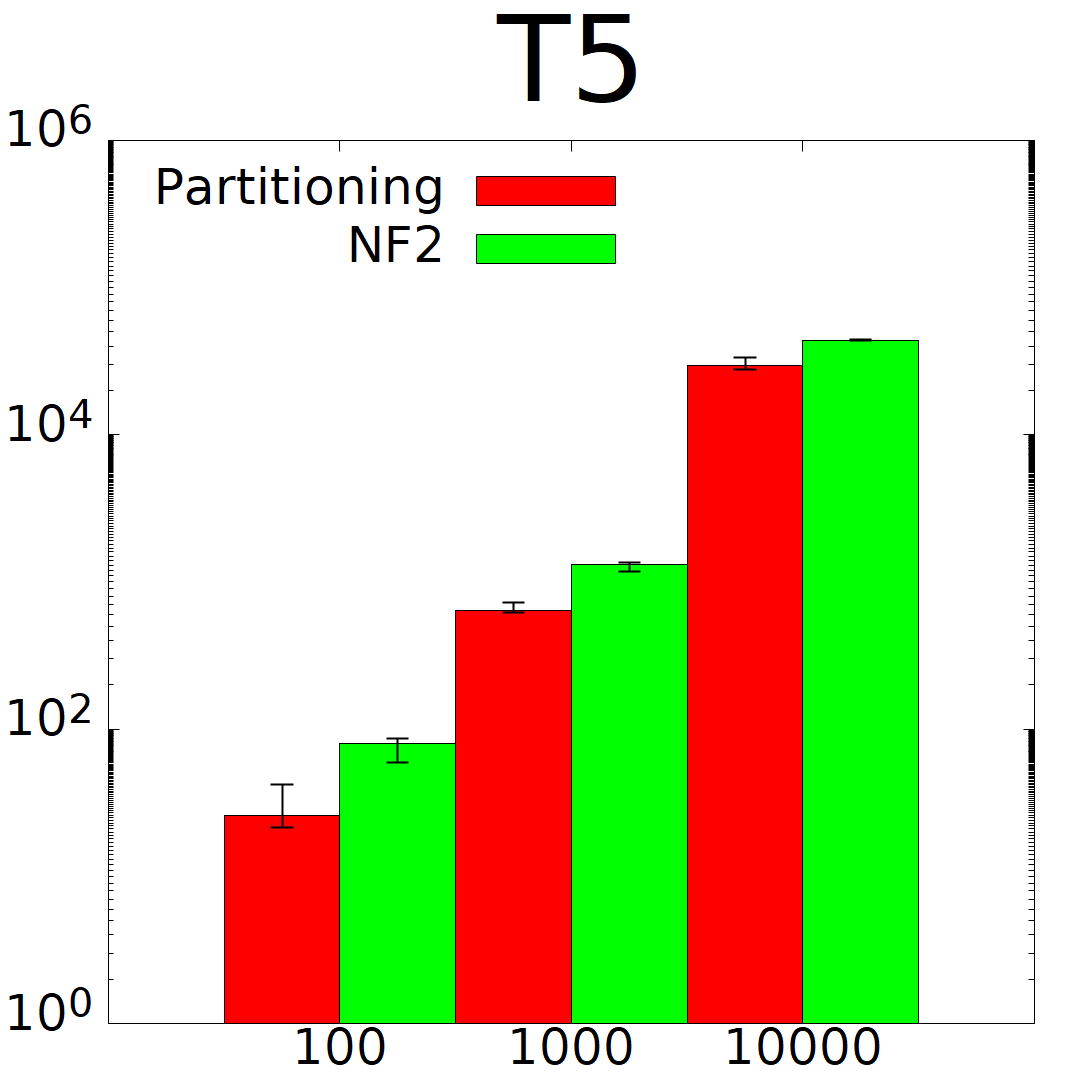}
        \subcaption{Equation generation performance (milliseconds)}
        \label{fig:eqgen-results}
    \end{minipage}
    \begin{minipage}[t]{5cm}
         \includegraphics[scale=0.06]{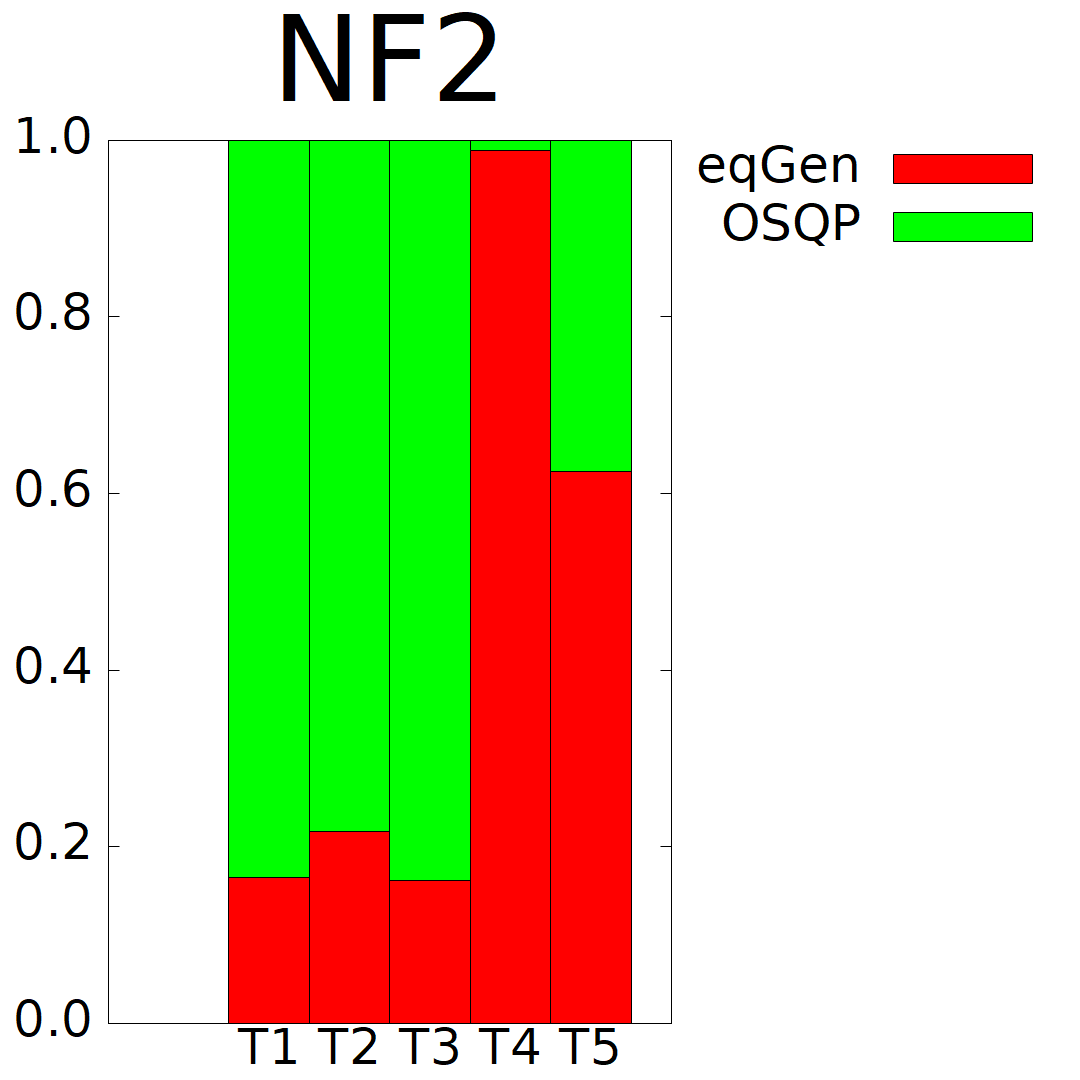}
         \includegraphics[scale=0.06]{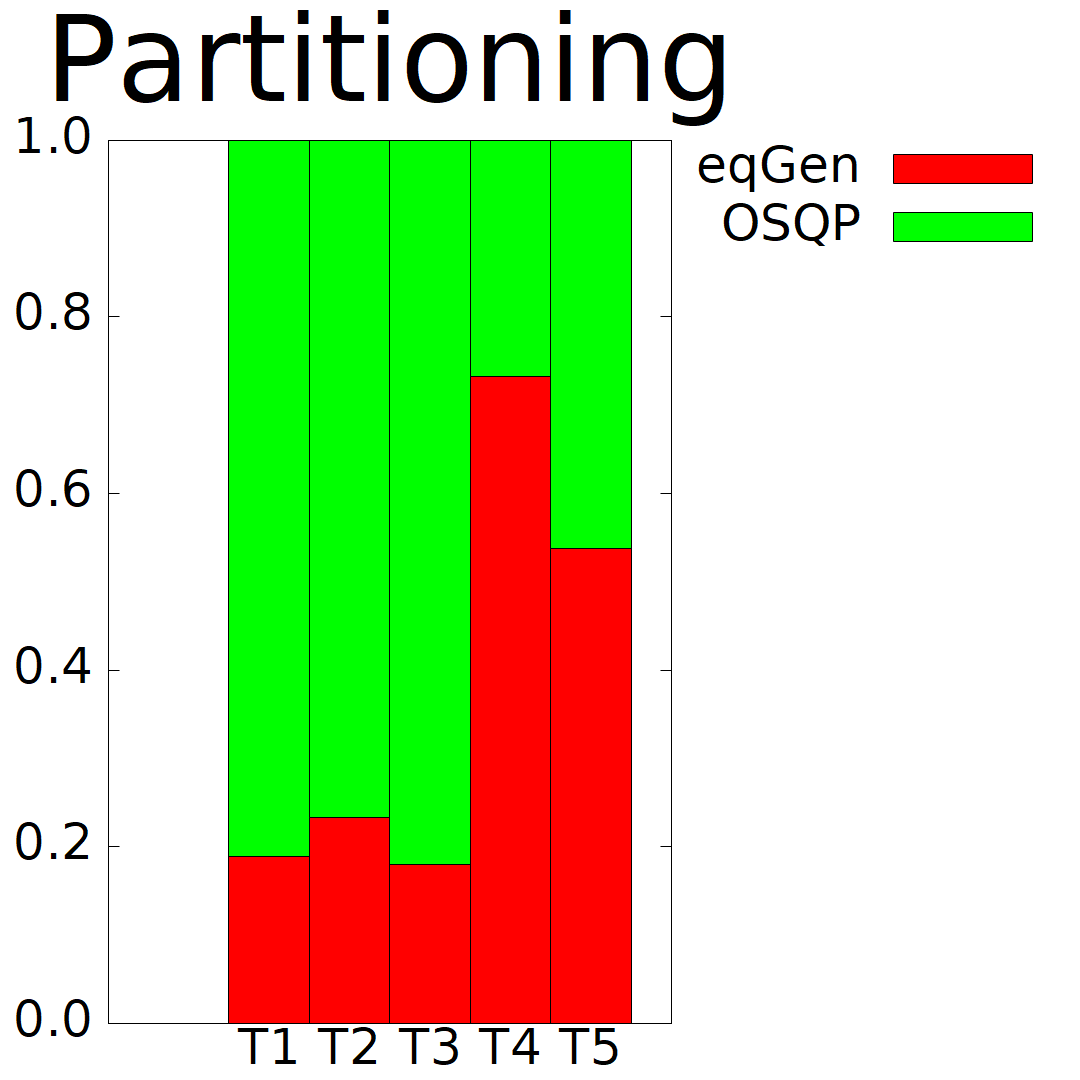}
        \subcaption{Percentages of time spent
  generating equations and solving}
        \label{fig:solve-percentage-results}
    \end{minipage}
 \caption{Equation generation and solving}
    \label{fig:solve-results}
\begin{center}
\includegraphics[scale=0.085]{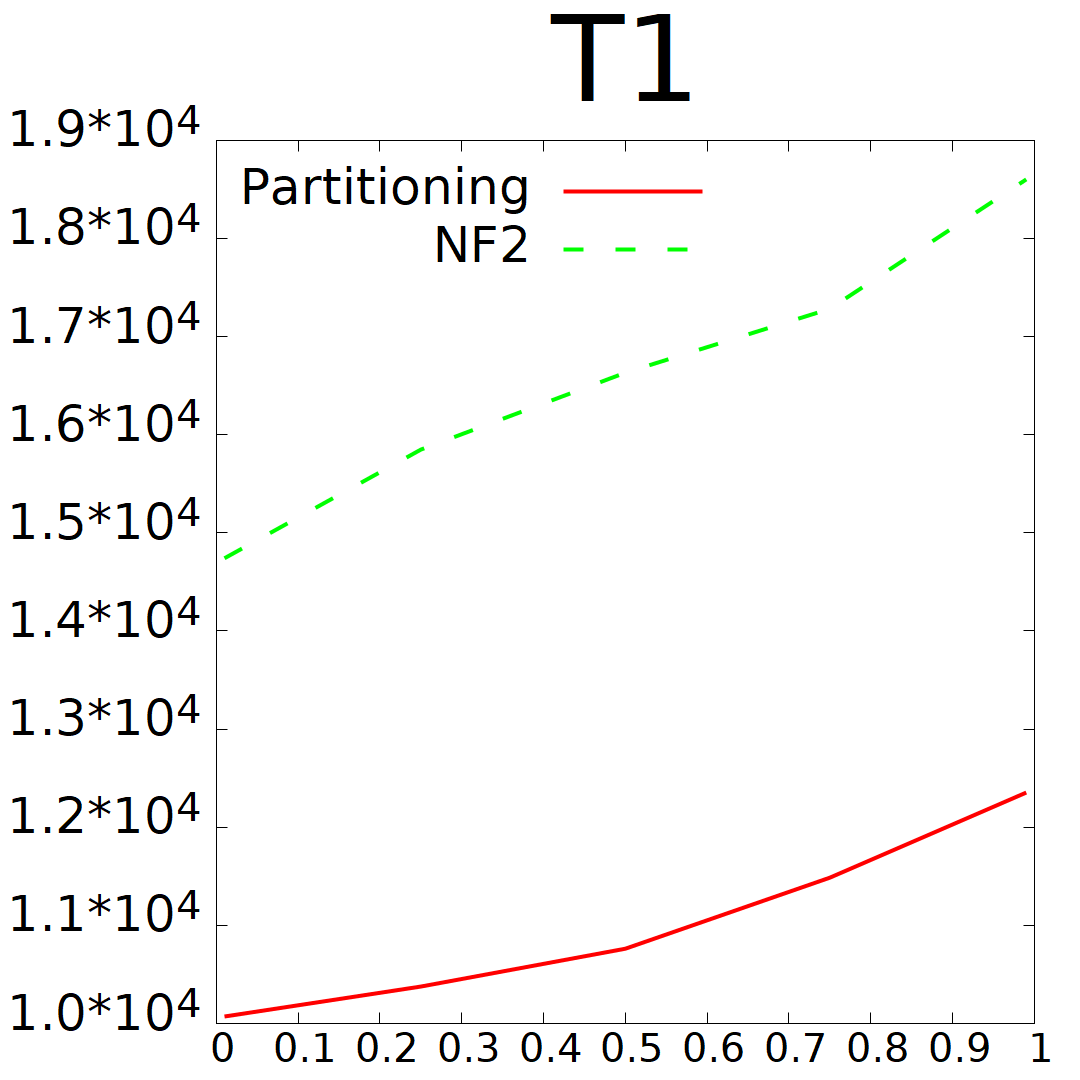}
\includegraphics[scale=0.085]{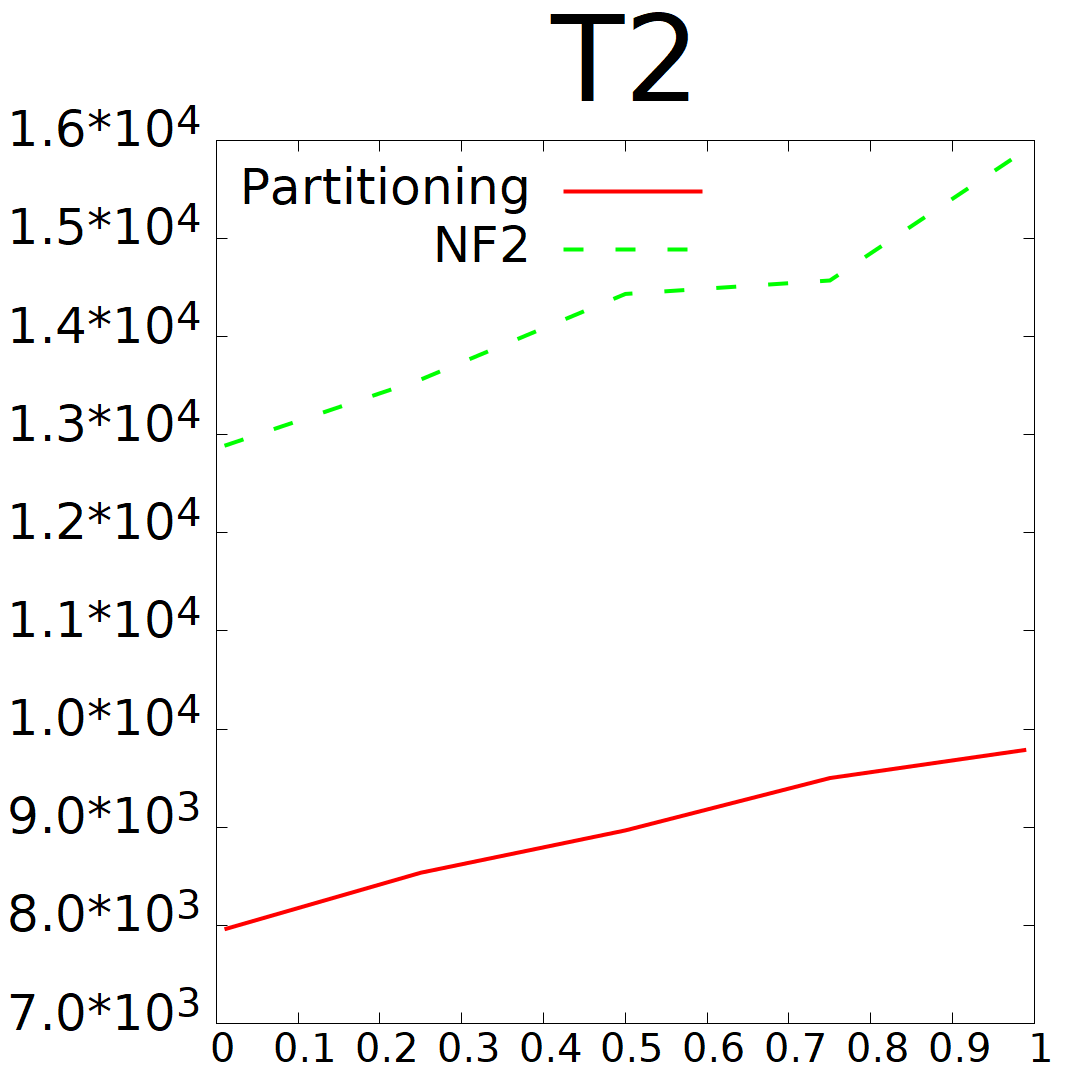}
\includegraphics[scale=0.085]{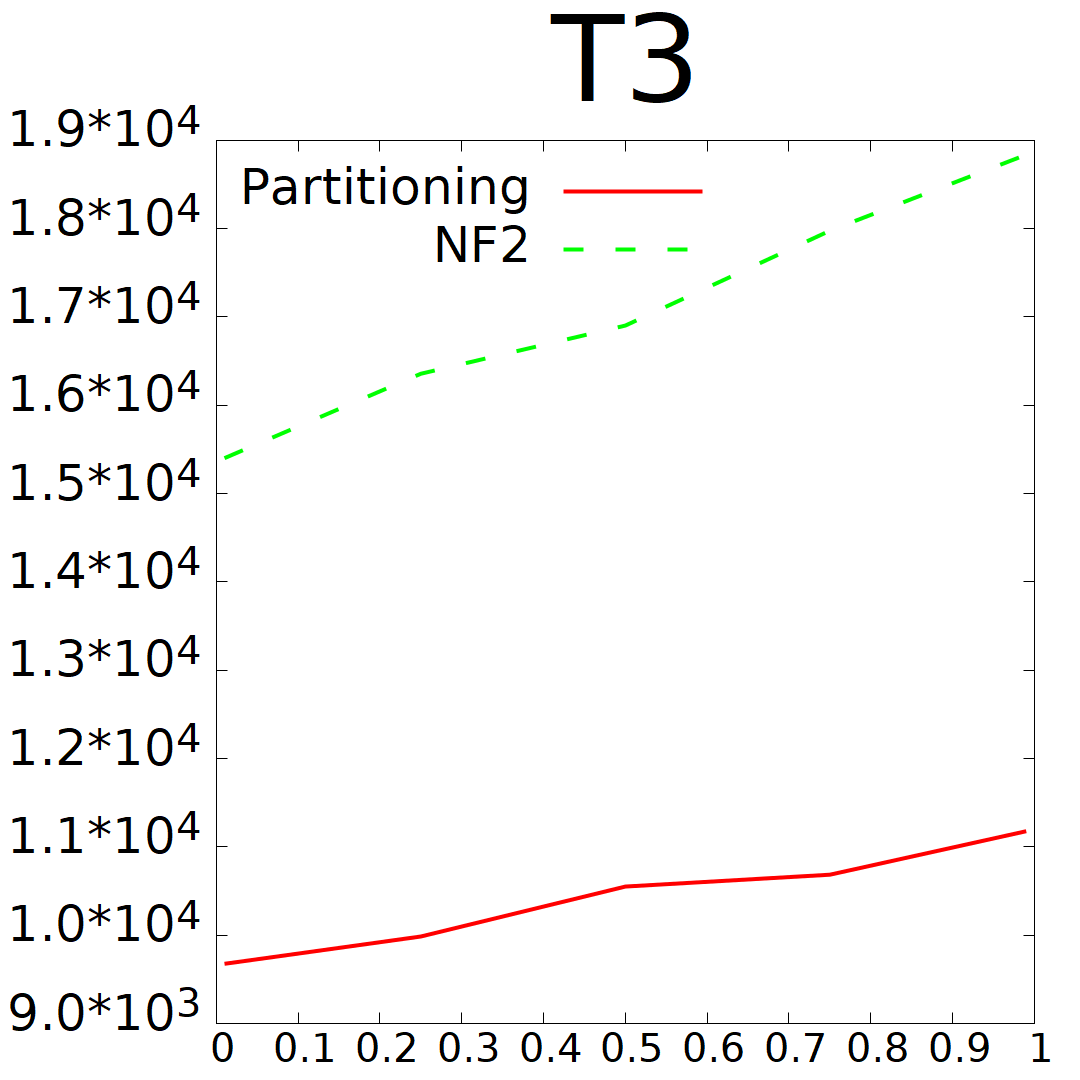}
\includegraphics[scale=0.085]{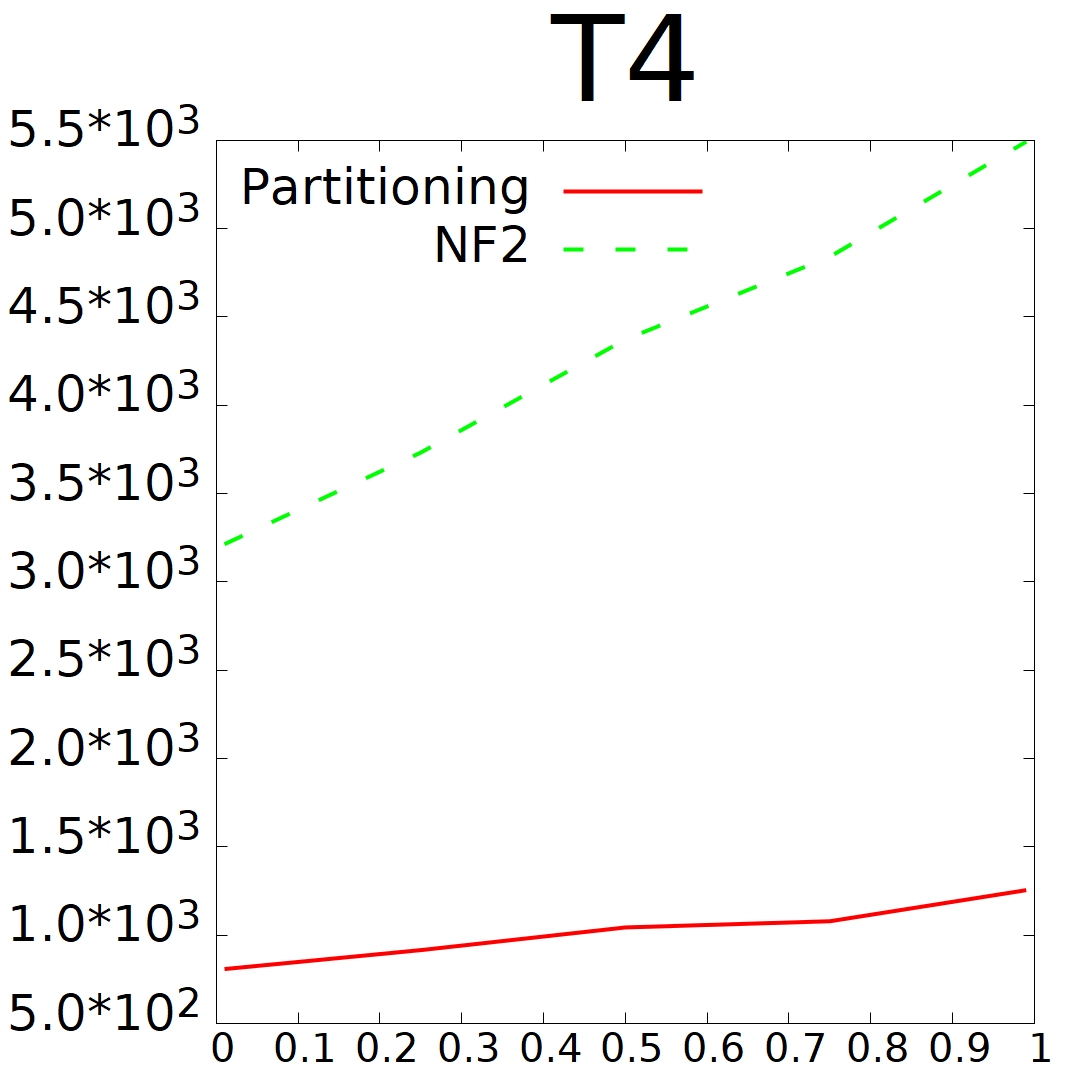}
\includegraphics[scale=0.085]{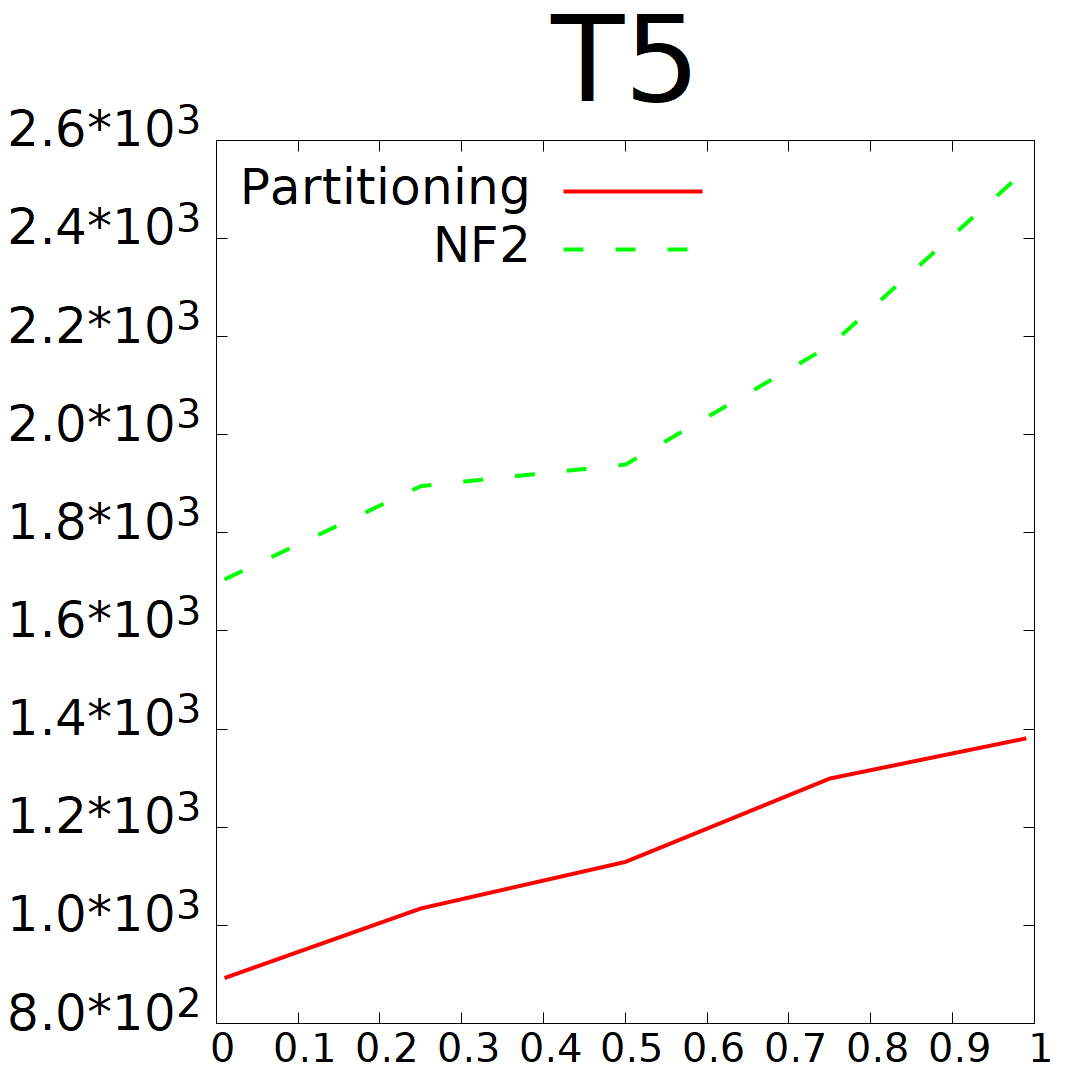}
\end{center}
\caption{Number of variables evaluation impact}\label{fig:eval-variables}
\begin{center}
\includegraphics[scale=0.085]{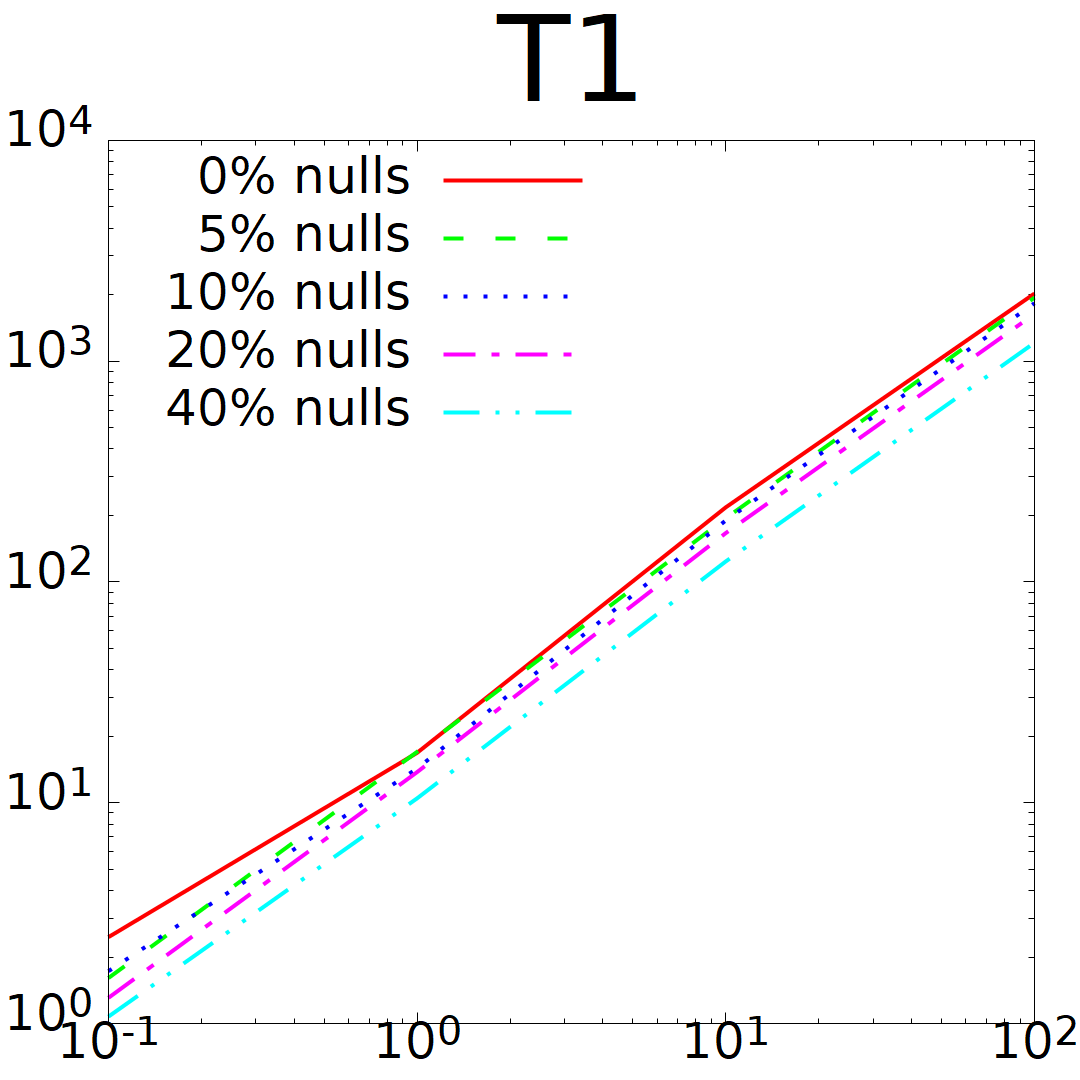}
\includegraphics[scale=0.085]{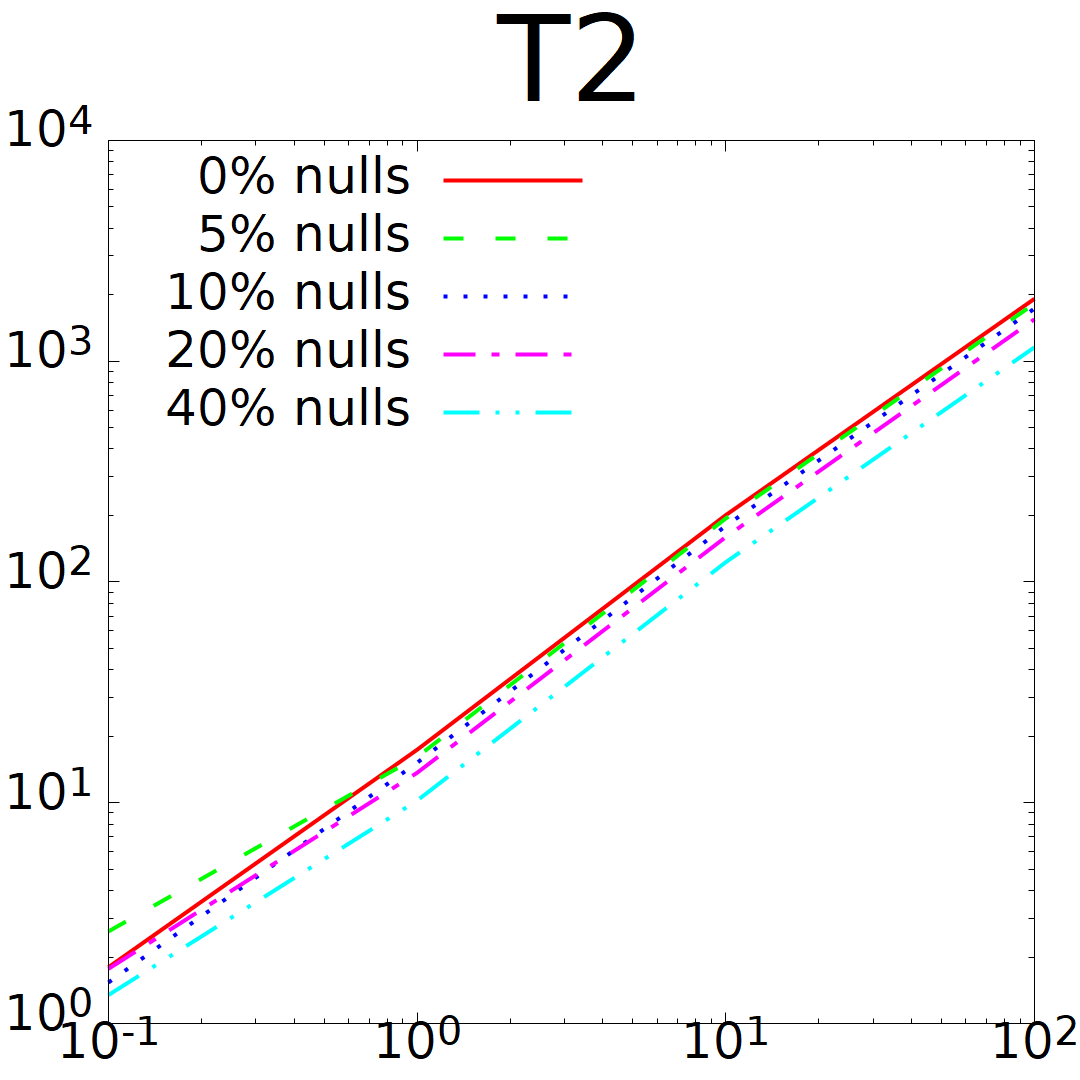}
\includegraphics[scale=0.085]{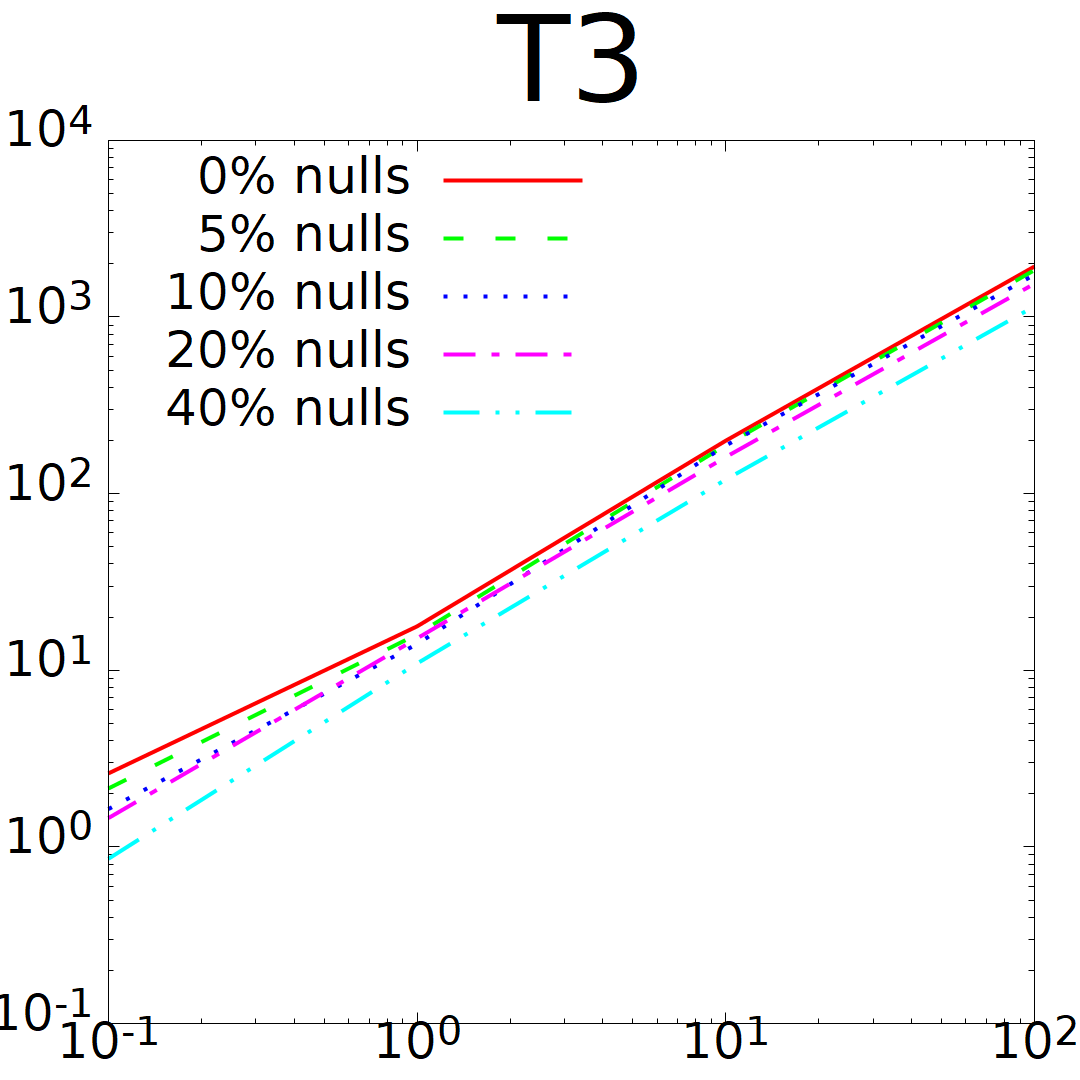}
\includegraphics[scale=0.085]{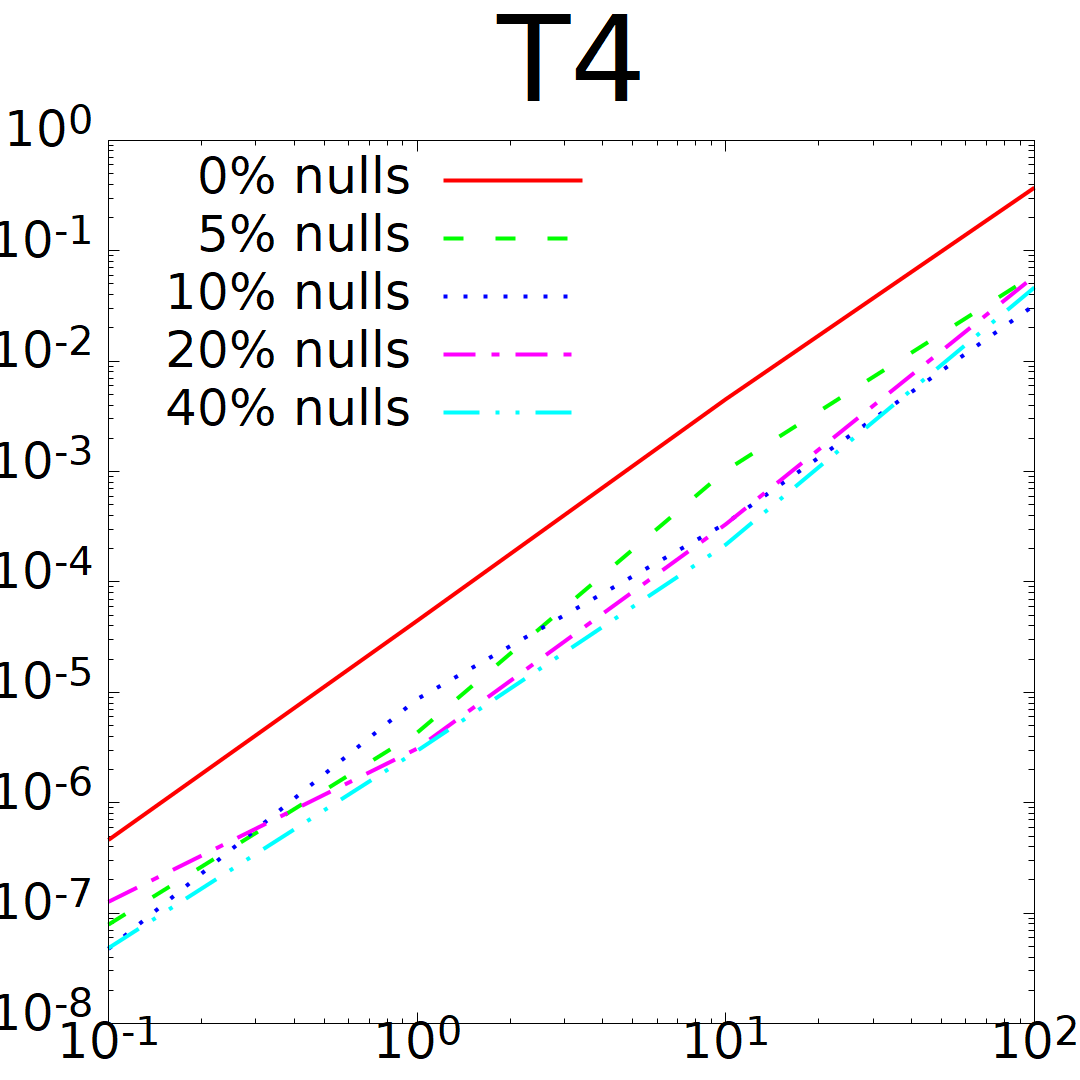}
\includegraphics[scale=0.085]{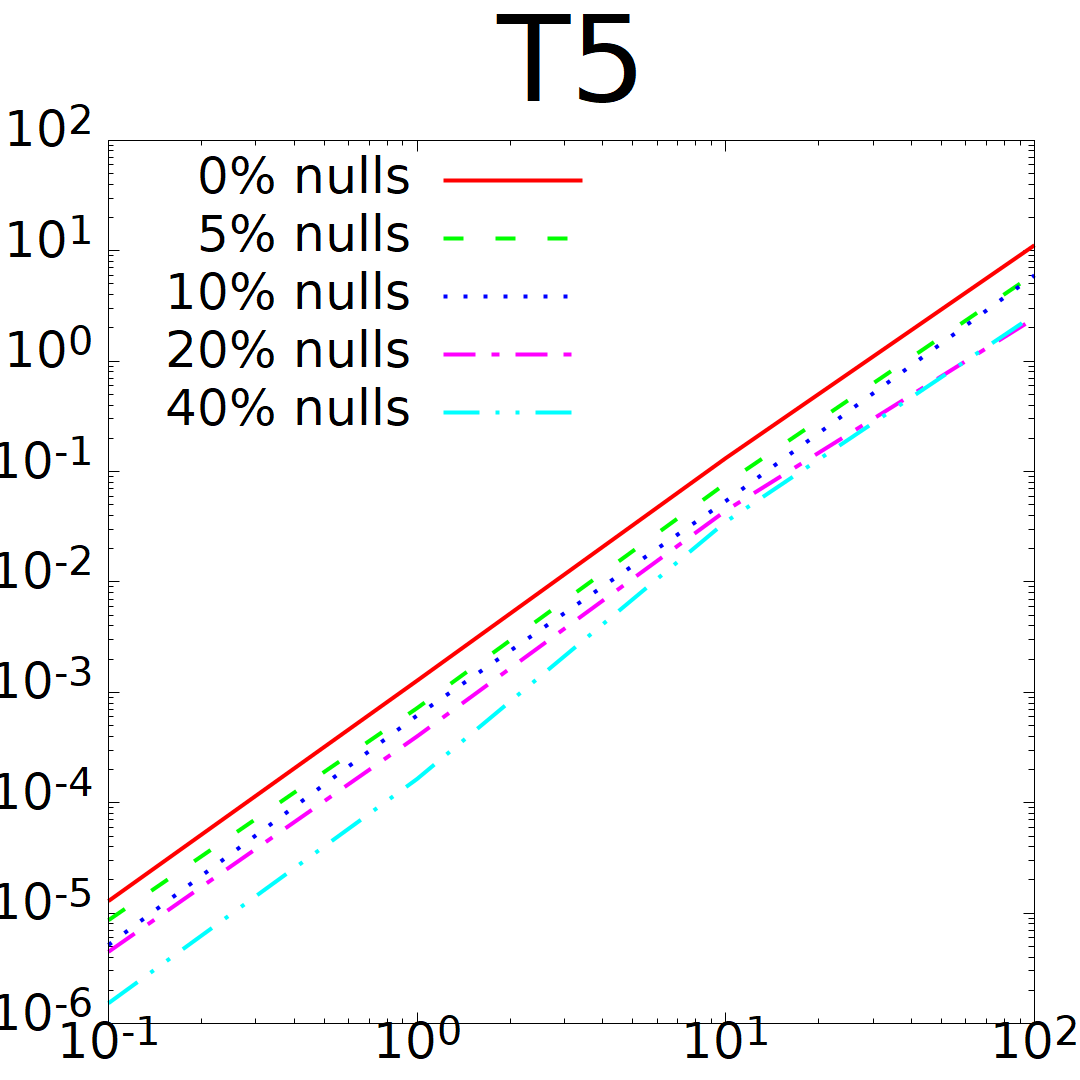}
\end{center}
\caption{Distortion evaluation impact}\label{fig:eval-distortion}
\end{figure*}

We considered the following questions:
\begin{enumerate}[Q1.]
    \item How does the time taken for symbolic query evaluation using \sparsenftwo and \partitioning vary depending on data size?
    \item How does the time taken for equation generation vary depending on data size?
    \item How does the time taken by OSQP for solving compare to that needed for equation generation?
    \item How does overall time taken vary depending on the number of variables?    
    \item How does the measured discordance vary depending on the amount of distortion in the data?
    \end{enumerate}
Q1 and Q2 measure the performance of our system without considering the time taken by OSQP.  Q3 determines whether our system produces QP problems that are feasible for OSQP to solve, because such problems could be encoded in several different ways.  Q4 assesses whether and how performance depends on the amount of source data being symbolic, while Q5 investigates how discordance behaves when data that we know to be consistent is distorted to different degrees.

Although there are several benchmarks for entity resolution and evaluation of the distance between descriptive data, there is not any available benchmark with multiple sources of overlapping numerical data suitable for our system, so we adopted a microbenchmarking approach with synthetic data and simple queries.  We defined a simple schema with tables $R : A,B \tri C,D$ and $S : B \tri E,F$ and a random data generator that populates this schema, for a given parameter $n$, by generating $n$ rows for $S$ and for each such row $(b,e,f)$, generating between 0 and $\sqrt n$ rows for $R$ with the same $B$ field.  Thus on average the resulting database contains $n + \frac{n}{2}\sqrt n$ rows in total.  We generated databases for $n\in\{100;1{,}000;10{,}000\}$; note that $n=10{,}000$ actually corresponds to approximately $510{,}000$ rows.  For each $n$, we performed five trials using five different randomly-generated datasets and took the median running time (or for Q5, median distortion) over these five runs.
We consider the following queries to exercise the most complex cases of the translation of Section~\ref{sec:implementation}:
\[\small
\begin{array}{rcl}
q_1 & = & R' \Join S' \\
q_2 & = & \varepsilon_{W := C + D}(R')\\
q_3 & = & \varepsilon_{X := W*C}(\varepsilon_{W := 1}(R'))\\
q_4 & = & \gamma_{A;C}(R')\\
q_5 & = & \gamma_{B;C}(R')\\
q_6 &=& R \uplus_Z R'\\
q_7 &=& \kappa_Z(q_6)
\end{array}
\begin{array}{rcl}
T_1 &=& \kappa_Z(q_1 \uplus_Z (R \Join S))\\
T_2 & = & \kappa_Z(q_2 \uplus_Z \varepsilon_{W := C + D}(R))\\
T_3 & = & \kappa_Z(q_3 \uplus_Z \varepsilon_{X := W*C}(\varepsilon_{W := 1}(R)))\\
T_4 & = & \kappa_Z(q_4 \uplus_Z \gamma_{A;C}(R))\\
T_5 & = & \kappa_Z(q_5 \uplus_Z \gamma_{B;C}(R))
\end{array}\]

Given two source tables $R, S$, in a database generated as explained above, we create \emph{observation tables} $R_o$, $S_o$ by distorting them as follows: For each row, we randomly replace each value field with NULL with some probability (i.e., $p = 0.01$) and otherwise add a normally-distributed distortion. Next, symbolic views $R', S'$ of both distorted tables are defined, as outlined in Section~\ref{section.symbolicEvaluation}.
Once we have these two versions of the tables (i.e., the source $R$, $S$ and the distorted, symbolic one $R'$, $S'$), we considered two modes of execution of these queries: in the first mode ($q_1\dots q_7$), we simply evaluate the query over a symbolic input (i.e., $R'$ and $S'$) and construct the result; in the second mode ($T_1\dots T_5$), we evaluate the result of aligning the distorted query result with the result over the original source tables (i.e., $R$ and $S$).  Thus, for example, for $T_1$ we generate the equations resulting from the fusion expression $q_1 \sqcup (R \Join S)$, actually implemented like $\kappa_Z(q_1 \uplus_Z (R \Join S))$.
Finally, the resulting system of equations is solved, subject to the metric giving each error variable $x$ a weight of $x^2$ and each null variable a weight of 0.

For Q1, executions are summarized in Fig.~\ref{fig:eval-results}, where reported times include the time to receive the symbolic query results.
These show that the \partitioning and \sparsenftwo have broadly similar performance; despite \sparsenftwo's comparative simplicity, its running time is often faster with the exceptions being $q_4$ and $q_5$, the two aggregation queries.  Particularly for $q_4$, aggregation can result in large symbolic expressions which are not always handled efficiently by the \sparsenftwo sparse vector operations using PostgreSQL arrays; we experimented with several alternative approaches to try to improve performance without success.  Thus, in cases where the symbolic expressions do not grow large, \sparsenftwo seems preferable.

For Q2 and Q3, we measured the time taken for equation generation and for OSQP solving for each query, using different database sizes as described above.  The results are shown in Fig.~\ref{fig:solve-results}.  In Fig.~\ref{fig:eqgen-results}, the time taken for equation generation, including querying and serializing the resulting OSQP problem instances, is shown (again in logarithmic scale).  The OSQP solving times for \partitioning and \sparsenftwo are coincident and so not shown.  In  Fig.~\ref{fig:solve-percentage-results}, the percentage of time spent on equation generation and on OSQP solving for the largest database instance ($n=10{,}000$) is shown, and we can appreciate that they are always in a similar order of magnitude so neither can be claimed to be a bottleneck in front of the other.  

For Q4, we considered a fixed database size (n=1000) and modified the data generation process and specifications so that for each input table, each row was treated as symbolic with some probability $p_{sym}$.  We considered $p_{sym} \in \{0.01, 0.25, 0.5, 0.75, 0.99\}$.  Only the values in these symbolic rows were augmented with variables and only these rows were distorted.  We reran the evaluation for Q2 and Q3 to compute the total time in each case, for both encodings, in order to assess how the performance varies as the number of variables/symbolic fields in the input increases.  Figure~\ref{fig:eval-variables} shows the results, in each case reporting the median time observed out of five runs.  For both \partitioning and \sparsenftwo strategies, the total time increases roughly linearly.  We further inspected the results for equation generation and solving time and found that generally the solving times for problems generated by \partitioning and \sparsenftwo were close to each other, thus the difference in performance (especially in the case of $T_4$) is mostly due to difference in query evaluation times for equation generation, in line with the general trends noticed in Figure~\ref{fig:eqgen-results}. 
Thus, the scalability of the approach is not compromised by our addition to the model, but follows the expected behaviour of regular ground queries (without variables).

For Q5 we again considered a fixed database size (n=1000) and separately varied the probability of replacing a value with null ($p_{null}$) and the standard deviation of the normally-distributed noise ($\sigma$).  We would expect increasing the number of NULLs to decrease the discordance (all else being equal) because null variables carry no weight, while increasing the standard deviation of the distortion should increase the discordance.  We considered $p_{null} \in \{0, 0.05, 0.1, 0.2, 0.4\}$ and $\sigma \in \{0.1, 1, 10, 100\}$.  For each combination of parameters we evaluated five randomly-generated inputs and computed the distance found by OSQP, taking the median discordance in each case.  The results are shown in Figure~\ref{fig:eval-distortion}.  We report only once the results obtained, because the discordance value found does not depend on the implementation strategy. These results confirm that increasing the amount of distortion ($\sigma$) generally increases the discordance, while increasing the number of NULLs ($p_{null}$) tends to decrease discordance (because it introduces degrees of freedom to the problem that do not incur any penalty in the cost metric).

\subsection{Case study}\label{section.CaseStudy}

We might use our tool to get a best-fit database.
However, this would only be useful if sources are close to each other (and hence to reality).
If they are relatively discordant (like the blind men describing the elephant), all we can aim at is to measure and study the evolution of such discordancy.
Thus, we applied our prototype to the study of challenging COVID-19 data, which is publicly available, and see from that the improvement of reporting in different countries during the pandemic.
More specifically, we considered two different sources:
\begin{description}
\item[Johns Hopkins University (JHU)] The Center for Systems Science and Engineering (CSSE) at JHU was gathering COVID-19 data since the very beginning of the pandemic and became a referent worldwide \cite{jhu}. On the one hand, we have used its daily time series at country level\footnote{\url{https://github.com/CSSEGISandData/COVID-19/tree/master/csse\_covid\_19\_data/csse\_covid\_19\_time\_series}} containing both cases and deaths. Unfortunately, on the other hand, regional data is scattered in different files in the JHU repository, so we used a more compact version.\footnote{\url{https://github.com/coviddata/coviddata}}
\item[EuroStats]
As second data source for comparison, we used the weekly European mortality by EuroStats,\footnote{\url{https://ec.europa.eu/eurostat/databrowser/view/demo\_r\_mwk2\_ts/default/table}} following the Nomenclature of Territorial Units for Statistics (NUTS).\footnote{\url{https://ec.europa.eu/eurostat/web/nuts/background}}
\end{description}

JHU was going through a continuous consolidation and cleaning process, but still resulted in quite poor quality. Obviously, EuroStats data are of much higher quality and more reliable. Indeed, the weekly mortality per country appears to be historically quite stable (less than 5.5\% coefficient of variation for the six countries of our study).
Hence, we took the weekly mortality of the five years previous to the pandemic as ground truth.
However, for some countries, most recent figures were either tagged as provisional or estimated.
While we considered the former to be an administrative issue and still part of the error-free ground truth, we put the latter together with the mortality of 2020/2021 in an \emph{s-table}, and treated those data in the same way as the ones coming from JHU.

\begin{table} \tiny
    \begin{center}
    \begin{tabular}{|l|r|r|r|l|l|}\hline
        \multicolumn{1}{|c}{\textbf{Table}} & \multicolumn{1}{|c}{\textbf{Loc.}} & \multicolumn{1}{|c}{\textbf{Times}} & \multicolumn{1}{|c}{\textbf{Rows}} & \multicolumn{1}{|c}{\textbf{First}} & \multicolumn{1}{|c|}{\textbf{Last}} \\ \hline \hline
        EU(\underline{r,w},\#d) & 222 & 1,043 & 152,938 & 2000W01 & 2019W52 \\ \hline
        EUe(\underline{r,w},\#d) & 222 & 73 & 15,001 & 2000W01 & 2021W20 \\ \hline
        EU(\underline{c,w},\#d) & 33 & 1,043 & 26,177 & 2000W01 & 2019W52 \\ \hline
        EUe(\underline{c,w},\#d) & 34 & 1,116 & 4,125 & 2000W01 & 2021W20 \\ \hline
        JHU(\underline{c,d},\#c,\#d) & 197 & 479 & 94,363 & 20200119 & 20210521 \\ \hline
        JHU(\underline{r,d},\#c,\#d) & 550 & 385 & 211,365 & 20200129 & 20210216 \\
        \hline
    \end{tabular}
    \end{center}
    \caption{Summary of the tables in the experiment}\label{table:coviddescriptive}
\end{table}

We loaded the different data in a PostgreSQL data\-base with Pentaho Data Integration.
These were divided in the six tables shown in Table~\ref{table:coviddescriptive}, together with the counters of different locations and times, number of rows, and first and last time point available.
Data was split firstly according to the source (namely EuroStats or JHU).
Ground truth mortality (i.e., until the end of 2019 and free of errors) is in \emph{ground} tables $EU$, while estimates and data of 2020/2021 are in \emph{s-tables} $\mathit{EUe}$.
Different \emph{s-tables} are also generated for different geographic granularities (namely region $r$ or country $c$), and relevantly, data from Eurostats is available per week $w$, while data from JHU is available daily $d$.
Both location and temporal dimensions result in different (underlined) key attributes for the corresponding tables.
From EuroStats, we only used the number of deaths $\#d$, while from JHU we took both COVID-19 cases $\#c$ and deaths $\#d$.
Attribute $\#d$ is declared as free of variables in both $EU$ \emph{ground} tables and its instances are consequently constants.
Values coming from EuroStats correspond exactly to the reported ones, but to mitigate the noise (e.g., cases not reported during weekends being moved to the next week by some regions) in those coming from JHU, we followed the common practice of taking the average in the previous seven days for both cases and deaths.

\begin{figure*}[tb]
	\noindent\begin{minipage}{\linewidth}
	    \small\centering
		\begin{tikzpicture}[
		vertex/.style = {shape=rectangle,rounded corners,draw},
		edge/.style={-> = latex'},
		node distance=0.25cm and .85cm
		]
    \node (EUrw) [vertex] {EU(\underline{r,w},\#d)};
    \node (EUrwe) [vertex, below=of EUrw] {EUe(\underline{r,w},\#d)};
    \node (union1) [vertex, right=of EUrw] {$\cup$};
        \draw[edge] (EUrw) -- (union1);
        \draw[edge] (EUrwe) -- (union1);
    \node (sel1) [vertex, right=of union1] {$\sigma_{w>=2020W06\ and\ w<=2021W06}$};
        \draw[edge] (union1) -- (sel1);
    \node (sel2) [vertex, below=of sel1] {$\sigma_{y>2014\ and\ y<2020}$};
        \draw[edge] (union1) -- (sel2);
    \node (gb1) [vertex, right=of sel2] {$\gamma_{r,woy;avg(\#d)}$};
        \draw[edge] (sel2) -- (gb1);
    \node (join1) [vertex, above=of gb1] {$\Join$};
        \draw[edge] (sel1) -- (join1);
        \draw[edge] (gb1) -- (join1);
    \node (expr1) [vertex, right=of join1] {$\varepsilon_{s:=\#d-avg(\#d)}$};
        \draw[edge] (join1) -- (expr1);
    \node (gb1bis) [vertex, right=of expr1] {$\gamma_{c,w;sum(s)}$};
        \draw[edge] (expr1) -- (gb1bis);
    \node (rename0) [vertex, below=of gb1bis] {$\rho_{sum(s)\mapsto s}$};
        \draw[edge] (gb1bis) -- (rename0);
    \node (EUcw) [vertex, below=of EUrwe] {EU(\underline{c,w},\#d)};
    \node (EUcwe) [vertex, below=of EUcw] {EUe(\underline{c,w},\#d)};
    \node (union2) [vertex, right=of EUcw] {$\cup$};
        \draw[edge] (EUcw) -- (union2);
        \draw[edge] (EUcwe) -- (union2);
    \node (sel3) [vertex, right=of union2] {$\sigma_{w>=2020W06\ and\ w<=2021W06}$};
        \draw[edge] (union2) -- (sel3);
    \node (sel4) [vertex, below right=of union2] {$\sigma_{y>2014\ and\ y<2020}$};
        \draw[edge] (union2) -- (sel4);
    \node (gb2) [vertex, right=of sel4] {$\gamma_{c,woy;avg(\#d)}$};
        \draw[edge] (sel4) -- (gb2);
    \node (join2) [vertex, above=of gb2] {$\Join$};
        \draw[edge] (sel3) -- (join2);
        \draw[edge] (gb2) -- (join2);
    \node (expr2) [vertex, right=of join2] {$\varepsilon_{s:=\#d-avg(\#d)}$};
        \draw[edge] (join2) -- (expr2);
    \node (project2) [vertex, right=of expr2] {$\pi_{s}$};
        \draw[edge] (expr2) -- (project2);
    \node (dunion1) [vertex, right=of project2] {$\uplus_{z}$};
        \draw[edge] (rename0) -- (dunion1);
        \draw[edge] (project2) -- (dunion1);
    \node (JHUcd) [vertex, below=of EUcwe] {JHU(\underline{c,d},\#c,\#d)};
    \node (JHUrd) [vertex, node distance=0.7, below=of JHUcd] {JHU(\underline{r,d},\#c,\#d)};
    \node (gb4) [vertex, above right=of JHUrd] {$\gamma_{c,d;sum(\#d)}$};
        \draw[edge] (JHUrd) -- (gb4);
    \node (rename5) [vertex, right=of gb4] {$\rho_{sum(\#d)\mapsto \#d}$};
        \draw[edge] (gb4) -- (rename5);
    \node (project2) [vertex, right=of JHUcd] {$\pi_{\#d}$};
        \draw[edge] (JHUcd) -- (project2);
    \node (dunion3) [vertex, right=of project2] {$\uplus_{z}$};
        \draw[edge] (project2) -- (dunion3);
        \draw[edge] (rename5) -- (dunion3);
    \node (coal3) [vertex, right=of dunion3] {$\kappa_{z}$};
        \draw[edge] (dunion3) -- (coal3);
    \node (gb3) [vertex, right=of coal3] {$\gamma_{c,w;sum(\#d)}$};
        \draw[edge] (coal3) -- (gb3);
    \node (rename1) [vertex, right=of gb3] {$\rho_{sum(\#d)\mapsto s}$};
        \draw[edge] (gb3) -- (rename1);
    \node (gb5) [vertex, right=of JHUrd] {$\gamma_{r,w;sum(\#d)}$};
        \draw[edge] (JHUrd) -- (gb5);
    \node (gb6) [vertex, below=of gb5] {$\gamma_{r,w+3;sum(\#c)}$};
        \draw[edge] (JHUrd) -- (gb6);
    \node (rename2) [vertex, right=of gb5] {$\rho_{sum(\#d)\mapsto s}$};
        \draw[edge] (gb5) -- (rename2);
    \node (expr3) [vertex, right=of gb6] {$\varepsilon_{s:=1.5\%*sum(\#c)}$};
        \draw[edge] (gb6) -- (expr3);
    \node (projection3) [vertex, right=of expr3] {$\pi_{s}$};
        \draw[edge] (expr3) -- (projection3);
    \node (dunion4) [vertex, right=of rename2] {$\uplus_{z}$};
        \draw[edge] (rename2) -- (dunion4);
        \draw[edge] (projection3) -- (dunion4);
    \node (coal4) [vertex, right=of dunion4] {$\kappa_{z}$};
        \draw[edge] (dunion4) -- (coal4);
    \node (gb7) [vertex, right=of coal4] {$\gamma_{c,w;sum(s)}$};
        \draw[edge] (coal4) -- (gb7);
    \node (rename4) [vertex, above=of gb7] {$\rho_{sum(s)\mapsto s}$};
        \draw[edge] (gb7) -- (rename4);
    \node (dunion5) [vertex, right=of rename1] {$\uplus_{z}$};
        \draw[edge] (rename1) -- (dunion5);
        \draw[edge] (rename4) -- (dunion5);
    \node (dunion2) [vertex, below=of dunion1] {$\uplus_{z\prime}$};
        \draw[edge] (dunion1) -- (dunion2);
        \draw[edge] (dunion5) -- (dunion2);
    \node (coal2) [vertex, right=of dunion2] {$\kappa_{z, z\prime}$};
        \draw[edge] (dunion2) -- (coal2);

		\end{tikzpicture}
		\captionof{figure}{COVID data alignment ($\kappa_z(R \uplus_z S)$ is the implementation of fusion operation $R \sqcup S$ in Definition~\ref{definition.fusion})} \label{fig:COVIDquery}
	\end{minipage}
\end{figure*}

Fig.~\ref{fig:COVIDquery} shows a logical representation of our alignment of the sources.  Notational elements are introduced to facilitate the understanding, like ``avg'' instead of the ``sum/count'' actually used in the current prototype.
Dimensional tables like $\textit{date}$  and $\textit{firstadminunit}$ and their corresponding joins to facilitate selections over $\textit{year}$ and week of year ($\textit{woy}$), or the relationships between countries and regions, are omitted for the sake of simplicity.
This alignment reflects the knowledge about the behaviour of COVID-19 pandemic, but other alternative alignments could have been easily explored with \sysname. On the first hand,  we take $EU$ and $\mathit{EUe}$ tables and generate the weekly surplus of deaths after the sixth week of 2020 by subtracting from the declared amounts, the average deaths in the last five years for the same week.
This is done both per region and country, since these values are not always concordant (even if coming from the same source).
Then, regional results are aggregated per country and merged in the same table with the information provided already at that level using a discriminated union to keep track of the different origins.
On the other hand, looking now at JHU tables, we aggregate regional data in three different ways: deaths per country and day, also deaths per region and week, and finally cases per region and week with a lag of three weeks (we will empirically justify this concrete value later).
Under the assumption of Case-Fatality Ratio of $1.5\%$ (observed median on June 22nd, 2021 is 1.7\% according to JHU\footnote{\url{https://coronavirus.jhu.edu/data/mortality}}), such transformation is applied to the cases before merging and coalescing the weekly regional cases and deaths.
Daily deaths reported per country and those obtained after aggregating regions are also coalesced and then aggregated per week.
Both branches of JHU data are finally merged with a discriminated union into a single table.
Finally, the four branches (namely EuroStats regional data, EuroStats country data, JHU regional data aligning cases and deaths, and JHU regional data coalesced with JHU country data) are merged into a single table with a discriminated union and finally coalesced to generate the overall set of equations.
\begin{table} \small
    \begin{center}
    \begin{tabular}{|c|r|r|r|r|r|r|r|}\hline
        \multicolumn{1}{|c}{\textbf{Country}} & \multicolumn{1}{|c}{\textbf{\#Sys}} & \multicolumn{1}{|c}{\textbf{\#Eqs}} & \multicolumn{1}{|c}{\textbf{\#Vars}} &  \multicolumn{1}{|c}{\textbf{Gener.}} & \multicolumn{1}{|c|}{\textbf{Solve}} \\ \hline \hline
        DE & 37 & 50 & 247 & 2.77s & 0.24s \\ \hline
        ES & 37 & 54 & 278 & 2.79s & 0.24s \\ \hline
        IT & 37 & 60 & 322 & 2.80s & 0.24s \\ \hline
        NL & 37 & 42 & 187 & 2.73s & 0.24s \\ \hline
        SE & 37 & 59 & 306 & 2.68s & 0.25s \\ \hline
        UK & 30 & 26 & 75 & 2.73s & 0.24s \\
        \hline
    \end{tabular}
    \end{center}
    \caption{Characteristics of the equations per country}\label{table:covidequations}
\end{table}

We restricted our analysis to only the six countries in Table~\ref{table:covidequations}, chosen because of their relevance in the pandemic and availability of regional data in both EuroStats and JHU.
Regarding the time, we only considered until February 2021, to avoid the impact of vaccination.
As previously explained, to avoid the cost of errors is scaled to some extent by the magnitude of the value, we replaced uncertain values $v$ in the ground tables with $v\cdot(1+x)$ (or simply $x$ if $v=0$) where $x$ is an error variable.
For each country and week, our alignment generates a different system of intertwined equations, which is solved minimizing the discord (i.e., sum of squared error variables as proposed in \cite{DBLP:journals/sigkdd/LiGMLSZFH15}).

In the table, we can see for each country, the number of systems of equations with the maximum number of variables\footnote{We ignored 2020W53, because of its exceptional nature (nonexistent for other years).} (i.e., all possible data is available, what happens between weeks 2020W26 and 2021W06, except for UK whose data is only available in EuroStats until 2020W51), as well as the equations and variables per system in those cases. The average time in seconds to generate each system of equations as a Python input to OSQP and solve it are also reported.

\begin{figure}
    \begin{center}
    \includegraphics[width=10cm]{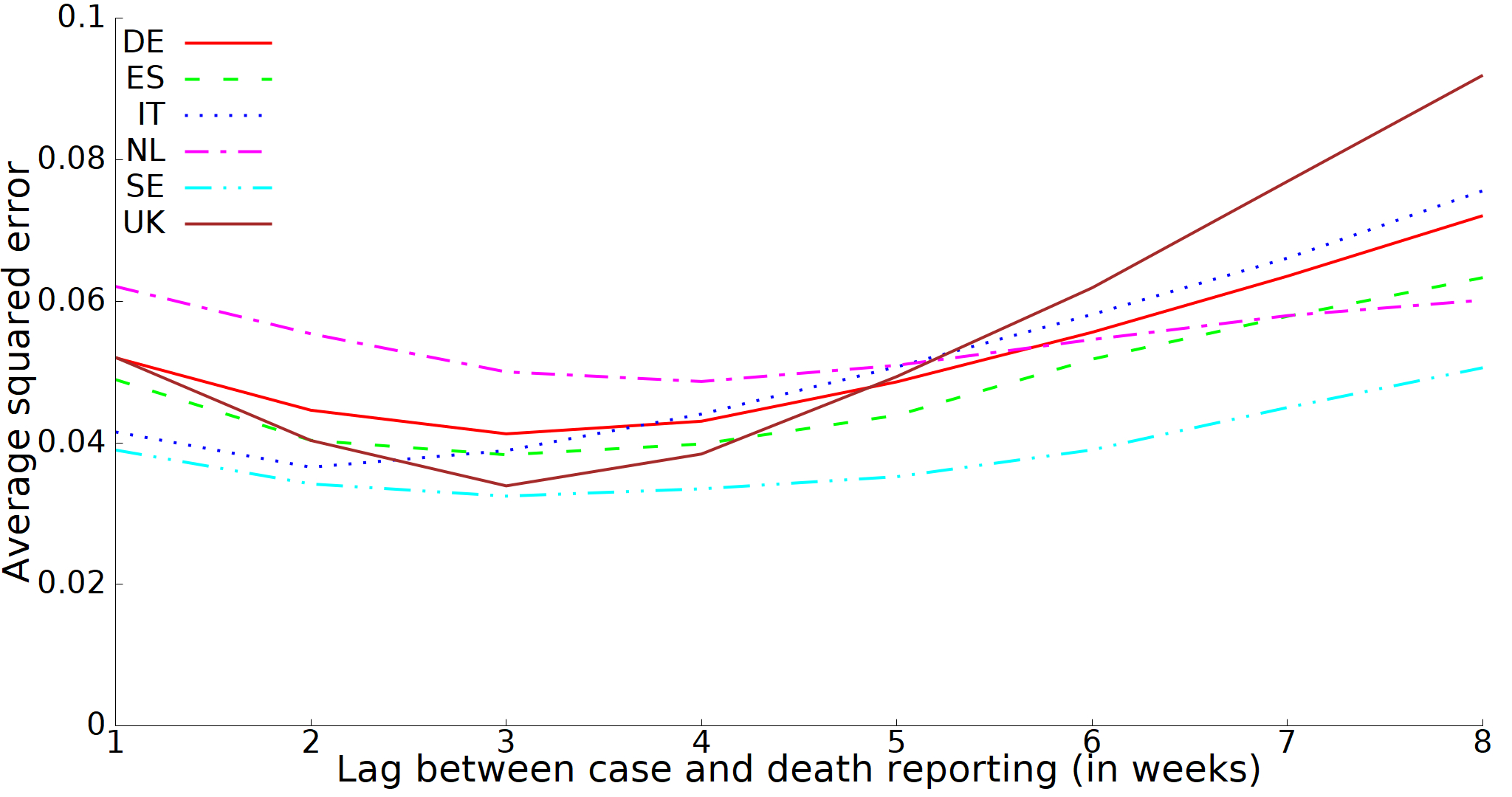}
    \end{center}
    \caption{Error per different alignments of $\#c$ and $\#d$}
    \label{fig:alignments}
\end{figure}

\begin{figure*}[t]
    \centering
    \begin{minipage}[b]{10cm}
        \center
        \includegraphics[width=\textwidth]{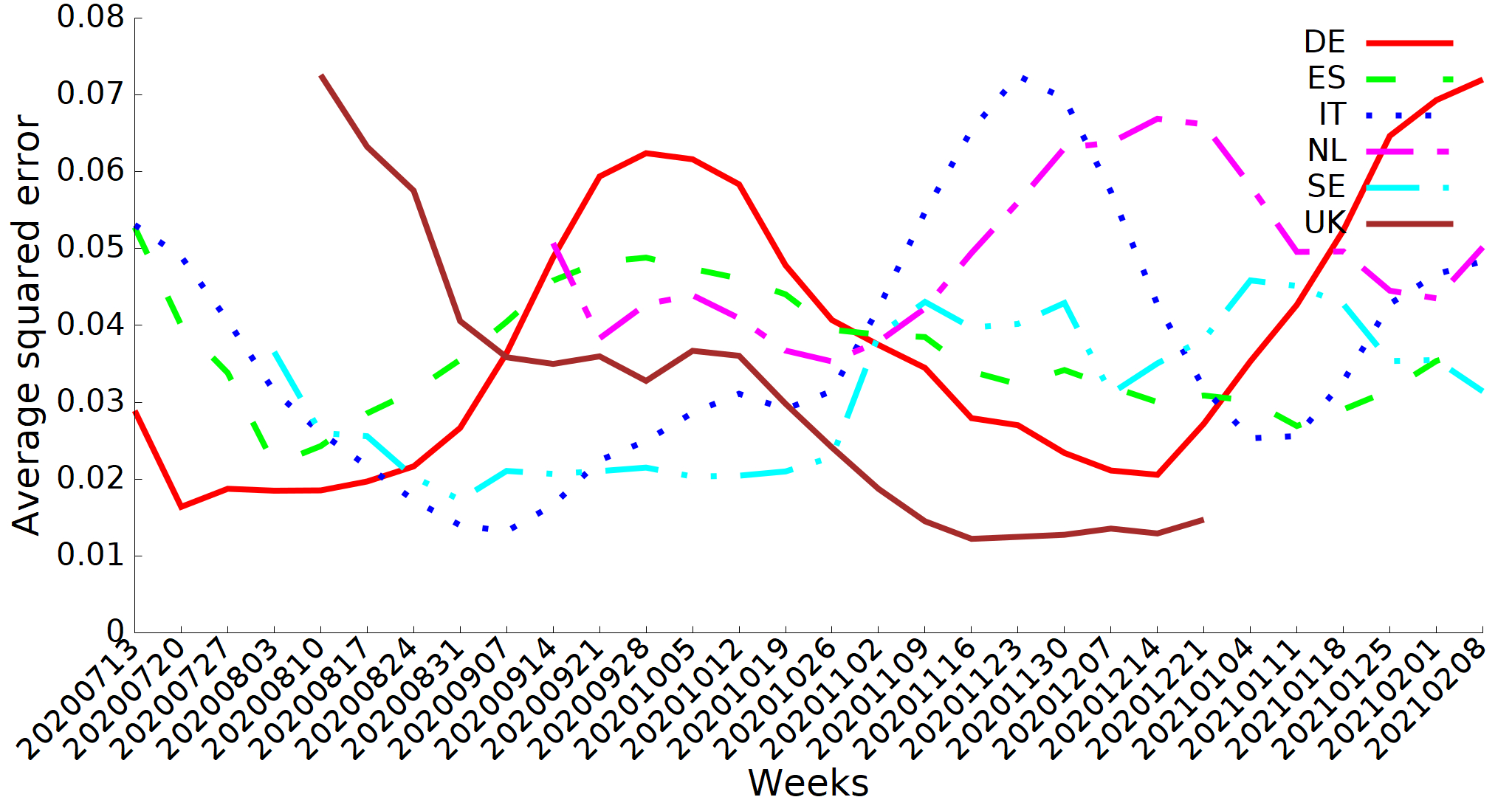}
        \subcaption{Running average of weekly errors considering regional data}
        \label{fig:errorsPerRegion}
    \end{minipage}
    \hspace{0.2cm}
    \begin{minipage}[b]{10cm}
        \center \includegraphics[width=\textwidth]{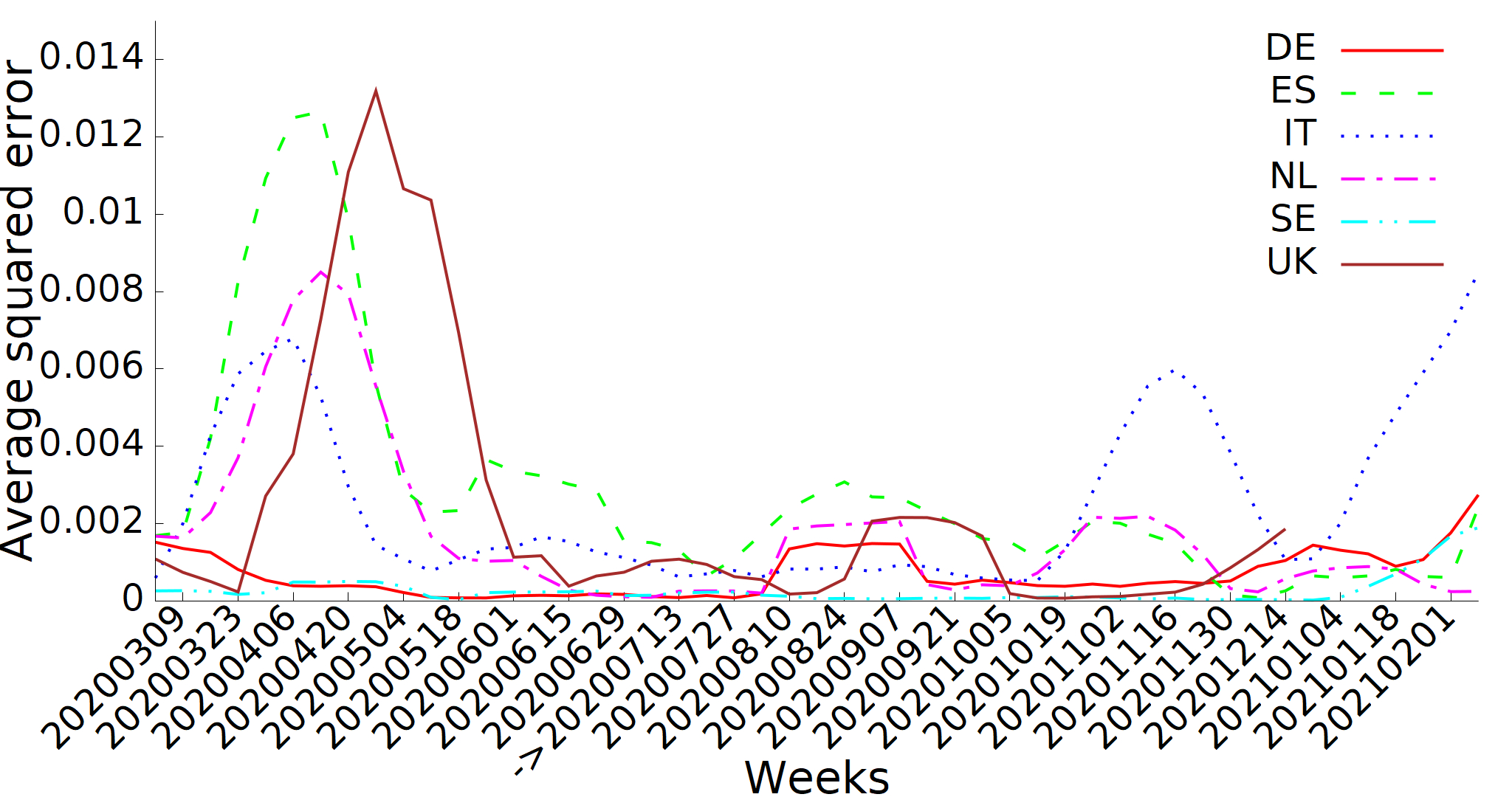}
        \subcaption{Running average of weekly errors not considering regional data}
        \label{fig:errorsPerCountry}
    \end{minipage}
 \caption{Discordancy analysis of COVID-19 data aligning $\#c$ and $\#d$ with a lag of three weeks}
    \label{fig:ConcordancyAnalysis}
\end{figure*}

The line charts in Fig.~\ref{fig:alignments} and~\ref{fig:ConcordancyAnalysis} plot in the vertical axis discordance (i.e., sum of squared errors in our case) divided by the number of variables (i.e., average squared error per variable), to make them more comparable, since depending on the number of regions, different countries generate more or less variables (see Table~\ref{table:covidequations}).
Firstly, Fig.~\ref{fig:alignments} varies the lag between reported cases and deaths, for values from one to eight weeks.
We can see that the average discordance is minimized in all cases between two and four weeks (three in the average).
Thus, in the other charts, we used a lag of three weeks between cases and deaths, which minimizes the average squared error of all six countries.

Fig.~\ref{fig:errorsPerRegion} shows the evolution of the discordance since 2020W26 until the last week being considered. We can appreciate that during the first weeks reporting regional data, countries adjusted and eventually improved their COVID-19 surveillance mechanisms. However, all of them except UK are too sensitive to the increase of cases and their concordancy with real deaths is clearly affected by the arrival of the second wave after summer and the third one at the end of the year (we can clearly appreciate the two peaks in the five other countries).  Unfortunately the UK data reporting to Eurostats stopped on December 31, 2020 due to Brexit, so we cannot see from the Eurostats data whether the UK's error remained low during the rest of the infection wave in early 2021.

Finally, Fig.~\ref{fig:errorsPerCountry} shows the clearer but less computationally challenging evolution of discordancy without considering regionally reported data (the small pointer in the horizontal axis indicates the horizontal coincidence of both charts). We can see a clear peak of discordancy during the first wave that eventually improves, just to be more or less affected again by the second and third waves depending on the country. As a derived calculation of the observed discordancy, we can take the Pearson correlation coefficient between those and the running average of the number of cases (i.e., DE: 36\%; ES: 80\%; IT: 50\%; NL: 73\%; SE:	23\%; UK: 93\%). Thus, we can observe that in the case of ES, NL and UK, more than 70\% of variation in the discordancy can be explained by changes in the number of cases.

Without \sysname, this study had been simply impossible, because manually generating the corresponding SQL and Python code had been too difficult for a human being. Instead, we generated them automatically and with all the correctness guarantees from a relatively simple algebraic sequence of operations.
Alternatively, we could have also somehow easily defined the corresponding static assertions over JHU data (indicating that they must exactly coincide) and count how many times they were violated.
Nevertheless, given the poor quality of the source, all we had gotten is a flat line indicating that all values are simply inconsistent in each and every week with regard to those from EuroStats.
Thus, our novel discord measuring mechanism could be easily further integrated in a more complex truth discovery iterative method to obtain the trustworthiness weight of each source.
We consider this to be a potential direction for future work.

\section{Related work}\label{sec:relatedwork}

Out of the many papers identified in a recent systematic literature review on Information Fusion techniques~\cite{DBLP:journals/inffus/GutierrezRPLS22}, just one~\cite{DBLP:series/isrl/TreBMB13} was found making use of consistency to evaluate the quality of input data. However, it only uses the difference between pairs of numbers as the basis for this evaluation. Moreover, this is not actually done with the purpose of evaluating the overall quality of the sources, but rather the coincidence between instances with matching purpose.

There is existing work such as~\cite{dyreson.incomplete}, \cite{DBLP:conf/icde/BaikousiRV11} and
\cite{DBLP:journals/dss/GolfarelliT14} on measuring differences in the descriptive multidimensional data and their structure.
Instead, we aim at evaluating the reliability of the numerical indicators, given some required alignment declaration (e.g., aggregation or scale correction). At this point, it is important to highlight that, even if some work like \cite{DBLP:journals/jdwm/OukidBBB16} proposes to treat textual data as indicators (allowing to aggregate them too), we restricted ourselves to numerical measures, whose discrepancies cannot be evaluated using string similarity metrics. The latter would instead be part of a preliminary step of entity matching over dimensional descriptors.

Thus, the problems defined above are related to Consistent Query Answering (CQA)~\cite{chomicki.repair}, which tries to identify the subset of a database that fulfills some integrity constraints, and corresponds to the problem of identifying certain answers under open world assumption~\cite{baader.handbook}.
In CQA, distance between two database instances is captured by symmetric difference of tuples. However, in our case, the effects of an alignment are not only reflected in the presence/absence of a tuple, but also in the values it contains.
This leads to the much closer Database Fix Problem (DFP)~\cite{bertossi.complexity,bohannon.cost-based}, which aims at determining the existence of a fix at a bounded distance measuring variations in the numerical values.
Both DFP as well as CQA become undecidable in the presence of aggregation constraints.
Nonetheless, these have been used to drive deduplication \cite{chaudhuri.levaraging}. However, our case is different since we are not questioning correspondences between entities to create aggregation groups, but instead trying to quantify their (in)consistency in the presence of complex transformations. 

Another known result in the area of DFP is that prioritizing the repairs by considering preferences or priorities (like the data sources in our case) just increases complexity.
An already explored idea is the use of where-provenance in the justification of the new value \cite{geerts.llunatic}, but with pure direct value imputation (without any data transformation).
In contrast, we consider that there is not any master data, but multiple contradictory sources, and we allow aggregates, while \cite{geerts.llunatic} only uses equalities (neither aggregation nor any real arithmetic) between master and target DBs.
Related to this, in the area of machine learning, we have \cite{DBLP:journals/pvldb/SchleichGZS21}, which aims at finding counterfactual explanations for a prediction. The purpose in this case is not to directly change the data, but to tell the user what should have been done to get a different prediction. Like in our case, this is treated as an optimization problem.

From the perspective of incompleteness in multidimensional databases, attention is paid to missing values in the measures.~\cite{DBLP:journals/tkde/PalpanasKM05} presents an approach to maximize entropy, and~\cite{bimonte.linearProgramming} a linear program\-ming-based framework to impute missing values  under constraints generated by sibling data at the same aggregation level, and parent data in higher levels.
We could consider the later a special case of our approach, with a single data source and predefined alignment.

In the context of data fusion, although our purpose is not to merge records, but only evaluate how far they are from one another, we can still position our work according to the characteristics in \cite{DBLP:journals/jiis/CanalleSL21} as follows:
\begin{description}
\item[Data types] we consider are continuous (a.k.a. quantitative);
\item[Heterogeneity] of data types is not considered in our work, as we focus on dealing with purely numerical attributes;
\item[Single-truths] (i.e., each attribute has a single value in reality) is assumed;
\item[Source quality] is the focus of our work by measuring their discrepancies;
\item[Copying between sources] is not considered (i.e., instead, we consider they provide their values independently);
\item[Object relationships] can be naturally expressed in the form of symbolic expressions sharing variables;
\item[Object popularity and difficulty] can be considered in the cost functions.
\end{description}

The setting we have described in the context of data fusion and truth discovery shares also many motivations with previous work on provenance.
The semiring provenance model~\cite{green.provenance} is particularly related, explaining why why-provenance~\cite{buneman01icdt} is not enough (e.g., in the case of alternative sources for the same data) and we need how-provenance to really understand how different inputs contribute to the result.
They propose the use of polynomials to capture such kind of provenance.
Further, \cite{amsterdamer.aggregates} extended the semiring provenance model to aggregations by mixing together annotations and values, but the fine-grained provenance information may become  large.
However, to the best of our knowledge no practical implementations of it exist.  In contrast, our approach does not have row-level annotations recording the conditions that make a row present in the result, limits aggregation to value fields, and considers only sum and averaging forms of aggregation, but we have provided practical implementations of this more limited model.
As noted earlier, our s-tables are similar in some respects to c-tables studied in incomplete information databases~\cite{imielinski84jacm}. Our data model and queries are more restricted in some ways, due to the restriction to finite maps, and because we do not allow for conditions affecting the presence of entire rows, but our approach supports aggregation, which is critical for our application area and which was not handled in the original work on c-tables.
Similarly, \emph{attribute-level uncertainty bounds} (AU-BDs) in \cite{DBLP:conf/sigmod/FengGHK21} allow to annotate values with intervals (and a selected guess) encoding a set of possible worlds. However, this does not help to find the most likely world (beyond the provided selected guess).

There have been implementations of semiring provenance or c-tables in systems such as Orchestra~\cite{DBLP:journals/sigmod/IvesGKTTTJP08}, ProQL~\cite{DBLP:conf/sigmod/KarvounarakisIT10}, ProvSQL~\cite{senellart.provsql}, and Mimir~\cite{DBLP:journals/corr/NandiYKGFLG16}.  In Orchestra provenance annotations were used for update propagation in a distributed data integration setting.  ProQL and ProvSQL implement the semiring model but do not allow for symbolic expressions in data or support aggregation. Mimir is a system for querying uncertain and probabilistic data based on c-tables; however, in Mimir symbolic expressions and conditions are not actually materialized as results, instead the system fills in their values with guesses in order to make queries executable on standard RDBMSs. Thus, Mimir's approach to c-tables would not suffice for our needs since we need to generate the symbolic constraints for the QP solver to solve.  On the other hand, our work shows how some of the symbolic computation involved in c-tables can be implemented in-database and it would be interesting to consider whether this approach can be extended to support full c-tables in future work.

We have reduced the concordancy evaluation problem to quadratic programming, a well-studied optimization problem. Solvers such as OSQP~\cite{osqp} can handle systems with thousands of equations and variables.  However, we have not made full use of the power of linear/quadratic programming.  For example, we could impose additional inequalities on unknowns, to ensure that certain error or null values have to be positive or within some range.  Likewise, we have defined the cost function in one specific way but quadratic programming permits many other cost functions, like using different weights for each variable or with additional linear cost factors. As suggested at the end of the last section, it may be worthwhile to combine \sysname with other truth discovery techniques to simultaneously estimate the weights needed for the cost function and the guessed true values of the uncertain data.

As noted in Section~\ref{sec:example}, we have focused on the problem of evaluating concord/discord among data sources and not on using the different data sources to estimate the actual values being measured (like \cite{DBLP:journals/inffus/MotroA06}).  It would be interesting to extend our framework by augmenting s-tables and queries with a probabilistic interpretation, so that the optimal solution found by quadratic programming produces statistically meaningful consensus values (similarly to~\cite{mayfield.eracer}).

\section{Conclusions}\label{sec:conclusions}

In most data integration and cleaning scenarios, it is assumed that there is some source of ground truth available (i.e., master data or user input).  However, in many realistic data fusion settings, such as epidemiological surveillance, ground truth is not obtainable and we need to integrate discordant data sources with different levels of trustworthiness, completeness and self-consistency.  In such scenarios, we still would like to be able to flexibly and efficiently measure how close the observed data is to our idealized expectations.  Thus, we proposed definitions of \emph{concordance} and \emph{discordance} capturing respectively when data sources we wish to fuse are compatible with one another, and measuring how far away the observed data are from being concordant.  Consequently, we can compare measurements over time to understand whether the different sources are becoming more or less consistent with one another.  We showed how to flexibly and efficiently solve this problem by extending multidimensional relational queries with symbolic evaluation, and gave two relational implementations of this approach reducing it to linear programming or quadratic programming problems that can be solved by an off-the-shelf library.  We explored the performance of the two approaches via microbenchmarks to assess the scalability in data size and number of variables,  illustrated the value of this information using a case study based on COVID-19 case and death reporting from 2020-2021, and found that the error calculated for six European countries at different times correlates with intuition.

Different cost functions, alternatives in the management of NULL values and zeros, as well as alternatives for variables generation need to be carefully analyzed. However, the most appropriate one will be case-dependant and so we plan to do this separately.
Moreover, our approach to symbolic evaluation of multidimensional queries appears to have further applications which we plan to explore next, such  supporting other forms of uncertainty expressible as linear constraints, and adapting our approach to produce statistically meaningful estimates of the consensus values.

\section*{Acknowledgments}
The work of A. Abelló has been done under project PID2020-117191RB-I00 funded by MCIN/AEI/ 10.13039 /501100011033.
The work of J. Cheney was supported by ERC Consolidator Grant Skye (grant number 682315).

\bibliographystyle{plain}
\bibliography{references}

\appendix
\section{Proofs}\label{appendix.correctness}

We state the desired correctness property for the type system formally as follows:
\begin{theorem}
Suppose $\Sigma \vdash q : K \tri V$ is derivable using the rules in Fig.~\ref{fig:wf}.  Then $q$ denotes a function from $Inst(\Sigma)$ to relations $R : K \tri V$.
\end{theorem}
\begin{proof}
The proof is by induction on the structure of derivations of $\Sigma \vdash q : K \tri V$.  Most cases are standard or straightforward.  The interesting cases are those where constraints are necessary to  preserve the finite map property: for example, projection-away and discriminated union.  We sketch the reasoning in each case.

For projection-away, we may only discard value fields, so the key fields in the result are the same as those in the input relation resulting from evaluating the subquery.  Hence, the finite map property is preserved.

For discriminated union, it is clear from the definition of the semantics that the keys of the result are a tagged disjoint union of the keys of the two input relations, which are both finite maps.  Hence the result is a finite map satisfying the FD $K,B \to V$.

For aggregation, the result may drop both key and value fields, but the value fields will be aggregated (summed) according to whatever keys remain, so the result will be a finite map $K' \tri V'$ by construction.
\end{proof}
Note that some of the constraints on queries are not necessary to ensure well-formed queries produce valid finite maps, but are only needed to ensure that symbolic evaluation is correct on s-tables.  For example, if selection conditions were allowed on value fields (that might be symbolic), then the presence of a tuple with symbolic tested fields in the output would depend on the unknown variable values.  This conditional membership is not supported in s-tables, but was considered in c-tables in which the presence of a tuple in the output can be constrained by a formula.  While this would be an interesting extension, we do not have a pressing need for this capability in OLAP tools whereas it would significantly complicate the formalism and  implementation.

Fig.~\ref{fig:spec-wf} presents additional well-formedness rules for alignment specifications.  The judgment $\Sigma \vdash \Delta : \Omega$ says that in schema $\Sigma$, the definition of views $\Delta$ is well-formed and produces a result matching schema $\Omega$; that is, the new tables defined in $\Delta$ are as specified in $\Omega$.
This well-formedness judgment satisfies the following correctness property:
\begin{theorem}
Suppose $\Spec =[\Sigma,\Omega,\Delta]$ is an alignment specification and $\Sigma \vdash \Delta : \Omega$ holds.  Then we may interpret $\Delta$ as a partial function from instances of $\Sigma$ to instances of $\Omega$.
\end{theorem}

The proof is straightforward; the interpretation of $\Delta$ attempts to construct the instance of $\Omega$ one relation at a time, using the (partial) fusion operation in each step.  Fusion is associative and commutative, so the result is well-defined.

\subsection{Linearity}

In order to ensure that constraints generated by symbolic evaluation and fusion/coalescing are valid linear programming problems, we need to restrict the s-tables/s-instances to include linear expressions only and also restrict queries so that derivation steps only involve linear combinations of attributes.  Subject to these restrictions, we can verify that for well-formed queries and alignment specifications, the resulting s-tables, s-instances, and constraints are linear as well.
\begin{theorem}\label{thm:linear-query}
Suppose $q$ is well-formed, satisfying $\Sigma \vdash q : K \tri V$, and all occurrences of derivations in $q$ are linear expressions. Suppose in addition $I$ is a well-formed linear s-instance.  Then, $q(I)$ is a well-formed linear s-table $q(I) : K \tri V$.
\end{theorem}
\begin{proof}
The proof is by induction on the structure/well-formed\-ness derivation of $q$.
Many cases, e.g. the base case $q=R$, selection, projection, join,  renaming etc. are straightforward. The case of derivation follows because derivations are required to be linear expressions over fields.  The case of aggregation follows because the only kind of aggregation allowed is SUM, and adding together any number of linear expressions is still a linear expression.
\end{proof}

\begin{theorem}
Suppose $\Spec = [\Sigma,\Omega,\Delta]$ is well-formed, satisfying $\Sigma \vdash \Delta : \Omega$, and all occurrences of derivations in $\Delta$ are linear expressions. Suppose in addition $I$ is a well-formed linear s-instance matching $\Sigma$.  Then, the result of evaluating $\Delta$ on $I$, a constrained s-instance $(J,\phi)$, is linear, that is, $J$ is a linear s-instance matching $\Omega$, and $\phi$ is a conjunction of linear constraints.
\end{theorem}
\begin{proof}
The proof is by induction on the structure/well-formed\-ness derivation of $\Sigma \vdash \Delta : \Omega$.  The base case is immediate.  For the inductive step where $\Delta$ consists of a binding $T := q_1 \sqcup \cdots \sqcup q_n$ followed by another sequence $\Delta'$, note that in this case $\Omega$ must be of the form $T : K \tri V,\Omega'$ and the well-formedness relations $\Sigma \vdash q_i : K \tri V$ (for each $i$) and $\Sigma,T : K \tri V \vdash \Delta' : \Omega'$ must hold.
Using Theorem~\ref{thm:linear-query}, since each $q_i$ is well-formed in $\Sigma$ satisfying $\Sigma \vdash q_i : K \tri V$, we know that each $q_i$ preserves linearity, so $q_i(I)$ is linear for each $i$.  Moreover, the fusion of all of the $q_i$'s can be expressed as an $n$-way coalescing of $q_i(I)$ and we can inspect the definition of coalescing to check that its result $(T',\phi)$ is a linear s-instance and a conjunction of linear constraints.  Since $\Sigma,T : K \tri V \vdash \Delta' : \Omega'$ holds, we can apply  the induction hypothesis using the specification $[(\Sigma,T:K \tri V), \Delta',\Omega']$ and to $I$ extended with $T = T'$, since $I$ and $T'$ are both linear.
Thus, we can conclude that the final s-instance and constraint $(J,\phi')$ obtained are linear also.
\end{proof}

We assume from now on that the s-tables are linear and the queries only involve linear expressions in derivation steps.

\subsection{Correctness of symbolic evaluation}

We require that symbolic evaluation correctly abstracts ground evaluation, in the sense that evaluating symbolically and then filling in ground values yields the same results as evaluating on fully ground input tables.  We also expect that, as s-tables represent sets of ground tables, the symbolic evaluation of query operations over tables correctly reflects the possible sets of ground tables resulting from the query operation.

These properties are closely related, and similar to the standard correctness properties used for incomplete information representations such as c-tables~\cite{imielinski84jacm}.  However, there is an important difference: in the classical setting, the variables representing unknown values are ``scoped'' at the level of tables.  That is, if table $R$ and $S$ both mention some variable $x$, the occurrences in $R$ and respectively $S$ are \emph{local} to the respective table, and the value of $x$ in $R$ may not have anything to do with that in $S$.  In our case, however, we wish to reason about situations where unknown values propagate from source tables in $I$ through view definitions in $J$, and we definitely do \emph{not} want the variables appearing in different tables to be unrelated; instead, we want the variables to have \emph{global} scope.

To prove the main correctness property, we first need a lemma concerning the behavior of the individual operators.
\begin{lemma}\label{lem:commutes}
  Each s-table operation commutes with valuations:
  \begin{enumerate}
  \item $\sigma_c(h(R)) = h(\sigma_c(R))$
\item $\hat{\pi}_W(h(R)) = h(\hat{\pi}_W(R))$
\item $h(R) \Join h(S) = h(R\Join S)$
\item $h(R) \uplus_B h(S) = h(R \uplus_B S)$
\item $h(R) \backslash h(S)  = h(R\backslash S)$
\item $\rho_{B \mapsto B'}(h(R)) = h(\rho_{B \mapsto B'}(R))$
\item $\varepsilon_{B := e}(h(R)) = h(\varepsilon_{B := e}(R))$
\item $\gamma_{K';V'}(h(R)) = h(\gamma_{K';V'}(R))$
  \end{enumerate}
\end{lemma}
\begin{proof}
We consider selected cases; the rest are straightforward.

For part (1), we need to show that the result of a selection applied to a grounded symbolic table $h(R)$ is the same as performing the selection symbolically and then applying the grounding valuation.  This is the case because the selection condition cannot mention value fields, and so the decision whether to select a tuple cannot depend on symbolic fields that might be affected by $h$.

For part (2), again since projection-away can only affect value fields, the key fields are unaffected so performing the projection-away on the grounded table $h(R)$ is the same as grounding the projected-away symbolic table.

For part (3) the reasoning is similar to that for selection, since joins can only involve common key fields and not value fields.

For part (4) again since the discriminant field $B$ is a key field which cannot be affected by $h$, the result is immediate.

Parts (5) and (6) are likewise straightforward.

For part (7), the expression $e$ used in the derivation will not mention any variables in $X$ (but only attributes of $R$), so $h$ will commute with $e$, and the desired result follows by calculation.

Finally, for part (8) we again note that since key fields may not be symbolic, the keys of the result of aggregating the grounded $h(R)$ are the same as for aggregating the symbolic table $R$, and by unfolding definitions we can check that the results obtained are the same.
\end{proof}
We now prove a slight generalization of Theorem~\ref{thm:commutation}:
\begin{theorem}
For any well-formed query $q$ such that $\Sigma \vdash q : K \tri V$, any s-instance $I\in Inst(\Sigma)$ with variables $X$ and any valuation $h : \mathbb{R}^X$, we have $h(q(I)) =q(h(I))$.  Moreover, $q\banana{I} = \banana{q(I)}$.
\end{theorem}

\begin{proof}
The second part (corresponding to Theorem~\ref{thm:commutation}) follows from the first.
The proof of the first part is by induction on the structure/ well-formedness judgement of $q$.  Each case (except for the base case of a relation name) corresponds to part of Lemma~\ref{lem:commutes}.  For the second part, we reason as follows:
\[\begin{array}{rclr}
q\banana{I} &=& \{q(J) \mid J \in \banana{I}\} = \{q(h(I)) \mid h : \mathbb{R}^X\} = \\
& & \{h(q(I)) \mid h : \mathbb{R}^X\} = \banana{q(I)} & \hphantom{00000000}
\end{array}
\]
\end{proof}
This is similar to the standard results about c-tables showing that they  form a strong representation system for relational queries over incomplete databases~\cite{imielinski84jacm}.  The main difference is that, as previously said, in our setting, the variables mapped by $h$ are globally scoped.  This means that to correctly simulate operations that take multiple tables, such as joins and unions, we do not need to rename the variables to avoid unintended overlap.  In fact, this would be incorrect: suppose $I(R) = \{(a:1, b:x)\}$ and $I(S) = \{(a:1,c:x)\}$.  Then using our semantics, the join $R \Join S$ evaluated in $I$ is $\{(a=1,b=x,c=x)\}$ which represents all tuples where $a= 1$ and the $b$ and $c$ fields are equal real numbers.  In contrast, using c-table semantics, the variables in $R$ and $S$ would be renamed so the join result would be some variant of $\{(a:1,b:x',c:y')\}$ which represents all tuples whose $a$ component is 1, with no constraint relating $b$ and $c$.

\end{document}